\documentclass[11pt]{article}
\usepackage{fullpage}

\usepackage{thmtools} 
\usepackage{thm-restate}

\usepackage{dsfont}
\usepackage{amsmath}
\usepackage{amsthm}
\usepackage{amssymb}
\usepackage{graphicx}
\usepackage{url}
\usepackage{nicefrac}
\usepackage{enumitem}
\usepackage{authblk}
\usepackage{tikz} 
\usepackage{subcaption}
\usepackage{xspace}
\usepackage[compact]{titlesec}
\usetikzlibrary{arrows,decorations.markings}
\usepackage[normalem]{ulem}

\usepackage{xcolor}
\usepackage[linesnumbered,ruled,vlined]{algorithm2e}

\SetCommentSty{mycommfont}

\SetKwInput{KwInput}{Input}                
\SetKwInput{KwOutput}{Output}              

\usepackage{typearea}

\usepackage{pgf}
\usetikzlibrary{snakes}
\usetikzlibrary{patterns}
\usetikzlibrary{backgrounds}
\usetikzlibrary{scopes}
\usetikzlibrary{arrows, automata, positioning}
\usetikzlibrary{decorations.markings}
\usepackage[T1]{fontenc}
\usepackage{calc}
\usepackage{color}
\usepackage{lipsum}

\makeatother
\usepackage{color}

\definecolor{blueblack}{rgb}{0,0,.7}

\newcounter{sideremark}

\definecolor{Darkblue}{rgb}{0,0,0.4}
\definecolor{Brown}{cmyk}{0,0.61,1.,0.60}
\definecolor{Purple}{cmyk}{0.45,0.86,0,0}
\definecolor{brickred}{rgb}{0.8, 0.25, 0.33}
\usepackage[colorlinks,citecolor=blue,linkcolor=red,bookmarks=true]{hyperref}

\theoremstyle{plain}
\newtheorem{theorem}{Theorem}
\newtheorem{invariant}{Invariant}
\newtheorem{lemma}{Lemma}
\newtheorem{claim}{Claim}
\newtheorem{definition}{Definition}

\newtheorem{fact}{Fact}
\newtheorem{remark}{Remark}
\newtheorem{corollary}{Corollary}

\newcommand{\eps}{\varepsilon}

\newcommand{\calC}{\mathcal{C}}

\newcommand{\ALG}{\mathsf{ALG}}

\newcommand{\RED}{\mathsf{RED}}
\newcommand{\OPT}{\mathsf{OPT}}
\newcommand{\opt}{\mathsf{opt}}
\newcommand{\Red}{red}
\newcommand{\APX}{\mathsf{APX}}
\newcommand{\apx}{\mathsf{apx}}

\newcommand{\credit}{\mathsf{cr}}

\usepackage[colorinlistoftodos,textsize=tiny,textwidth=2cm,color=red!25!white,obeyFinal]{todonotes}

\ifdefined\DEBUG
 \newcommand{\fab}[1]{\textcolor{black}{#1}}

\newcommand{\amEdit}[1]{\textcolor{violet}{#1}}

  \def\rem#1{{\marginpar{\raggedright\scriptsize #1}}}
  \newcommand{\fabr}[1]{\rem{\textcolor{red}{$\bullet$ #1}}}

\else
  \newcommand{\fab}[1]{#1}
  \newcommand{\fabr}[1]{}
  
  \newcommand{\amEdit}[1]{#1}

\fi

\title{Improved Approximation for Two-Edge-Connectivity\thanks{Partially supported by the SNSF Excellence Grant 200020B\_182865 / 1 and the SNSF Grant 200021\_200731 / 1.}}

\author[1]{Mohit Garg}
\author[2]{Fabrizio Grandoni}
\author[3]{Afrouz Jabal Ameli}
\affil[1]{University of Bremen, University of Hamburg, Germany, garg@uni-bremen.de}
\affil[2]{IDSIA, USI-SUPSI, Switzerland, fabrizio@idsia.ch}
\affil[3] {TU Eindhoven, Netherlands, a.jabal.ameli@tue.nl}
\date{}
\begin{document}
\maketitle
\thispagestyle{empty}





\begin{abstract}
\noindent The basic goal of survivable network design is to construct low-cost networks which preserve a sufficient level of connectivity despite the failure or removal of a few nodes or edges. One of the most basic problems in this area is the $2$-Edge-Connected Spanning Subgraph problem (2-ECSS): given an undirected graph $G$, find a $2$-edge-connected spanning subgraph $H$ of $G$ with the minimum number of edges (in particular, $H$ remains connected after the removal of one arbitrary edge).

2-ECSS is NP-hard and the best-known (polynomial-time) approximation factor for this problem is $4/3$. Interestingly, this factor was achieved with drastically different techniques by [Hunkenschr{\"o}der, Vempala and Vetta '00,'19] and [Seb{\"o} and Vygen, '14]. In this paper we present an improved $\frac{118}{89}+\eps<1.326$ approximation for 2-ECSS. 

The key ingredient in our approach (which might also be helpful in future work) is a reduction to a special type of structured graphs: our reduction preserves approximation factors up to $6/5$. While reducing to 2-vertex-connected graphs is trivial (and heavily used in prior work), our structured graphs are ``almost'' 3-vertex-connected: more precisely, given any 2-vertex-cut $\{u,v\}$ of a structured graph $G=(V,E)$, $G[V\setminus \{u,v\}]$ has exactly 2 connected components, one of which contains exactly one node of degree $2$ in $G$.
\end{abstract}
\newpage

\clearpage
\pagenumbering{arabic}



\section{Introduction}

\amEdit{Real-world} networks are prone to failures. For this reason it is important to design them so that they are still able to support a given traffic despite a few (typically temporary) failures of nodes or edges. The basic goal of survivable network design is to construct cheap networks which are resilient to such failures. Most problems in this family are NP-hard, therefore it makes sense to develop approximation algorithms to address them. In this paper we study one of the most fundamental survivable network design problem\amEdit{s}, the 2-edge-connected\footnote{A graph is $k$-edge-connected ($k$EC) if it is connected after the removal of up to $k-1$ edges. A $k$-vertex-connected (kVC) graph is defined similarly w.r.t. the removal of $k-1$ nodes.} (2EC) spanning su\amEdit{b}graph problem (2-ECSS): given an $n$-node undirected 2EC graph $G=(V,E)$, find a minimum-cardinality subset of edges $S\subseteq E$ such that $G'=(V,S)$ is 2EC\footnote{$G'$ is a spanning subgraph, meaning that it connects all the nodes.}.

 2-ECSS is known to be \amEdit{MAX-}SNP-hard \cite{CL99,F98}, in particular it does not admit a PTAS assuming $P\neq NP$. It is easy to compute a $2$ approximation for this problem. For example it is sufficient to compute a DFS tree and augment it greedily. Khuller and Vishkin \cite{KV92} found the first non-trivial $3/2$-approximation algorithm. Cheriyan, Seb{\"{o}} and Szigeti \cite{CSS01} improved the approximation factor to $17/12$. The current best approximation factor for 2-ECSS is $4/3$, and this is achieved with two drastically different techniques by Hunkenschr{\"o}der, Vempala and Vetta \cite{HVV19,VV00}\footnote{\cite{VV00} contains some bugs which were later fixed in \cite{HVV19}.} and Seb{\"o} and Vygen \cite{SV14} (we will tell more about these approaches later on).

Several failed or incomplete attempts were made to improve the $4/3$-approximation factor for 2-ECSS (see Section \ref{sec:failedAttempts} for more details). This makes $4/3$ a natural approximation barrier to breach. In this paper we achieve this result (though with a small improvement).   
\begin{theorem}\label{thr:refinedApx}
For any constant $\eps>0$, there is a deterministic polynomial-time $(\frac{118}{89}+\eps)$-approximation algorithm for 2-ECSS.
\end{theorem}

An overview of our approach will be discussed in Section \ref{sec:overview}.

\subsection{Related Work}

The k-edge connected spanning subgraph problem (k-ECSS) is the natural generalization of 2-ECSS to any connectivity $k\geq 2$ (see, e.g., \cite{CT00,GG12}).

A problem related to k-ECSS is the $k$-Connectivity Augmentation problem ($k$-CAP): given a $k$-edge-connected undirected graph $G$ and a collection of extra edges $L$ (\emph{links}), find a minimum cardinality subset of links $L'$ whose addition to $G$ makes it $(k+1)$-edge-connected. It is known~\cite{DKL76} that $k$-CAP can be reduced to the case $k=1$, a.k.a. the Tree Augmentation problem (TAP), for odd $k$ and to the case $k=2$, a.k.a. the Cactus Augmentation problem (CacAP), for even $k$. Several approximation algorithms better than $2$ are known for TAP \cite{A17,CG18,CG18a,EFKN09,FGKS18,GKZ18,KN16,KN16b,N03}, culminating with the current best $1.393$ approximation \cite{CTZ21}. Till recently no better than $2$ approximation was known for CacAP (excluding the special case where the cactus is a single cycle \cite{GGJS19}): the first such algorithm was described in \cite{BGJ20}, and later improved to $1.393$ by \cite{CTZ21}. 
A related connectivity augmentation problem is the Matching Augmentation problem (MAP). Some of the techniques used in this paper (in particular the credit scheme in the analysis) are closely related to recent work on MAP \cite{BDS22,CCDZ20,CDGKN20}.

For all the mentioned problems one can define a natural weighted version. Here a general result by Jain \cite{J01} gives a $2$ approximation, and in most cases this is the best known. The most remarkable exception is a very recent breakthrough $1.694$ approximation for the Weighted version of TAP in \cite{TZ21} (later improved to $1.5+\eps$ in \cite{TZ22}). Partial results in this direction where achieved earlier in \cite{A17,CN13,FGKS18,GKZ18,N17}. Achieving a better than $2$ approximation for the weighted version of 2-ECSS is a major open problem closely related to this paper. The $2$-approximation barrier was very recently breached \cite{GJT22} for the special case of weighted 2-ECSS where the \amEdit{edge} weights are $0$ or $1$, also known as the Forest Augmentation Problem (FAP).

\section{Overview of our Approach}
\label{sec:overview}

Our basic approach is similar in spirit to the one of \cite{HVV19}. The high-level idea in their approach is to perform some approximation preserving reductions to obtain an input graph $G$ which is conveniently structured. Then the authors compute an optimal (i.e. minimum-size) 2-edge-cover\footnote{A 2-edge-cover $F$ of $G$ is a subset of edges such that each node has at least $2$ edges of $F$ incident to it. A minimum-size 2-edge-cover can be computed in polynomial time (see e.g., \cite{S03}).}. $H$ of $G$, and then gradually transform $H$ into a feasible solution $S$ by adding and removing edges. This is done in such a way that the total number of edges of the initial $H$ grows at most by a factor \amEdit{of }$4/3$. The claimed approximation ratio then derives from the fact that $|H|$ is a lower bound on the optimal solution size $\opt$ of a 2EC spanning subgraph (since any such solution must be a 2-edge-cover). Their analysis might also be rephrased in terms of a credit-based charging scheme analogously to \cite{CCDZ20,CDGKN20}: one can assign $\alpha=\frac{1}{3}$ credits to each edge of the initial $H$, and then use these credits to pay for the increase of the size of $H$.


We improve on the above approach mostly by exploiting more powerful reduction steps.
This will allow us to use a smaller (average) value of $\alpha$ in the later stages, hence leading to a better approximation factor. Our construction can be also combined with the ear-decomposition based construction in \cite{SV14}, leading to a better approximation factor. We remark that we do not see a way to directly improve on $4/3$ by combining the ideas in \cite{HVV19} and \cite{SV14} (namely, without passing through our structural results). We next provide some more details. 
%
%
%

\subsection{A Reduction to Structured Graphs}

Let $\eps>0$ be some constant which will appear in the approximation factor. W.l.o.g. we will always assume that the number of nodes of the input graph $G$ is $\Omega(1/\eps)$, otherwise we can solve 2-ECSS optimally by brute force in constant time. Analogously to \cite{HVV19} (and some prior work), we can reduce to the case that the input graph $G$ is simple\footnote{Namely, there are no loops nor parallel edges} and 2VC\footnote{A $k$-vertex cut is a subset of $k$ nodes whose removal from $G$ creates $2$ or more connected components. A graph is $k$-vertex-connected (kVC) is there is no vertex cut of size at most $k-1$.}.

The authors of \cite{HVV19} define the notion of contractible cycles. We conveniently generalize this notion as follows (see Figure \ref{fig:ForbiddenStructures}).  
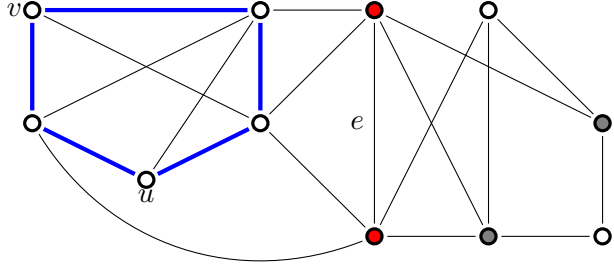
\begin{figure}
\begin{center}
\begin{tikzpicture}[scale=1.5]

\tikzset{vertex/.style={draw=black, very thick, circle,minimum size=0pt, inner sep=2pt, outer sep=2pt}
}

\begin{scope}[every node/.style={vertex}]

\node[fill=red] (1) at (0,0){};
\node[fill=red] (2) at (0,2){};

\node (3) at (-2,0.5){};
\node (4) at (-3,1){};
\node (5) at (-1,1){};
\node (6) at (-3,2){};
\node (7) at (-1,2){};

\node[fill=gray] (8) at (1,0){};
\node (9) at (2,0) {};
\node[fill=gray] (10) at (2,1){};
\node (11) at (1,2){};

\draw[ultra thick, blue] (3) to (4) to (6) to (7) to (5) to (3);
\draw (3) to (7) to (4);
\draw (5) to (2);
\draw (5) to (1);
\draw (7) to (2);
\draw[bend right=40] (4) to (1);
\draw (6) to (5);
\draw (2) to (1);

\draw (10) to (9) to (8) to (1) to (11) to (10) to (2) to (8) to (11);
\end{scope}

\node[below] () at (3){$u$};
\node[left] () at (6){$v$};
\node[left] () at (0,1){$e$};
\begin{scope}[very thick]
\end{scope}

\begin{scope}[very thick, densely dashed]
\end{scope}

\end{tikzpicture}
\end{center}
\caption{The $5$-cycle $C$ induced by blue edges is a $\frac{5}{4}$-contractible subgraph: in particular, every 2EC spanning subgraph must contain at least 4 distinct edges with endpoints in $V(C)$. The red (resp., gray) nodes form a non-isolating (resp. isolating) 2-vertex-cut. The edge $e$ is irrelevant.
\label{fig:ForbiddenStructures}
}
\end{figure}

\begin{definition}[$\alpha$-contractible subgraph]\label{structuredGraphDef:contractible}
Let $\alpha \geq 1$ be a fixed constant. A 2EC subgraph $C$ of a 2EC graph $G$ is \emph{$\alpha$-contractible} if every 2EC spanning subgraph of $G$ contains at least $\frac 1 \alpha |E(C)|$ edges with \fab{both} endpoints in $V(C)$.\end{definition}
Intuitively, an $\alpha$-contractible subgraph $C$ is a 2EC subgraph that we can safely contract with the goal of computing an $\alpha$ approximation for 2-ECSS. The first extra property that we enforce is that $G$ does not contain $\alpha$-contractible subgraphs of size at most $1/\eps$. 

The authors of \cite{HVV19} also show that it is possible to delete any edge $uv$ incident to the two neighbours of a node $w$ of degree two (\emph{beta node}) provided that $G[V\setminus \{u,v\}]$ has at least $3$ connected components (including $w$ itself). We generalize their approach by introducing the notion of irrelevant edges and non-isolating $2$-vertex-cuts (see Figure \ref{fig:ForbiddenStructures}). 
\begin{definition}[irrelevant edge]
Given a graph $G$ and $e=uv\in E(G)$, we say that $e$ is \emph{irrelevant} if $\{u,v\}$ is a $2$-vertex-cut of $G$.
\end{definition}
We show that irrelevant edges can be deleted w.l.o.g. Hence in a structured instance we can assume that there are no such edges. 
\begin{lemma}\label{lem:irrelevant} 
Let $e=uv$ be an irrelevant edge of a 2VC simple graph $G=(V,E)$. Then there exists a minimum-size 2EC spanning subgraph of $G$ not containing $e$.
\end{lemma}
\begin{proof}
Assume by contradiction that all optimal 2EC spanning subgraphs of $G$ contain $e$, and let $\OPT$ be one such solution. Define $\OPT'=\OPT\setminus \{e\}$. Clearly $\OPT'$ cannot be 2EC since this would contradict the optimality of $\OPT$. Let $(V_1,V_2)$ be any partition of $V\setminus \{u,v\}$, $V_1\neq \emptyset\neq V_2$, such that there are no edges between $V_1$ and $V_2$. Notice that at least one between $\OPT'_1:=\OPT'[V_1\cup \{u,v\}]$ and $\OPT'_2:=\OPT'[V_2\cup \{u,v\}]$, say $\OPT'_2$, needs to be connected otherwise $\OPT$ would not be 2EC. We also have that $\OPT'_1$ is disconnected otherwise $\OPT'$ would be 2EC. More precisely, $\OPT'_1$ consists of precisely two connected components $\OPT'_1(u)$ and $\OPT'_1(v)$ containing $u$ and $v$, resp.
Assume w.l.o.g. that $|V(\OPT'_1(u))|\geq 2$. Observe that there must exist an edge $f$ between $V(\OPT'_1(u))$ and $V(\OPT'_1(v))$: indeed otherwise $u$ would be a 1-vertex-cut separating $V(\OPT'_1(u))\setminus \{u\}$ from $V\setminus (V(\OPT'_1(u))\cup \{u\})$. Thus $\OPT'':=\OPT'\cup \{f\}$ is an optimal 2EC spanning subgraph of $G$ not containing $e$, a contradiction.  
\end{proof}

\begin{definition}[isolating $2$-vertex-cut]
Given a graph $G$, and a $2$-vertex-cut $\{u,v\} \subseteq V(G)$ of $G$, we say that $\{u,v\}$ is \emph{isolating} if $G\setminus \{u,v\}$ has exactly two connected components, one of which is of size $1$. Otherwise it is \emph{non-isolating}.
\end{definition}
\begin{figure}
\begin{center}
\begin{tikzpicture}[scale=1.5]

\tikzset{vertex/.style={draw=black, very thick, circle,minimum size=0pt, inner sep=2pt, outer sep=2pt}
}


\node[right] () at (4.5,1) {$x$};
\node[above] () at (4,1.5) {$y$};
\node[above] () at (3,1.5) {$v$};
\node[below] () at (4,0.5) {$z$};
\node[below] () at (3,0) {$u$};
\begin{scope}[every node/.style={vertex}]

\node (1) at (1,0) {};
\node (2) at (2,0.75) {};
\node (3) at (2,0) {};

\node (6) at (4.5,1) {};
\node (7) at (4,1.5) {};
\node (u) at (3,0) {};
\node (v) at (3,1.5) {};
\node (5) at (1,1.5) {};
\node (8) at (4,0.5) {};



\draw (u) to (2);
\draw (v) to (5);
\draw (v) to (3);
\draw (u) to (3);
\draw (5) to (1);
\draw (1) to (2);
\draw (1) to (3);
\draw (2) to (3);
\draw[blue] (u) to (8);
\draw[red] (6) to (7);
\draw[red] (8) to (6);
\draw[red] (7) to (v);
\draw[blue] (8) to (7);
\draw[green] (8) to (v);

\end{scope}

\end{tikzpicture}
\end{center}
\caption{
The black edges, red and blue \amEdit{edges}, and red and green edges induce a subgraph of type A, B\amEdit{,} and C, resp., w.r.t. the 2-vertex-cut $\{u,v\}$.}
\label{fig:TypeABC}
\end{figure}
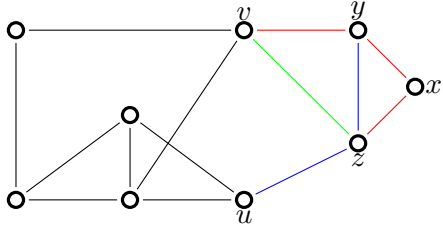
We show how to get rid of non-isolating $2$-vertex-cuts and hence can assume that a structured instance does not contain such cuts. The construction here is slightly complex (details in Section \ref{sec:preprocessing}). The intuition is as follows: consider a non-isolating 2-vertex-cut $\{u,v\}$. W.l.o.g. we will assume $uv\notin E(G)$ since we can remove irrelevant edges. Let $(V_1,V_2)$ be a partition of $V\setminus \{u,v\}$ with $|V_2|\geq |V_1|\geq 2$ and such that there are no edges between $V_1$ and $V_2$. If $|V_1|\geq \Omega(1/\eps)$, we contract $\{u,v\}$ into a single node: this leads to two independent problems that we can solve recursively. Then by adding up to $2$ edges one can obtain a solution for the input problem. The cost of the extra two edges can be conveniently amortized over the cost of the optimal solution restricted to $V_1\cup \{u,v\}$. Otherwise, consider the restriction $\OPT_i$ of the optimal solution $\OPT$ w.r.t. $G[V_i\cup \{u,v\}]$. Let $C(u)$ and $C(v)$ be the 2EC components of $\OPT_i$ containing $u$ and $v$, resp. We observe that $\OPT_i$ must be of one of the following three types (see Figure \ref{fig:TypeABC}) after collapsing each 2EC component $C$ into a single super-node $C$ (see also Definition \ref{def:ABC} and Lemma \ref{lem:ABC}): (A) one super-node $C(u)=C(v)$ (i.e. $\OPT_i$ is 2EC); (B) a path of length at least 1 between $C(u)$ and $C(v)$; (C) two isolated super-nodes $C(u)$ and $C(v)$. Observe also that if $\OPT_1$ is of type C then $\OPT_2$ must be of type A and vice versa. Since $V_1$ has constant size, we can compute the cheapest subgraphs of type A, B, and C in $G[V_1\cup \{u,v\}]$. Depending on the relative sizes of these 3 subgraphs and based on a careful case analysis (critically exploiting the fact that there are no $\alpha$-contractible subgraphs of small size), we recurse on different subproblems over the nodes $V_2\cup \{u,v\}$.

Altogether, we obtain the following definition and lemma (proof in Section \ref{sec:preprocessing}). 
\begin{restatable}[$(\alpha,\eps)$-structured graph]{definition}{defstructuredGraph}\label{structuredGraphDef:structured}
Given $\alpha \geq 1$ and $\epsilon > 0$, an undirected graph $G$ is $(\alpha,\epsilon)$-structured if it is simple and 2VC, it contains at least $2/\eps$ nodes and it does not contain:
\begin{enumerate}\itemsep0pt
\item $\alpha$-contractible subgraphs of size at most $1/\eps$;
\item irrelevant edges;
\item non-isolating $2$-vertex-cuts.
\end{enumerate}
\end{restatable}

\begin{restatable}{lemma}{lempreprocessing}\label{preprocessingLemma}
For all $\alpha \geq \frac{6}{5}$ and $\eps \in ( 0, \frac 1 {24}]$,
if there exists a deterministic polynomial-time $\alpha$-approximation algorithm for 2-ECSS on $(\alpha,\epsilon)$-structured graphs, then there exists a deterministic polynomial-time $(\alpha+2\epsilon)$-approximation algorithm for 2-ECSS.
\end{restatable}
When $\alpha$ and $\eps$ are clear from the context, we will simply talk about contractible subgraphs and structured graphs. Notice that we will present an algorithm with \amEdit{an} approximation factor worse than $6/5$ (though better than $4/3$), hence by Lemma \ref{preprocessingLemma} it is sufficient to consider $(\alpha,\eps)$-structured graphs for some $\alpha<4/3$. In particular, we will consider $(5/4,\eps)$-structured graphs for a sufficiently small constant $\eps>0$\footnote{The reader may wonder why we set $\alpha=5/4$ instead of $\alpha=6/5$. The reason is that a larger value of $\alpha$ excludes more contractible subgraphs, hence making our analysis simpler. Since our approximation factor for $(5/4,\eps)$-structured instances is larger than $5/4$, this has no impact on the final approximation factor.}.
Lemma \ref{preprocessingLemma} could be a valid starting point for future research on 2-ECSS approximation, and therefore might be of independent interest.

One of the key properties of a structured graph (which we will exploit multiple times) is the existence of the following 3-matchings:
\begin{restatable}[3-Matching Lemma]{lemma}{3MatchingLemma}\label{lem:matchingOfSize3}
Let $G=(V,E)$ be a 2VC simple graph without irrelevant edges and without non-isolating 2-vertex-cuts. Consider any partition $(V_1,V_2)$ of $V$ such that for each $i\in\{1,2\}$, $|V_i|\geq 3$ and if $|V_i|=3$, then $G[V_i]$ is a triangle. Then, there exists a matching of size $3$ between $V_1$ and $V_2$.
\end{restatable}
\begin{proof}
Consider the bipartite graph $F$ induced by the edges with exactly one endpoint in $V_1$. Let $M$ be a maximum (cardinality) matching of $F$. Assume by contradiction that $|M|\leq 2$. By K\"{o}nig-Egev\'{a}ry theorem\footnote{\fab{This theorem states that, in a bipartite graph, the cardinality of a maximum matching equals the cardinality of a minimum vertex cover, see e.g. \cite{S03}.}} there exists a vertex cover $C$ of $F$ of size $|M|$. We distinguish $3$ subcases:

\smallskip\noindent{\bf (1)} $C=\{u\}$. Assume w.l.o.g. $u\in V_1$. Since $C$ is a vertex cover of $F$, there are no edges in $F$ (hence in $G$) between the non-empty sets $V_1\setminus \{u\}$ and $V_2$. Hence $C$ is a $1$-vertex-cut, a contradiction.

\smallskip\noindent{\bf (2)} $C=\{u,v\}$ where $C$ is contained in one side of $F$. Assume w.l.o.g. $C\subseteq V_1$. Since $|V_1|\geq 3$ and $C$ is a vertex cover of $F$, $C$ is a 2-vertex-cut separating $V_1\setminus C$ from $V_2$. This implies that $uv\notin E(G)$ (otherwise $uv$ would be an irrelevant edge). Since $G[V_1]$ is not a triangle, $|V_1\setminus C|\geq 2$. This implies that $C$ is a non-isolating 2-vertex-cut of $G$, a contradiction.

\smallskip\noindent{\bf (3)} $C=\{u,v\}$ where $u$ and $v$ belong to different sides of $F$. Assume w.l.o.g. $u\in V_1$ and $v\in V_2$. Consider the sets $V'_1:=V_1\setminus \{u\}$ and $V'_2:=V_2\setminus \{v\}$, both of size at least $2$. Notice that there are no edges in $F$ (hence in $G$) between $V'_1$ and $V'_2$ (otherwise $C$ would not be a vertex cover of $F$). This implies that $C$ is a non-isolating 2-vertex-cut of $G$, a contradiction.\qedhere
\end{proof}

\subsection{A Canonical 2-Edge-Cover}

Let $G=(V,E)$ be an input instance of 2-ECSS which is $(\frac{5}{4},\eps)$-structured for a sufficiently  small constant $\eps>0$ to be fixed later (so that $|V(G)|\geq \frac{2}{\eps}$ is large enough). In the following we will simply say structured instead of $(\frac{5}{4},\eps)$-structured for shortness.
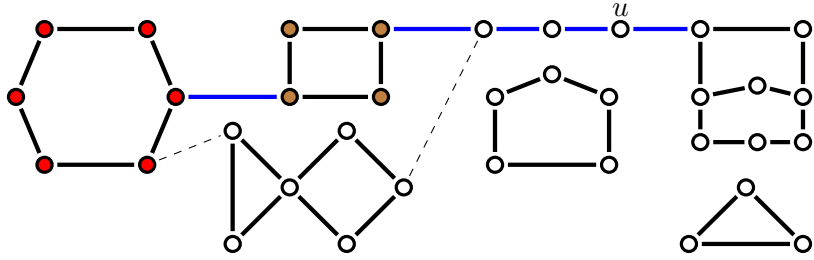
\begin{figure}
\begin{center}
\begin{tikzpicture}[scale=1.5]

\tikzset{vertex/.style={draw=black, very thick, circle,minimum size=0pt, inner sep=2pt, outer sep=2pt}
}


\begin{scope}[every node/.style={vertex}]

\node (1) at (-4.5,0.5){};
\node (2) at (-4.5,1.5){};
\node (3) at (-4,1){};
\node (a5) at (-3,1){};
\node (a4) at (-3.5,1.5){};
\node (a6) at (-3.5,0.5){};

\node[fill=brown] (4) at (-4,1.8){};
\node[fill=brown] (5) at (-4,2.4){};
\node[fill=brown] (7) at (-3.2,1.8){};
\node[fill=brown] (8) at (-3.2,2.4){};

\node (9) at (-0.4,2.4){};
\node (10) at (0.5,2.4){};
\node (11) at (-0.4,1.8){};

\node (16) at (0.1,1.9){};
\node (12) at (0.5,1.8){};
\node (13) at (-0.4,1.4){};
\node (14) at (0.1,1.4) {};
\node (15) at (0.5,1.4){};

\node (lonely1) at (-1.1,2.4){};
\node (lonely2) at (-1.7,2.4){};
\node (lonely3) at (-2.3,2.4){};

\node (18) at (0,1){};
\node (19) at (-0.5,0.5){};
\node (20) at (0.5,0.5){};

\node[fill=red] (lb1) at (-5,1.8){};
\node[fill=red] (lb3) at (-6.15,2.4){};
\node[fill=red] (lb5) at (-6.15,1.2){};
\node[fill=red] (lb2) at (-5.25,2.4){};
\node[fill=red] (lb6) at (-5.25,1.2){};
\node[fill=red] (lb4) at (-6.4,1.8){};

\node (c1) at (-1.7,2){};
\node (c2) at (-2.2,1.8){};
\node (c5) at (-1.2,1.8){};
\node (c4) at (-1.2,1.2){};
\node (c3) at (-2.2,1.2){};

\draw[ultra thick] (3) to (a4) to (a5) to (a6) to (3); 

\draw[ultra thick] (c1) to (c2) to (c3) to (c4) to (c5) to (c1);

\draw[ultra thick] (lb1) to (lb2) to (lb3) to (lb4) to (lb5) to (lb6) to (lb1);

\draw[ultra thick, blue] (lb1) to (4);
\draw[ultra thick, blue] (8) to (lonely3) to (lonely2) to (lonely1) to (9); 

\draw[ultra thick] (9) to (10) to (12) to (16) to (11) to (9);
\draw[ultra thick] (11) to (13) to (14) to (15) to (12);
\draw[ultra thick] (2) to (1) to (3) to (2);
\draw[ultra thick] (5) to (4) to (7) to (8) to (5);
\draw[ultra thick] (18) to (19) to (20) to (18);

\draw[dashed] (a5) to (lonely3);
\draw[dashed] (lb6) to (2);



\end{scope}

\node[above=1pt] () at (lonely1) {$u$};

\begin{scope}[very thick]

\end{scope}

\begin{scope}[very thick, densely dashed]

\end{scope}

\end{tikzpicture}
\end{center}
\caption{
A canonical 2-edge-cover $H$. The blue edges are bridges. Node $u$ is lonely. The red and brown nodes induce a leaf and inner block, resp., of the unique connected component $C$ of $S$ which is not 2EC. The dashed edges induce a (cheap) bridge-covering path $P_C$ for $C$ involving $bl=2$ blocks of $C$, $br=2$ bridges of $C$, and one connected component other than $C$.
\label{fig:canonical}
}
\end{figure}

Analogously to \cite{HVV19}, an initial step in our construction is to compute a minimum-size 2-edge-cover $H\subseteq E$. We call each (connected) component of $H$\footnote{Here and in similar cases in the rest of the paper, we mean the components of the corresponding subgraph $(V,H)$.} which is 2EC, a \emph{2EC component} of $H$. Given a component $C$ of $H$ which is \emph{not} 2EC, we call any maximal 2EC subgraph of $C$ containing at least $2$ (hence $3$) nodes a \emph{block}, any node of $C$ not contained in a block \emph{lonely}, and any edge of $C$ whose removal splits $C$ into two components a \emph{bridge} (see Figure \ref{fig:canonical}). Notice that in our notation blocks and 2EC components of $H$ are distinct. In particular the 2EC components, blocks\amEdit{,} and bridges of $H$ induce a partition of $H$. 
A block $B$ of $C$ is a \emph{leaf block} if there exists precisely one bridge $b$ of $C$ incident to $B$, and an \emph{inner block} otherwise.

Thanks to the fact that $G$ is structured (and in particular exploiting the 3-matching Lemma \ref{lem:matchingOfSize3}), we can enforce that $H$ is of the following canonical form (see Figure \ref{fig:canonical}).

\begin{restatable}{definition}{defCanonical}[canonical 2-edge-cover]\label{def:canonical}
A 2-edge-cover $H$ of $G$ is \emph{canonical} if the following properties hold:
\item[(i)] The 2EC components of $H$ are $i$-cycles for $3\leq i\leq 6$ or contain at least $7$ edges.
\item[(ii)] The leaf blocks contain at least $6$ edges and the inner ones at least $4$ edges. 

\item[(iii)] ($3$-optimality) It is not possible to obtain from $H$ a 2-edge-cover of the same size with fewer \fab{connected} components by adding up to $3$ edges and removing the same number of edges.
\end{restatable}

\begin{restatable}{lemma}{lemCanonical}\label{lem:canonical}
Given a 2-edge-cover $H$ of a $(5/4,\eps)$-structured graph $G$, $\eps\leq 1/5$, in polynomial time one can compute a 2-edge-cover $H'$ of no larger size which is canonical.  
\end{restatable}
The proof of Lemma \ref{lem:canonical} is given in Section \ref{sec:canonical}. Enforcing property (iii) is trivial. Property (i) does not hold only if a 2EC component $C$ of $H$ is a bowtie (i.e., two triangles that share a common node) or a $K_{2,3}$ (i.e., a complete bipartite graph with one side of size $2$ and the other of size $3$). In both cases it is possible to reduce the number of edges or the number of connected components (using the fact that $G$ is structured). In order to enforce (ii) we consider any block $B$ with at most $5$ edges and such that all the bridges incident to $B$ are incident to a single node (this includes leaf blocks as a special case). The 3-matching Lemma \ref{lem:matchingOfSize3} allows us to reduce the number of connected components or the number of bridges while not increasing the number of edges.

\subsection{A Credit Assignment Scheme}

We will next assume that we are given a minimum-size 2-edge-cover $H$ which is also canonical.
The high-level structure of our algorithm is as follows. Starting from $S=H$, we gradually transform $S$ by adding and removing edges, until $S$ becomes a feasible solution. In order to keep the size of $S$ under control, we define a credit assignment scheme that assigns credits $cr(S)\geq 0$ to $S$. Let us define the cost of $S$ as $cost(S)=|S|+cr(S)$. Initially $cost(S)=cost(H)\leq \alpha |H|$ for some $\alpha<\frac{4}{3}$. During the process we iteratively replace $S$ with a different solution $S'$ satisfying $cost(S')\leq cost(S)$. Typically we make progress towards a feasible solution by guaranteeing that $S'$ has fewer bridges and/or fewer connected components than $S$. This way at the end of the algorithm one obtains a feasible solution $S$ of size $|S|\leq cost(S)\leq cost(H)\leq \alpha |H|$. This provides the desired $\alpha<\frac{4}{3}$ approximation since, as already observed, $|H|\leq \opt$.    

In a preliminary draft of this paper we initially assigned $\delta<\frac{1}{3}$ credits to each edge of $H$, so that $cost(H)<\frac{4}{3}\opt$. The notation and case analysis however were extremely complex. So we decided to present a simpler (still unfortunately not very simple) approach which leads to a slightly worse approximation factor. This simplified approach uses two different algorithms and credit assignment schemes depending on the number of 2EC components of $H$ which are triangles. In more detail, suppose that $t|H|$ edges of $H$ belong to triangle 2EC components, and $b|H|$ edges of $H$ are bridges. Obviously $0\leq t+b\leq 1$. Then the following Lemma (proof in Section \ref{sec:manyTriangles}) will give a better than $4/3$ approximation when $t$ is sufficiently close to $1$ \fab{(namely, for any $t>19/20$)}. 

\begin{restatable}{lemma}{lemmaManyTriangles}\label{lem:manyTriangles:main}
Given a canonical minimum-size 2-edge-cover $H$ of a structured graph $G$ with $b|H|$ bridges and $t|H|$ edges belonging to triangle 2EC components. In polynomial time one can compute a $\frac{13}{8}-\frac{1}{3}t+\frac{1}{2}b$ approximate solution for 2-ECSS on $G$.
\end{restatable}

In the complementary case (i.e. for $t$ sufficiently close to $0$, \fab{namely} for any $t<1$), we exploit the following Lemma (proof in Section \ref{sec:fewTriangles}).

 \begin{restatable}{lemma}{lemmaFewTriangles}\label{lem:fewTriangles:main}
Given a canonical minimum-size 2-edge-cover $H$ of a structured graph $G$ with $b|H|$ bridges and $t|H|$ edges belonging to triangle EC components. In polynomial time one can compute a $\frac{13}{10}+\frac{1}{30}t-\frac{1}{20}b$
approximate solution for 2-ECSS on $G$.
\end{restatable}
Combining Lemmas \ref{preprocessingLemma}, \ref{lem:canonical}, \ref{lem:manyTriangles:main}, and \ref{lem:fewTriangles:main}, one obtains a $\min\{\frac{13}{8}-\frac{1}{3}t+\frac{1}{2}b,\frac{13}{10}+\frac{1}{30}t-\frac{1}{20}b\}+\eps$ approximation for 2-ECSS. The reader might easily check that this is $\frac{117}{88}+\eps$ in the worst case, hence already improving on $\frac{4}{3}$ for a small enough constant $\eps>0$.

The refined approximation factor claimed in Theorem \ref{thr:refinedApx} is based on a refinement of Lemma \ref{lem:manyTriangles:main} to a $\frac{14}{9}-\frac{8}{27}t+\frac{4}{9}b$ approximation (Lemma \ref{lem:manyTriangles:refined} in Section \ref{sec:refinedApproximation}) that exploits some ideas in \cite{SV14}. The rest of this section is devoted to a high-level description of the proofs of Lemmas \ref{lem:manyTriangles:main} and \ref{lem:fewTriangles:main}

\subsection{Case of Many Triangles: Overview}
\label{sec:manyTrianglesOverview}

We sketch the proof of Lemma \ref{lem:manyTriangles:main} (for the details see Section \ref{sec:manyTriangles}). The first step of the construction is to remove all the bridges $H_{brg}$ of $H$, hence obtaining a subgraph $S=H\setminus H_{brg}$ whose connected components are 2EC. We call the singleton nodes \emph{lonely}, and we consider them as (degenerate) 2EC components.
 
Then there is a \emph{gluing stage} where we gradually merge the 2EC components of $S$ by adding and removing edges (keeping the invariant that each connected component of $S$ is 2EC). In order to keep the size of $S$ under control, we assign certain credits to the 2EC components of $S$ as follows. We say that a 2EC component of $S$ is \emph{light} if it is a triangle connected component of $H$ (hence of the initial $S$) and \emph{heavy} otherwise. We remark that the gluing phase might create new triangle connected components, which are therefore heavy. Let $\calC(S)$ be the 2EC components of $S$. We define a credit assignment scheme which assigns $cr(C)$ credits to each $C\in \calC(S)$ according to the following simple rule:
\begin{itemize}\itemsep0pt
\item each light 2EC component $C$ of $S$ receives $cr(C)=\frac{1}{2}$ credits;
\item each heavy 2EC component $C$ of $S$ receives $cr(C)=2$ credits.
\end{itemize}
Let $cr(S):=\sum_{C\in \calC(S)}cr(C)$ be the total amount of credits assigned to $S$, and define the cost of $S$ as $cost(S)=|S|+cr(S)$. 
Recall that $H$ contains $b|H|$ bridges and $t|H|$ edges in its triangle 2EC components. Let us upper bound the cost of the initial $S$. \fab{Recall that $H$ is a canonical 2-edge-cover.}

\begin{restatable}{lemma}{lemmaCaseBinitialCost}\label{lem:caseB:initialCost}
$cost(H\setminus H_{brg})\leq (\frac{3}{2}-\frac{1}{3}t+\frac{1}{2}b)|H|$.
\end{restatable}
\begin{proof}
We initially assign $\frac{1}{6}$ credits to each edge of a triangle 2EC component of $H$, $2$ credits to each bridge of $H$, and $\frac{1}{2}$ credits to the remaining edges in $H$. Hence the total number of credits is $(\frac{1}{6}t+2b+\frac{1}{2}(1-t-b))|H|=(\frac{1}{2}-\frac{1}{3}t+\frac{3}{2}b)|H|$. We next redistribute these credits to the 2EC components of $H\setminus H_{brg}$ as follows. Each light (triangle) component $C$ of $H\setminus H_{brg}$ keeps the $\frac{1}{6}\cdot 3$ credits of its edges. Similarly each 2EC component $C$ of $H\setminus H_{brg}$ deriving from a block of $H$ or from a 2EC component of $H$ which is not a triangle, keeps the credits of its edges (at least $\frac{1}{2}\cdot 4$ credits since $|E(C)|\geq 4$ being $H$ canonical). Each bridge assigns its 2 credits to some lonely node: this is doable since the number of lonely nodes is upper bounded by $|H_{brg}|$. Altogether
$
cost(S)=|H|-|H_{brg}|+cr(S)\leq |H|-b|H|+(\frac{1}{2}-\frac{1}{3}t+\frac{3}{2}b)|H|.
$
\end{proof}
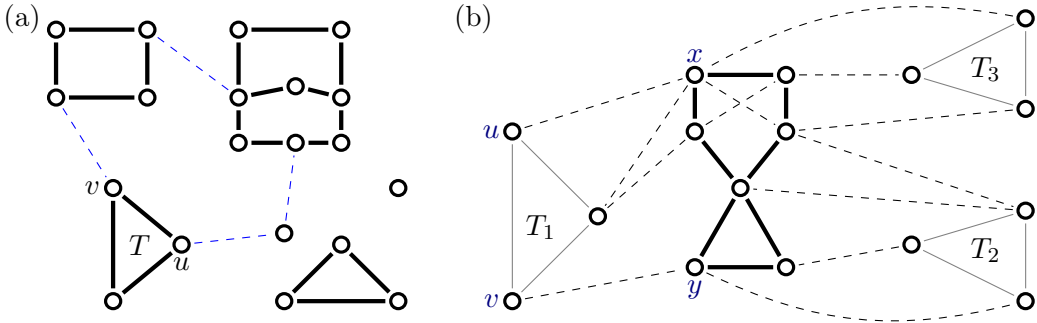
\begin{figure}
\begin{center}
\begin{tikzpicture}[scale=1.5]

\tikzset{vertex/.style={draw=black, very thick, circle,minimum size=0pt, inner sep=2pt, outer sep=2pt}
}

\node[right] () at (-4.55,2.5) {(a)};
\node[right] () at (-0.6,2.5) {(b)};
\node[right] () at (-3.45,0.5){$T$};

\begin{scope}[every node/.style={vertex}]

\node (1) at (-3.5,0){};
\node (2) at (-3.5,1){};
\node (3) at (-2.9,0.5){};

\node (4) at (-4,1.8){};
\node (5) at (-4,2.4){};
\node (7) at (-3.2,1.8){};
\node (8) at (-3.2,2.4){};

\node (9) at (-2.4,2.4){};
\node (10) at (-1.5,2.4){};
\node (11) at (-2.4,1.8){};

\node (16) at (-1.9,1.9){};
\node (12) at (-1.5,1.8){};
\node (13) at (-2.4,1.4){};
\node (14) at (-1.9,1.4) {};
\node (15) at (-1.5,1.4){};
\node (17) at (-2,0.6){};
\node (18) at (-1.5,0.5){};
\node (19) at (-2,0){};
\node (20) at (-1,0){};
\node (21) at (-1,1){};

\draw[ultra thick] (9) to (10) to (12) to (16) to (11) to (9);
\draw[ultra thick] (11) to (13) to (14) to (15) to (12);
\draw[ultra thick] (2) to (1) to (3) to (2);
\draw[ultra thick] (5) to (4) to (7) to (8) to (5);
\draw[ultra thick] (18) to (19) to (20) to (18);

\draw[dashed,blue] (3) to (17);
\draw[dashed,blue] (17) to (14); 
\draw[dashed,blue] (11) to (8);
\draw[dashed,blue] (4) to (2);

\end{scope}

\node[below=1pt] () at (3) {$u$};
\node[left=1pt] () at (2) {$v$};


\node[below] () at (0.25,0.85) {$T_1$};
\node[below] () at (4.15,0.65) {$T_2$};
\node[below] () at (4.15,2.25) {$T_3$};

\begin{scope}[every node/.style={vertex}]

\node (a15) at (0,0) {};
\node (a16) at (0,1.5) {};
\node (a14) at (0.75,0.75) {};
\node (a5) at (1.6,2) {};
\node (a4) at (1.6,1.5) {};
\node (a0) at (2,1) {};
\node (a7) at (2.4,1.5) {};
\node (a6) at (2.4,2) {};
\node (a2) at (2.4,0.3) {};
\node (a3) at (1.6,0.3) {};
\node (a12) at (4.5,0) {};
\node (a8) at (4.5,0.8) {};
\node (a11) at (3.5,0.5){};

\node (a10) at (4.5,1.7){};
\node (a13) at (4.5,2.5){};
\node (a9) at (3.5,2){};

\draw[gray] (a15) to (a14);
\draw[gray] (a16) to (a15);
\draw[gray] (a16) to (a14);

\draw[gray] (a11) to (a8) to (a12);
\draw[gray] (a11) to (a12);

\draw[gray] (a9) to (a10) to (a13);
\draw[gray] (a9) to (a13);

\draw[dashed] (a16) to (a5) to (a14) to (a4) to (a6) to (a9);
\draw[dashed,bend right=20] (a13) to (a5);
\draw[dashed] (a5) to (a7) to (a8) to (a0);
\draw[dashed] (a15) to (a3);
\draw[dashed] (a11) to (a2);
\draw[dashed] (a10) to (a7);
\draw[dashed,bend right=20] (a3) to (a12);

\draw[ultra thick] (a0) to (a4) to (a5) to (a6) to (a7) to (a0) to (a3) to (a2) to (a0);

\end{scope}

\node[left=1pt, blue!50!black] (u) at (a16) {$u$};
\node[left=1pt, blue!50!black] (v) at (a15) {$v$};
\node[above=1pt, blue!50!black] (x) at (a5) {$x$};
\node[below=1pt, blue!50!black] (y) at (a3) {$y$};

\begin{scope}[red!80!black, very thick]
\end{scope}

\begin{scope}[very thick]

\end{scope}

\begin{scope}[very thick, densely dashed]

\end{scope}

\end{tikzpicture}
\end{center}
\caption{
(a) The dashed blue edges induce a merging cycle $M$ w.r.t. the current solution $S$ given by black edges. Observe that all the connected components of $S$ are 2EC. A new 2-edge-cover $S'$ is obtained from $S$ by adding $M$ are removing $uv$. (b) The solid edges define a core-triangle 2-edge-cover $S$, where the core $C$ is induced by the thick edges and the triangles $T_1,T_2,T_3$ by the gray edges. The remaining edges of the input graph $G$ are dashed. Observe that $C$ is 2EC and there are no edges in $G$ between distinct triangles. A feasible solution can be obtained by adding the edges $\{ux,vy\}$ and removing the edge $uv$ w.r.t. $T_1$, and performing a similar update for the remaining 2 triangles.
\label{fig:MergingCycleAndCore-triangle}
}
\end{figure}

The gluing phase halts when we obtain an $S$ which has the following special form (see Figure \ref{fig:MergingCycleAndCore-triangle}\textcolor{red}{.b}):
\begin{definition}[core-triangle 2-edge-cover]\label{specialConfigurationDef}
A 2-edge-cover of a graph $G$ is \emph{core-triangle} if it consists of $k+1$ 2EC components, namely one \emph{core} component $C$ plus $k$ triangles $T_1,\ldots,T_k$ for $k\geq 0$. Furthermore, nodes of distinct triangles are not adjacent in $G$. 
\end{definition}
We make progress towards a 2-edge-cover of the latter form by iteratively applying the following lemma. 
\begin{restatable}{lemma}{lemmaCaseBgluingStep}\label{lem:caseB:gluingStep} 
Suppose that $S$ is not core-triangle. Then in polynomial time we can compute $S'\subseteq E(G)$ such that: (1) the connected components of $S'$ are 2EC, (2) $S'$ has fewer connected components than $S$, and (3) $cost(S')\leq cost(S)$.
\end{restatable}

In order to prove the above lemma, we introduce the notion of merging cycle\footnote{Here and in other parts of the paper we interpret a path or cycle as a subgraph, a subset of nodes or edges, or a sequence of nodes or edges, depending on which notation is more convenient.} (see Figure \ref{fig:MergingCycleAndCore-triangle}\textcolor{red}{.a}). 
\begin{restatable}{definition}{definitionMergingCycle}[merging cycle]\label{def:mergingCycle} 
A merging cycle $M$ for $S$ is a collection of edges $M\subseteq E(G)\setminus S$ such that: (1) $M$ induces one cycle (excluding loops) in the graph obtained by collapsing the connected components $\calC(S)$ of $S$; (2) if two edges of $M$ are incident to a light component $C\in \calC(S)$, they are incident to two distinct nodes of $C$.
\end{restatable}

Given a merging cycle $M$ of $S$, one naturally obtains a new $S'$ as follows: we add the edges of $M$ to $S$ and, for each light (triangle) component $C$ of $S$ such that $M$ is incident to the nodes $u$ and $v$ of $C$, we remove the edge $uv$. Notice that $S'$ satisfies conditions (1) and (2) of Lemma \ref{lem:caseB:gluingStep}.
More specifically, all the 2EC components of $S$ incident to $M$ become a unique 2EC component $C'$ in $S'$, and the remaining 2EC of $S$ remain the same in $S'$. Let us analyze $cost(S')$. Suppose that $M$ is incident to $n_L$ light components of $S$ and $n_H$ heavy components of $S$. One has that $|S'|-|S|=n_H+n_L-n_L=n_H$. Furthermore $cr(S')-cr(S)\leq 2-\frac{1}{2}n_L-2n_H$, where the term $2$ is due to $cr(C')$.  Thus
$
cost(S)-cost(S')\geq n_H+\frac{1}{2}n_L-2.
$
Notice that $n_H=0$ and $n_L\leq 3$ is not possible since this would imply that the initial 2-edge-cover $H$ is not canonical (in particular the 3-optimality property would not hold). A simple case analysis shows that $cost(S)-cost(S')\geq 0$ (hence (3) holds) unless $n_H=1=n_L$. In other words, the merging cycle $M$ is incident to precisely one light and one heavy component: we call a merging cycle of the latter type \emph{expensive}, and any other merging cycle \emph{cheap}. 


The proof of Lemma \ref{lem:caseB:gluingStep} proceeds constructively as follows. Given the current $S$ (which is not core-triangle), we maintain an auxiliary \emph{forest} $G'=(\calC(S),F)$ whose nodes are the 2EC components of $S$ and whose edges correspond to expensive merging cycles that we discovered along the way. Initially $F=\emptyset$. While $G'$ consists of at least $2$ trees we are able to compute efficiently a merging cycle $M$, exploiting the fact that $G$ is structured (see Lemma \ref{lem:caseB:mergingCycle} in Section \ref{sec:manyTriangles}). If the merging cycle $M$ obtained this way is cheap, we \fab{use it to obtain an $S'$ with the desired properties as described before}. Otherwise, we add to $F$ the edge between the two 2EC components of $S$ (one heavy and one light) incident to $M$.
Clearly the process ends in a polynomial number of steps by returning the desired cheap merging cycle or by obtaining a single tree $G'$. 

In the latter case, we check whether there exists an edge $e$ in $G$ between two connected components $C_1$ and $C_2$ of $S$ which are \emph{not} adjacent in $G'$. If such an $e$ exists, it is not hard to derive a cheap merging cycle involving all the (at least $3$) components along the $C_1$-$C_2$ path in $F$ (see Lemma \ref{lem:caseB:nonTreeEdge} in Section \ref{sec:manyTriangles}).

If there is no edge $e$ of the above type, since by assumption $S$ is not core-triangle, we can infer that there exists a light component $C$ adjacent to two heavy components $C_1$ and $C_2$ in $F$. A simple case analysis shows that we can merge these $3$ components together, hence obtaining the desired $S'$ (see Lemma \ref{lem:caseB:CoreTriangle} in Section \ref{sec:manyTriangles}).



Once $S$ is core-triangle, with core $C$ and triangles $T_1,\ldots,T_k$, we can transform it into a 2EC spanning subgraph  $\APX$ via the following simple construction (see Figure \ref{fig:MergingCycleAndCore-triangle}\textcolor{red}{.b}). For each $T_i$, let $e_i=u_iu'_i$ and $f_i=v_iv'_i$ be two edges with $u_i,v_i\in V(T_i)$, $u_i\neq v_i$, and $u'_i,v'_i\in V(C)$. Notice that the pair $e_i,f_i$ must exist otherwise some node in $V(T_i)$ would be a 1-vertex-cut. Set $Q_i=E(T_i)\setminus \{u_iv_i\}\cup \{e_i,f_i\}$ and $\APX:=E(C)\cup \bigcup_{i=1}^{k} Q_i$. 
The proof of Lemma \ref{lem:manyTriangles:main} follows easily.

\begin{proof}[Proof of Lemma \ref{lem:manyTriangles:main}]
Consider the feasible solution $\APX$ obtained with the above construction (which takes polynomial time). Observe that $|APX|=|S|+k$. By construction and due to Lemma \ref{lem:caseB:gluingStep} one has $cost(S)\leq cost(H\setminus H_{brg})$. Thus
$
cost(S)=|S|+cr(S)\leq cost(H\setminus H_{brg})\leq (\frac{3}{2}-\frac{1}{3}t+\frac{1}{2}b)|H|
$, where the last inequality comes from Lemma \ref{lem:caseB:initialCost}.

We also observe that $cr(S)\geq \frac{k}{2}$ where in the worst case each triangle $T_i$ is light\footnote{Indeed the construction can be slightly modified so that each heavy $T_i$ is merged with the core $C$ by exploiting the credits of $T_i$ to pay for the extra needed edge.}. Then
$
|\APX|=|S|+k\leq (\frac{3}{2}-\frac{1}{3}t+\frac{1}{2}b)|H|+\frac{k}{2}.
$
Observe that any feasible solution must contain at least $4$ distinct edges per triangle $T_i$ \fab{and nodes in distinct triangles are not adjacent by Definition \ref{specialConfigurationDef}}, hence $\opt\geq 4k$. We already observed that $\opt\geq |H|$. Altogether
$$
\hspace{2cm}|\APX|\leq (\frac{3}{2}-\frac{1}{3}t+\frac{1}{2}b)\opt+\frac{1}{8}\opt =(\frac{13}{8}-\frac{1}{3}t+\frac{1}{2}b)\opt.\hspace{2cm}\qedhere
$$
\end{proof}
In Section \ref{sec:refinedApproximation} we will present a more refined construction inspired by \cite{SV14}, and a corresponding refinement of Lemma \ref{lem:manyTriangles:main} (namely, Lemma \ref{lem:manyTriangles:refined}). \amEdit{at a high leve}l, instead of choosing the $Q_i$'s arbitrarily, we do that so as to minimize the number $\alpha(S)$ of connected components induced by $Q:=\cup_i Q_i$ (this can be done in polynomial time analogously to \cite{SV14} via a reduction to matroid intersection). Then, in a way similar to \cite{SV14}, we add a minimal set of edges $F$ to $Q$ to make $Q\cup F$ connected, and then add a minimum-size $T$-join $J$ over the odd-degree nodes of $Q\cup F$ to obtain an alternative feasible solution $\APX':=Q\cup F\cup J$. We show that $\APX'$ has size at most $4k+\alpha(S)+\frac{1}{2}|E(C)|-1$, and that $\opt\geq 4k+\alpha(S)-1$. Taking the best solution among $\APX$ and $\APX'$ gives the refined bound.


\subsection{Case of Few Triangles: Overview}
\label{sec:fewTrianglesOverview}

We next sketch the proof of Lemma \ref{lem:fewTriangles:main}. Let us initially set $S=H$. 
We gradually transform $S$ by adding and deleting edges, until $S$ becomes a 2EC spanning subgraph. During the process we will maintain the following invariant:
\begin{invariant}\label{inv:2ECcomponent}
The 2EC components of $S$ are an $i$-cycle for $3\leq i\leq 6$ belonging to $H$ or contain at least $7$ edges (we call \emph{large} the latter components).
\end{invariant}
Clearly $S$ initially satisfies this property since $H$ is canonical. During our construction we always merge together different connected components of $S$ so that the resulting connected component $C$ contains at least $7$ nodes (hence $7$ edges if $C$ is 2EC).

In order to keep the size of $S$ under control, we use a credit assignment scheme analogously to Section \ref{sec:manyTrianglesOverview}, however more complex. In this scheme we assign certain credits to the connected components, bridges, and blocks of $S$. In more detail, we assign the following credits:
\begin{enumerate}\itemsep0pt
\item Each 2EC component $C$ of $S$ receives $cr(C)=2$ credits if $C$ is large (i.e., contains at least $7$ edges), $cr(C)=1$ credits if $C$ is a triangle, and $cr(C)=\frac{3}{10}k$ credits if $C$ is a $k$-cycle, $4\leq k\leq 6$.
\item Each connected component $C$ of $S$ which is not 2EC receives $cr(C)=1$ credits.
\item Each bridge $e$ of $S$ receives $cr(e)=\frac{1}{4}$ credits.
\item Each block $B$ of $S$ receives $cr(B)=1$ credits. 
\end{enumerate} 
Let $cr(S)$ be the total number of credits assigned to $S$ according to the above rule, and define $cost(S)=|S|+cr(S)$ analogously to Section \ref{sec:manyTrianglesOverview}. The next lemma shows that $cost(S)$ is initially not too large. Its simple proof is based on initially assigning $\frac{1}{3}$, $\frac{1}{4}$ and $\frac{3}{10}$ credits to each edge in $H$ which belongs to a triangle 2EC component, which is a bridge, and to the remaining edges, resp. We then redistribute these credits in a natural way, exploiting the fact that $H$ is canonical.

\begin{restatable}{lemma}{lemmaCaseAinitialCost}\label{lem:caseA:initialCost}
For $S=H$, $cost(S)\leq (\frac{13}{10}+\frac{1}{30}t-\frac{1}{20}b)|H|$.
\end{restatable}

The \amEdit{transformation} of $S$ into a feasible solution consists of two main phases. In an initial \emph{bridge-covering} phase, we gradually reduce the number of bridges of $S$ until all its connected components are 2EC. This is done by iteratively applying the following lemma.
\begin{restatable}{lemma}{lemmaCaseAbridgeCovering}[Bridge-Covering]\label{lem:caseA:bridgeCovering}
Assume that $S$ contains at least one bridge. Then in polynomial time one can compute a 2-edge-cover $S'\subseteq E(G)$ such that: (1) $S'$ has fewer bridges than $S$; (2) $cost(S')\leq cost(S)$.
\end{restatable}

The main building block in our construction is a \emph{bridge-covering path} (see Figure \ref{fig:canonical}). The idea is to consider any connected component $C$ of $S$ which is not 2EC, and collapse into a single node each block of $C$ and any other connected component $C'\neq C$. Observe that $C$ becomes a tree $T_C$. In the obtained graph $G_C$ we consider any path $P_C$ with edges not in $E(T_C)$, two distinct endpoints in $V(T_C)$, and all the internal nodes not in $V(T_C)$ ($P_C$ might consist of a single edge). Consider $S'=S\cup E(P_C)$, which obviously has fewer bridges than $S$. W.r.t. $S$, we increase the number of edges by $|E(P_C)|$ and we create a new block $B'$ in $C$ (which requires $1$ credit) or we transform $C$ into a 2EC component $C'$ (which requires $1$ additional credit w.r.t. the credit owned by $C$). However we also decrease the credits by at least $|E(P_C)|-1$ due to the connected components which are merged with $C$. Hence it is sufficient to collect a total of $2$ credits among the $bl$ blocks and $br$ bridges of $C$ that disappear in $S'$. In particular, to guarantee $cost(S)\geq cost(S')$, it is sufficient that $\frac{1}{4}br+bl- 2\geq 0$. We say that $P_C$ is cheap if the latter quantity is non-negative, and expensive otherwise. In particular, any $P_C$ involving at least $2$ block nodes or $1$ block node and at least $4$ bridges satisfies this condition. If there exists a cheap bridge-covering path $P_C$, we simply set $S'=S\cup E(P_C)$. Notice that this condition can be easily checked in polynomial time. If no cheap bridge-covering path exists, by exploiting the 3-matching Lemma \ref{lem:matchingOfSize3} and a relatively simple case analysis, we show how to obtain the desired $S'$: to achieve that we will combine together up to two bridge-covering paths and/or we will remove some edges from $S$. 



After the bridge-covering phase there is a \emph{gluing} phase where we gradually merge together the 2EC components of $S$ until we achieve a feasible solution. This is similar in spirit to the gluing phase of Section \ref{sec:manyTrianglesOverview}, with the difference that here we can directly reach a feasible solution without the need for a separate procedure to handle a core-triangle 2-edge-cover as done in Lemma \ref{lem:manyTriangles:main}. Intuitively this is due to the fact that each triangle $T_i$ in a core-triangle 2-edge-cover now has $1$ credit which is sufficient to glue it with the core $C$. In the gluing phase we iteratively apply the following lemma. 
\begin{restatable}{lemma}{lemmaCaseAgluing}[Gluing]\label{lem:caseA:gluing}
Assume that the connected components of $S$ are 2EC, and there are at least 2 such components. Then in polynomial time one can compute a 2-edge-cover $S'\subseteq E(G)$ such that: (1) the connected components of $S'$ are 2EC; (2) $S'$ has fewer connected components than $S$; (3) $cost(S')\leq cost(S)$.
\end{restatable}
In order to prove the above lemma we introduce an auxiliary graph $\hat{G}_S$ (\emph{component graph}) whose nodes are the 2EC components of $S$ and where there is an edge $C_1C_2$ iff $G$ contains at least one edge between the components $C_1$ and $C_2$ of $S$. A relatively simple case analysis allows us to rule out certain local configurations (or find the desired $S'$ in polynomial time), see Lemmas \ref{lem:4cycle}-\ref{lem:3componentsC1_5cycle}. Given that, it is relatively easy to find the desired $S'$ in the case that  $\hat{G}_S$  is a tree (see Lemma \ref{lem:caseA:treeCase}): here we exploit the 3-matching Lemma.
\begin{figure}
\begin{center}
\begin{tikzpicture}[scale=1.5]


\tikzset{vertex/.style={draw=black, very thick, circle,minimum size=0pt, inner sep=2pt, outer sep=2pt}
}

\node () at (0.2,0.3){$C_0$};
\node () at (2,-0.5){$C_1$};
\node () at (3.8,-0.8){$C_2$};
\node () at (5.3,-0.6){$C_3$};
\node () at (7,-0.6){$C_4$};

\node() at (1.05,-0.25){$e_1$};
\node () at (2.85,-0.9){$e_2$};
\node () at (4.33,-0.2){$e_3$};
\node () at (6.15,-0.3){$e_4$};

\node () at (0.4,-0.37){$out_0$};
\node () at (1.4,-0.65){$in_1$};
\node () at (2.4,-1.15){$out_1$};
\node () at (3.3,-1.15){$in_2$};
\node () at (3.75,-0.22){$out_2$};
\node () at (4.85,0.02){$in_3$};
\node () at (5.55,0.02){$out_3$};
\node () at (6.6,0.02){$in_4$};

\begin{scope}[every node/.style={vertex}]

\node (a1) at (-0.2,0.3) {};
\node (a2) at (0.4,0.8) {};
\node (a3) at (0.4,-0.2) {};
\draw[ultra thick] (a1) to (a2) to (a3) to (a1);

\node(b1) at (1.5,-0.5){};
\node (b2) at (1.8,-1){};
\node (b3) at (2.4,-1){};
\node (b4) at (2.4,0){};
\node (b5) at (1.8,0){};
\draw[ultra thick] (b1) to (b2) to (b3) to (b4) to (b5) to (b1);
\draw[dashed] (a3) to (b1);
\node (c1) at (3.3,-1) {};
\node (c2) at (4.3,-1) {};
\node (c3) at (3.8,-0.4) {};
\draw[ultra thick] (c1) to (c2) to (c3) to (c1);
\draw[dashed] (b3) to (c1);
\node (d1) at (4.9,-0.2) {};
\node (d2) at (4.9,-1){};
\node (d3) at (5.7,-1){};
\node (d4) at (5.7,-0.2){};
\draw[ultra thick] (d1) to (d2) to (d3) to (d4) to (d1);
\draw[dashed] (c3) to (d1);
\node (e1) at (6.6,-0.2){};
\node (e2) at (6.6,-1){};
\node (e3) at (7.4,-1){};
\node (e4) at (7.4,-0.2){};
\draw[ultra thick] (e1) to (e2) to (e3) to (e4) to (e1);
\draw[dashed] (d4) to (e1);
\draw[red,bend right=30] (e4) to (a2);
\draw[green,bend right=25] (e1) to (a3);
\draw[blue,bend right=35] (e4) to (a3);
\end{scope}
\begin{scope}[very thick]
\end{scope}
\begin{scope}[very thick, densely dashed]
\end{scope}
\end{tikzpicture}
\end{center}
\caption{The dashed edges $\{e_1,\dots,e_4\}$ form a gluing path involving the component $C_0,\ldots,C_4$. Let $e'$ be any edge among the green, blue, and red \amEdit{edges}. By adding $\{e_1,e_2,e_3,e_4,e'\}$ to the components $C_0,\ldots,C_4$ and removing $in_2out_2$ and $in_3out_3$, one obtains a 2EC component $C'$ spanning $V(C_0)\cup \ldots \cup V(C_4)$. If $e'$ is the blue (resp., red) edge, then we can remove from $C'$ one extra edge (resp., two extra edges) while keeping it 2EC.}
\label{fig:GluingPath}
\end{figure}
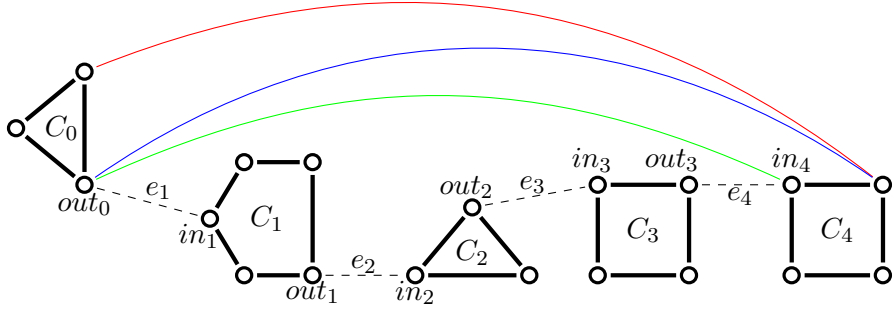

The most complicated part of our analysis is when $\hat{G}_S$ is \emph{not} a tree (see Lemma \ref{lem:caseA:nonTreeCase}). \amEdit{at a high leve}l, the idea (see Figure \ref{fig:GluingPath}) is to build a variant of the merging cycle $M$ from Section \ref{sec:manyTrianglesOverview} touching the components $C_0,\ldots,C_{|M|}$ of $S$, where we guarantee that the two edges of $M$ incident to $C_i$, $0<i<|M|$, are incident to adjacent nodes $in_i$ and $out_i$ of $C_i$ if $C_i$ is a triangle or a 4-cycle. We obtain $S'$ by adding $M$ to $S$ and removing the edges of type $in_iout_i$ mentioned before. Notice that $S'$ contains a 2EC component $C'$ spanning the nodes of the components $C_i$. W.r.t. $cost(S)$, $cost(S')$ grows by $|M|+2$ due to the edges of $M$ and the $2$ credits of $C'$. However, considering the removed credits and edges corresponding to each $C_i$, we obtain a decrease of $cost(S)$ by strictly more than $1$ from each $C_i$, $0<i<|M|$ (the smallest decrease is $1.5$ when $C_i$ is a 5-cycle). This allows one to easily rule out large values of $|M|$: for the remaining cases a more careful choice of $M$ and a refined case analysis provides the desired $S'$.

The desired approximation factor follows easily from Lemmas \ref{lem:caseA:initialCost}, \ref{lem:caseA:bridgeCovering} and \ref{lem:caseA:gluing}. 
\begin{proof}[Proof of Lemma \ref{lem:fewTriangles:main}]
Consider the above process. Clearly it ends within a polynomial number of transformation steps of $S$. Let $S'$ denote the final value of $S$. Observe that by construction 
$
|S'|\leq |S'|+cr(S')=cost(S')\leq cost(H)\leq (\frac{13}{10}+\frac{1}{30}t-\frac{1}{20}b)|H|,
$ 
where the first inequality comes from the fact that credits are non-negative and the last one from Lemma \ref{lem:caseA:initialCost}. The claim follows since $|H|\leq \opt$.
\end{proof}

\bibliographystyle{abbrv}
\bibliography{references}

\appendix

\section{Reduction to Structured Graphs}
\label{sec:preprocessing}

In this section we show how to reduce the 2-ECSS problem to $(\alpha,\eps)$-structured graphs. 
An optimal solution to a 2-ECSS instance $G$ will be denoted by $\OPT(G)$ or simply by $\OPT$ when $G$ is clear from the context. Likewise we will use $\opt(G)=|\OPT(G)|$ and $\opt=|\OPT|$.

Given $W\subseteq V$ we let $G\vert W$ denote the multi-graph obtained by contracting $W$ into one node $w$. For a subgraph $H$ of $G$, we use $G\vert H$ as a shortcut for $G\vert V(H)$. We will use the following facts and lemmas.
\begin{fact}\label{fact:contraction}
Let $G$ be a 2EC graph and $W\subseteq V(G)$. Then $G \vert W$ is 2EC. 
\end{fact}
\begin{fact}\label{fact:decontraction}
Let $H$ be a 2EC subgraph of a 2EC graph $G$, and $S$ and $S'$ be 2EC spanning subgraph of $H$ and $G\vert H$, resp. Then $S\cup S'$ is a 2EC spanning subgraph of $G$, where we interpret $S'$ as edges of $G$ after decontracting $V(H)$.
\end{fact}
We will also use the following standard results (see, e.g., \cite{HVV19}) which will allow us to focus on 2VC simple graphs\footnote{Part of our reductions temporarily create multi-graphs.}. 
\begin{lemma}\label{lem:cutNode} 
Let $G$ be a 2EC graph, $v$ be a 1-vertex-cut of $G$, and $(V_1,V_2)$, $V_1\neq \emptyset\neq V_2$, be a partition of $V\setminus \{v\}$ so that there are no edges between $V_1$ and $V_2$.
Then  $\opt(G)=\opt(G_1)+\opt(G_2)$, where $G_i:=G[V_i\cup \{v\}]$.
\end{lemma}
\begin{lemma}\label{lem:simple}
Let $G$ be a 2VC multigraph with $|V(G)| \geq 3$. Then, there exists a minimum size 2EC spanning subgraph of $G$ that is simple (i.e. contains no self-loops nor parallel edges).
\end{lemma}

We already showed how to handle irrelevant edges (Lemma \ref{lem:irrelevant}). In order to handle the case of a non-isolating 
$2$-vertex-cut $\{u,v\}$ we need to introduce some notation and technical lemmas. 
\begin{lemma}\label{lem:nicePartition}
Given a non-isolating $2$ vertex-cut $\{u,v\}$ of $G$ and assuming $|V(G)|\geq 6$, in polynomial time one can find a partition $(V_1,V_2)$ of $V\setminus \{u,v\}$ with $2\leq |V_1|\leq |V_2|$ such that there are no edges between $V_1$ and $V_2$.
\end{lemma}
\begin{proof}
Let $C_1,\ldots,C_k$ be the connected components of $G':=G\setminus \{u,v\}$ in non-decreasing order of size. If $k=2$ simply set $V_i=V(C_i)$. Observe that $|V_1|\geq 2$ as required since otherwise $\{u,v\}$ would be isolating. If $k\geq 4$, set $V_1=V(C_1)\cup V(C_2)$ and $V_2=V(C_2)\cup \ldots \cup V(C_k)$. Otherwise ($k=3$) notice that $|V(C_3)|\geq 2$ since otherwise $|V(G)|= 5$. In this case set $V_1=V(C_1)\cup V(C_2)$ and $V_2=V(C_3)$. 
\end{proof}

For the following definition and lemma, see Figure \ref{fig:TypeABC}.
\begin{definition}\label{def:ABC}
Consider any graph $H$ and two nodes $u,v\in V(C)$. Let $H'$ be obtained by contracting each 2EC component $C$ of $H$ into a single super-node $C$, and let $C(u)$ and $C(v)$ be the super-nodes corresponding to $u$ and $v$, resp. Then $H$ is: 
\begin{enumerate}\itemsep0pt
\item of type A w.r.t. $\{u,v\}$ if $H'$ consists of a single super-node $C(u)=C(v)$;
\item of type B w.r.t. $\{u,v\}$ if $H'$ is a $C(u)$-$C(v)$ path (of length at least one);
\item of type C w.r.t. $\{u,v\}$ if $H'$ consists of two isolated super-nodes $C(u)$ and $C(v)$. 
\end{enumerate}
\end{definition}

\begin{lemma}\label{lem:ABC}
Let $G$ be a simple 2VC graph without irrelevant edges, $\{u,v\}$ be a $2$-vertex-cut of $G$,  $(V_1,V_2)$, $V_1\neq \emptyset\neq V_2$, be a partition of $V\setminus \{u,v\}$ such that there are no edges between $V_1$ and $V_2$, and $H$ be a 2EC spanning subgraph of $G$. Define $G_i=G[V_i\cup \{u,v\}]$ and $H_i=E(G_i)\cap H$ for $i\in \{1,2\}$. The following claims hold:
\begin{enumerate}\itemsep0pt
\item\label{lem:ABC:feasible} Both $H_1$ and $H_2$ are of type A, B or C (w.r.t. $\{u,v\}$). Furthermore, if one of the two is of type C, then the other must be of type A.
\item\label{lem:ABC:ABexists} If $H_i$ is of type $C$, then there exists one edge $f\in E(G_i)$ such that $H'_i:=H_i\cup \{f\}$ is of type $B$. As a consequence, there exists a 2EC spanning subgraph $H'$ where $H'_i:=H'\cap E(G_i)$ is of type A or B.  
\end{enumerate} 
\end{lemma}
\begin{proof}
Notice that $uv\notin E(G)$ (hence $uv\notin E(H)$) since there are no irrelevant edges.
\eqref{lem:ABC:feasible} We prove the claim for $H_1$, the other case being symmetric. First notice that if $H_1$ contains a connected component not containing $u$ nor $v$, $H$ would be disconnected. Hence $H_1$ consists of one connected component or two connected components $H_1(u)$ and $H_1(v)$ containing $u$ and $v$, resp. 

Suppose first that $H_1$ \amEdit{consists} of one connected component. Let us contract the 2EC components of $H_1$, hence obtaining a tree $T$. If $T$ consists of a single node, $H_1$ is of type $A$. Otherwise, consider the path $P$ (possibly of length $0$) between the super-nodes resulting from the contraction of the 2EC components $C(u)$ and $C(v)$ containing $u$ and $v$, resp.. Assume by contradiction that $T$ contains an edge $e$ not in $P$. Then $e$ does not belong to any cycle of $H$, contradicting the fact that $H$ is 2EC. Thus $H_1$ is of type $B$ (in particular $P$ has length at least $1$ since $H_1$ is not 2EC). 

Assume now that $H_1$ consists of $2$ connected components $H_1(u)$ and $H_1(v)$. Let $T_u$ and $T_v$ be the two trees obtained by contracting the 2EC components of $H_1(u)$ and $H_2(v)$, resp. By the same argument as before, the 2-edge-connectivity of $H$ implies that these two trees contain no edge. Hence $H_1$ is of type $C$. 

The second part of the claim follows easily since if $H_1$ is of type $C$ and $H_2$ is not of type $A$, then $H$ would either be disconnected or it would contain at least one bridge edge (here we exploit the fact that $uv\notin E(H)$).

\eqref{lem:ABC:ABexists} There must exist an edge $f$ in $G_i$ between the two connected components of $H_i$ since otherwise at least one of $u$ or $v$ would be a cut node. Clearly $H'_i:=H_i\cup \{f\}$ satisfies the claim. \fab{For the second claim, consider any 2EC spanning subgraph $H'$ and let $H'_i:=H'\cap E(G_i)$. The claim holds if there exists one such $H'$ where $H'_i$ is of type A or B, hence assume this is not the case. Hence $H'_i$ is of type C by property \eqref{lem:ABC:feasible}. The same argument as above implies the existence of an edge $f$ in $G_i$ such that $H''_i:=H'_i\cup \{f\}$ is of type $B$, implying the claim for $H'':=H'\cup \{f\}$.}
\end{proof}

We are now ready to describe the reduction behind Lemma \ref{preprocessingLemma}. We assume that we are given  a polynomial time algorithm $\mathsf{ALG}$ that, given a $(\alpha,\eps)$-structured graph $G=(V,E)$, returns an $\alpha$-approximate solution $\mathsf{ALG}(G)\subseteq E$ to 2-ECSS on $G$. We consider the recursive procedure $\mathsf{RED}$ described in Algorithm \ref{fig:RED}: given a 2EC multi-graph $G=(V,E)$, $\mathsf{RED}$ returns a solution $\mathsf{RED}(G)\subseteq E$ to 2-ECSS on $G$. $\mathsf{RED}$ uses $\mathsf{ALG}$ as a black-box subroutine. In Step \ref{step:contraction} we interpret the edges of $\mathsf{RED}(G\vert H)$ as the corresponding edges after decontracting $H$. We interpret similarly the edge sets $H'_i$ from Step \ref{step:definitionH'}. We will show that the graphs $G'_1$ and $G'_2$ from Line \ref{step:contraction_uv}, the graph $G''_2$ from Line \ref{step:dummyedge}, and the graph $G'''_2$ from Line \ref{step:dummynode} are all 2EC: this justifies the recursive use of $\mathsf{RED}$ in the following lines. We will also show that a valid $F$ as in Steps \ref{step:definitionF'} and \ref{step:definitionF''} exists. \fab{The computation from Line \ref{step:ABCcomputation} is performed as follows. First notice that $G_1$ contains a constant number of nodes, hence by brute force we can compute minimum-size subgraphs $\OPT_{1A}$, $\OPT'_{1B}$ and $\OPT'_{1C}$ of the desired types, when any such subgraph exists. Obviously $\OPT_{1A}$, if defined, forms a 2EC spanning subgraph with $G_2$. If $\OPT'_{1C}$ is defined, it belongs to a 2EC spanning subgraph iff $G_2$ is 2EC (i.e., $G_2$ is a type A subgraph) by Lemma \ref{lem:ABC}. If the latter condition holds, we set $\OPT_{1C}:=\OPT'_{1C}$. Otherwise, no type C subgraph of $G_1$ belongs to a feasible solution and we consider $\OPT_{1C}$ as undefined. 
Similarly, if $\OPT'_{1B}$ is defined, it belongs to a feasible solution iff $G_2$ is a type A or B subgraph. We set $\OPT_{1B}$ accordingly.} Finally notice that if Step \ref{step:ALG} is reached, then the graph is $(\alpha,\eps)$-structured.


\begin{algorithm}
\DontPrintSemicolon
\caption{Reduction from non-structured to structured instances of 2-ECSS. Here ALG is an algorithm for structured instances. We assume w.l.o.g. that $\eps>0$ is a small enough constant, so that $2/\eps$ is lower bounded by a sufficiently large constant. $G$ is assumed to be 2EC.}
\label{fig:RED}

\small
\If{$|V(G)|\leq \frac{2}{\epsilon}$} {
compute $\OPT(G)$ by brute force (in constant time) and \Return $\OPT(G)$\label{step:bruteForce}
}
\If{$G$ has a 1-vertex-cut $v$} {\label{step:checkCutNode}
let $(V_1,V_2)$, $V_1\neq \emptyset\neq V_2$, be a partition of $V\setminus \{v\}$ such that there are no edges between $V_1$ and $V_2$. \Return $\RED(G[V_1\cup \{v\}])\cup \RED(G[V_2\cup \{v\}])$\label{step:cutNode}
}

\If{$G$ contains a self loop or a parallel edge $e$} {
\Return $\mathsf{RED}(G\setminus \{e\})$\label{step:simple}
}

\If{$G$ contains an $\alpha$-contractible subgraph $H$ with $|V(H)|\leq 1/\eps$} {\label{step:checkContractible}
\Return $H\cup \RED(G \vert H)$.\label{step:contraction}\\
}

\If{$G$ contains an irrelevant edge $e$}{
\Return $\RED(G\setminus \{e\})$\label{step:irrelevant}
}

\If{$G$ contains a non-isolating 2-vertex-cut $\{u,v\}$} {\label{step:nonIsolating}

let $(V_1,V_2)$, $2\leq |V_1|\leq |V_2|$ be a partition of $V\setminus \{u,v\}$ as in Lemma \ref{lem:nicePartition}\\ 
let $G_i:=G[V_i\cup \{u,v\}]$ for $i\in \{1,2\}$

\If{$|V_1|>\frac{1}{\eps}-2$} {\label{step:V1large}
let $G'_i$ for $i\in \{1,2\}$ be obtained from $G_i$ by contracting $\{u,v\}$ into one node \label{step:contraction_uv}\\ 
let $H'_i=\RED(G'_i)$ for $i\in \{1,2\}$\label{step:definitionH'}\\
let $F'\subseteq E$, $|F'|\leq 2$, such that $H':=H'_1\cup H'_2\cup F'$ is 2EC. \Return $H'$ \label{step:definitionF'} 
}
\Else{
let $\OPT_{1A}$, $\OPT_{1B}$, and $\OPT_{1C}$ be minimum-size subgraphs of $G_1$ of type $A$, $B$ and $C$ \fab{belonging to some 2EC spanning subgraph}, resp., \fab{if any such subgraph exists}. Let $\OPT_{1min}$ be the subgraph of type $\OPT_{1X}$ \fab{which is defined} and has minimum-size, where ties are broken by preferring $\OPT_{1A}$, $\OPT_{1B}$ and $\OPT_{1C}$ in this order\label{step:ABCcomputation}\\
\If{\fab{$\OPT_{1C}$ is defined and} $|\OPT_{1C}|\leq |\OPT_{1B}|-1$}{\label{step:OPT1Csmall}
let $G''_2:=(V(G_2),E(G_2)\cup \{uv\})$ where $uv$ is a dummy edge\label{step:dummyedge}\\  
let $H''_2:=\mathsf{RED}(G''_2)\setminus \{uv\}$\label{step:definitionH''}\\
let $F''\subseteq E$, $|F''|\leq 1$, such that $H'':=\OPT_{1min}\cup H''_2\cup F''$ is 2EC. \Return $H''$\label{step:definitionF''}
}
\Else{\label{step:OPT1Clarge}
let $G'''_2:= (V(G_2)\cup \{w\},E(G_2)\cup \{wu,wv\})$ where $w$ is a dummy node and $wu$ and $wv$ are dummy edges\label{step:dummynode}\\
let $H'''_2:=\RED(G'''_2)\setminus \{wu,wv\}$. \Return $H''':=\OPT_{1min}\cup H'''_2$
}
}
}
\Return $\ALG(G)$\label{step:ALG}
\end{algorithm}

\begin{lemma}\label{lem:runningTime}
$\RED$ runs in polynomial time in $|V(G)|$ if $\ALG$ does so.
\end{lemma}
\begin{proof}
Let $n=|V(G)|$ and $m=|E(G)|$. Define $S=n^2+m^2$ as the \emph{size} of the problem. Clearly every step of the algorithm can be performed in polynomial time $poly(S)$ (which is polynomial in $n$) excluding the recursive calls. In particular, it is possible to check in polynomial time whether $G$ contains and $\alpha$-contractible subgraph with at most $1/\eps$ nodes: simply enumerate over all the subsets $W$ of nodes of the desired size; for each such $W$ one can compute in polynomial time a minimum-size 2EC spanning subgraph $\OPT'$ of \fab{$G[W]$}, if any \fab{such subgraph exists}. Furthermore, one can compute in polynomial time a maximum cardinality subset of edges $F$ with endpoints in $W$ such that $G\setminus F$ is 2EC. Then it is sufficient to compare $|\OPT'|$, \fab{in case $\OPT'$ is defined}, with $|E(G[W])\setminus F|$.

Let $T(S)$ be the running time of the algorithm as a function of $S$. 
Observe that in each step from a problem of size $S$ we generate either a single subproblem of size $S'<S$ or two subproblems of size $S'$ and $S''$ with $S',S''<S$ and $S'+S''\leq S$. Thus $T(S)$ satisfies $T(S)\leq \max\{T(S'),T(S')+T(S'')\}+poly(S)$ with $S',S''<S$ and $S'+S''\leq S$. In particular $S\cdot poly(S)$ is a valid upper bound for $T(S)$.
\end{proof}

It is also not hard to prove that $RED$ returns a feasible solution.
\begin{lemma}\label{lem:feasibility}
If $G$ is 2EC, $\RED$ returns a 2EC spanning subgraph for $G$.
\end{lemma}
\begin{proof}
Let us prove the claim inductively for increasing values of the vector $(|V(G)|,|E(G)|)$ (sorted lexicographically). The base cases are given by Lines \ref{step:bruteForce} and \ref{step:ALG}, where the claim holds trivially and by the correctness of $\ALG$, resp.  If Line \ref{step:cutNode} is executed, both $G_1$ and $G_2$ (which must be 2EC since $G$ is so) contain at most $|V(G)|-1$ nodes, hence by inductive hypothesis $\RED(G_1)$ and $\RED(G_2)$ are 2EC spanning subgraphs of $V_1\cup\{v\}$ and $V_2\cup\{v\}$, resp. Thus their union is a 2EC spanning subgraph for $G$. If Line \ref{step:simple} is executed, $G':=G\setminus \{e\}$ must be 2EC by Lemma \ref{lem:simple}, and $(|V(G')|,|E(G')|)=(|V(G)|,|E(G)|-1)$. Hence by inductive hypothesis $\RED(G')$ is a 2EC spanning subgraph for $G'$, hence for $G$. The case when Line \ref{step:irrelevant} is executed can be handled similarly as in the previous case, using Lemma \ref{lem:irrelevant}. When Line \ref{step:contraction} is executed, the claim follows from Fact \ref{fact:decontraction}.

It remains to consider the case when the condition of Line \ref{step:nonIsolating} holds. We remark that in this case $G$ is 2VC since the condition of Line \ref{step:checkCutNode} does not hold. Hence Lemma \ref{lem:ABC} applies. We also notice that $G'_1$, $G'_2$, $G''_2$, and $G'''_2$ all have strictly fewer nodes than $G$ (since $|V_1|\geq 2$), which will be needed to apply the inductive hypothesis.

Suppose first that the condition of Line \ref{step:V1large} holds. Notice that $G'_1$ and $G'_2$ are 2EC by Fact \ref{fact:contraction} and, implicitly, by Lemma \ref{lem:cutNode} (since the contracted node becomes a 1-vertex-cut). Hence the inductive hypothesis holds for these graphs. Let us argue that a proper $F'$ exists. Notice that each $H'_i$ must be of type A, B or C (otherwise $\RED(G'_i)$ would not be 2EC). If one of them is of type $A$ or they are both of type $B$, then $H'_1\cup H'_2$ is 2EC. If one of them is of type $B$, say $H'_1$, and the other of type $C$, say $H'_2$, then let $H'_2(u)$ and $H'_2(v)$ be the two 2EC components of $H'_2$ containing $u$ and $v$ resp. 
By Lemma \ref{lem:ABC}.\ref{lem:ABC:ABexists} there must exist an edge $f\in G_2$ between $H'_2(u)$ and $H'_2(v)$. Hence $H'_1\cup H'_2\cup \{f\}$ is 2EC. The remaining case is that $H'_1$ and $H'_2$ are both of type C. In this case $H'_1\cup H'_2$ consists of precisely two 2EC components $H'(u)$ and $H'(v)$ containing $u$ and $v$, resp. Since $G$ is 2EC, there must exist two edges $f$ and $g$ between $H'(u)$ and $H'(v)$. Thus $H'_1\cup H'_2\cup \{f,g\}$ is 2EC. \fab{In all the cases the desired $F'$ exists.}

Suppose next that the condition of Line \ref{step:OPT1Csmall} holds. Observe that $G_2$ must contain a subgraph of type A or B by Lemma \ref{lem:ABC}.\ref{lem:ABC:ABexists}. As a consequence $G''_2$ is 2EC as required. Let us argue that a proper $F''$ exists.
Notice that $H''_2$ must be of type A or B. By the \amEdit{tie-}breaking rule of Line \ref{step:ABCcomputation}, $\OPT_{1min}$ must be of type A or C. In the first case $\OPT_{1min}\cup H''_2$ is 2EC. In the second case by Lemma \ref{lem:ABC}.\ref{lem:ABC:ABexists} there exists one edge $f$ such that $\OPT_{1min}\cup \{f\}$ is of type B for $G_1$, implying that $\OPT_{1min}\cup H''_2\cup \{f\}$ is 2EC.

It remains to consider the case when the condition of Line \ref{step:OPT1Clarge} holds. Again by Lemma \ref{lem:ABC}.\ref{lem:ABC:ABexists}, $G_2$ must contain a subgraph of type A or B, hence $G'''_2$ is 2EC as required. Notice that $\RED(G'''_2)$ must contain the two dummy edges $wv$ and $wu$, which induce a type B graph. Hence by Lemma \ref{lem:ABC}.\ref{lem:ABC:feasible}, $H'''_2$ must be of type A or B. Notice also that by the \amEdit{tie-}breaking rule of Line \ref{step:ABCcomputation} and given that the condition of Line \ref{step:OPT1Csmall} does not hold, $\OPT_{1min}$ must be of type A or B. Hence $\OPT_{1min}\cup H'''_2$ is 2EC as desired.
\end{proof}

It remains to upper bound the approximation factor of the algorithm. Let $\Red(G)=|\RED(G)|$. The following claim implies the approximation factor of $\RED$ since trivially $|V(G)| \leq \opt(G)$.
\begin{lemma}\label{lem:approximationFactor}
One has $\Red(G) \leq 
\begin{cases}
\opt(G) & \text{if } |V(G)| \leq \frac {2}\epsilon;\\
\alpha \cdot \opt(G) + 2\epsilon\cdot |V(G)| - 2 & \text{if } |V(G)| > \frac {2}\epsilon.
\end{cases}
$.
\end{lemma} 
\begin{proof}
Let us prove the claim by induction on $(|V(G)|,|E(G)|)$ similarly to Lemma \ref{lem:feasibility}. The base case is when Line \ref{step:bruteForce} is executed: here the claim trivially holds. For the inductive hypothesis, we distinguish a few cases depending on which line of $\mathsf{RED}$ is executed. If Line \ref{step:simple} is executed we consider a graph $G':=G\setminus \{e\}$ with the same number of nodes as $V(G)$ and the same optimal size by Lemma \ref{lem:simple}. The claim follows by inductive hypothesis. If Line \ref{step:irrelevant} is executed, the claim follows similarly using Lemma \ref{lem:irrelevant}. In Line \ref{step:ALG} is executed, one has $\Red(G)\leq \alpha \opt(G)\leq \alpha \opt(G)+2\eps |V(G)|-2$ since $\ALG$ is $\alpha$ approximate and $|V(G)|>\frac{2}{\eps}$. 

If Line \ref{step:cutNode} is executed, recall that by Lemma \ref{lem:cutNode} one has $\opt(G)=\opt(G_1)+\opt(G_2)$. 
We distinguish a few cases depending on the sizes of $V_1$ and $V_2$.
Assume w..l.o.g. $|V_1|\leq |V_2|$. If $|V_2|\leq \frac{2}{\eps}-1$ then $\Red(G)=\opt(G_1)+\opt(G_2)=\opt(G)\leq \alpha \opt(G)+2\eps |V(G)|-2$, where the last inequality follows from $\alpha\geq 1$ and $|V(G)|\geq \frac{2}{\eps}$. Otherwise, if $|V_1|\leq \frac{2}{\eps}-1$, by inductive hypothesis one has $\Red(G)\leq \opt(G_1)+\alpha \opt(G_2)+2\eps |V(G_2)|-2\leq \alpha \opt(G)+2\eps |V(G)|-2$, where we used $\alpha\geq 1$ and $|V(G_2)|\leq |V(G)|$. The remaining case is $|V_1|\geq \frac{2}{\eps}$. In this case $\Red(G)\leq \alpha \opt(G_1)+\alpha \opt(G_2)+2\eps (|V(G_1)|+|V(G_2)|)-4= \alpha \opt(G)+2\eps |V(G)|+2\eps-4\leq \alpha\opt(G)+2\eps |V(G)|-2$, where we used $2\eps\leq 2$. 

If Line \ref{step:contraction} is executed, one has $\opt(G)= |\OPT(G)\cap {{V(H)}\choose{2}}| +|\OPT(G)\setminus {{V(H)}\choose{2}}|\geq \frac{1}{\alpha}|H|+\opt(G\vert H)$ by the definition of $H$ and observing that $\OPT(G)\setminus {{V(H)}\choose{2}}$ induces a feasible 2EC spanning subgraph of $G\vert H$. If $|V(G\vert H)|\leq \frac{2}{\eps}$ one has $\Red(G)=|H|+\opt(G\vert H)\leq \alpha \opt(G)\leq \alpha \opt(G)+2\eps|V(G)|-2$ since $|V(G)|\geq 2/\eps$. Otherwise $\Red(G)\leq |H|+\alpha \opt(G\vert H)+2\eps|V(G\vert H)|-2\leq \alpha \opt(G)+2\eps|V(G)|-2$ since $|V(G)|\geq |V(G \vert H)|$.

We next assume that the condition of Line \ref{step:nonIsolating} holds. Let $\OPT_i=\OPT(G)\cap E(G_i)$ and $\opt_i=|\OPT_i|$. Assume first that the condition of Line \ref{step:V1large} holds. Observe that $\OPT_i$ induces a feasible solution for $G'_i$. Thus by inductive hypothesis 
$$
|H'|=|H'_1|+|H'_2|+|F|\leq \alpha (\opt_1+\opt_2)+2\eps(|V_1|+|V_2|+2)-4+2= \alpha \opt(G)+2\eps|V(G)|-2 .
$$ 

We next assume that one of the conditions of Lines \ref{step:OPT1Csmall} or \ref{step:OPT1Clarge} holds. 
Let $\OPT_{iX}$ be a minimum-size subgraph of type $X\in \{A,B,C\}$ in $G_i$ and set $\opt_{iX}=|\OPT_{iX}|$ (is there is no such subgraph, $\opt_{iX}=+\infty$). We will next assume w.l.o.g. that $\OPT_1$ is of  type $B$ or $C$ (hence $\OPT_2$ of type $A$ or $B$ by Lemma \ref{lem:ABC}.\ref{lem:ABC:feasible}) in view of the following claim.
\begin{claim}\label{claim:OPT1typeBC}
There exists and optimal solution $\OPT(G)$ such that $\OPT_1$ is of type $B$ or $C$.
\end{claim}
\begin{proof}
Let us assume that $\OPT_1$ is of type A, otherwise the claim holds. Notice that every feasible solution must use at least $\min\{\opt_{1A},\opt_{1B},\opt_{1C}\}$ edges from $G_1$ by Lemma \ref{lem:ABC}.\ref{lem:ABC:feasible}. Hence $\opt_{1A}> \alpha\cdot \min\{\opt_{1B},\opt_{1C}\}$ since otherwise $\OPT_{1A}$ would be an $\alpha$-contractible subgraph on at most $1/\eps$ nodes (which is excluded by the check of Line \ref{step:checkContractible}). If $\opt_{1A}\geq \opt_{1B}+1$, then we obtain an alternative optimum solution with the desired property by taking $\OPT_{1B}\cup \OPT_2$, and, in case $\OPT_2$ is of type C, by adding one edge $f$ between the connected components of $\OPT_2$ whose existence is guaranteed by Lemma \ref{lem:ABC}.\ref{lem:ABC:ABexists}. Otherwise ($\opt_{1A}\leq \opt_{1B}$), we must have $\opt_{1A}>\alpha \opt_{1C}$, hence in particular $\opt_{1A}\geq \opt_{1C}+1 \geq 4$ since $\alpha\geq 1$ and $\opt_{1C}\geq 3$. Observe however that Lemma \ref{lem:ABC}.\ref{lem:ABC:ABexists} implies $\opt_{1B}\leq \opt_{1C}+1$. Altogether $\opt_{1A}=\opt_{1B}=\opt_{1C}+1$. Then $\opt_{1A}> \alpha\opt_{1C}=\alpha(\opt_{1A}-1)$ implies $\opt_{1A}\leq 5$  since $\alpha\geq 6/5$. Altogether $\opt_{1A} \in \{4,5\}$.

If $\opt_{1A}=4$, $\OPT_{1A}$ has to be a $4$-cycle $C$ of type $u,a,v,b,u$. Notice that the edge $ab$ cannot exist since otherwise $\opt_{1B}=3$ due to the path $u,a,b,v$. Then every feasible solution must include the $4$ edges of $C$ (to guarantee degree at least $2$ on $a$ and $b$), which makes $C$ an $\alpha$-contractible subgraph on $4\leq 1/\eps$ nodes: this is excluded by the check of Line \ref{step:checkContractible}.

The other possible case is $\opt_{1A}=5$. The minimality of $\OPT_{1A}$ implies that $\OPT_{1A}$ is a $5$-cycle $C$: w.l.o.g. let $C=u,a,v,b,c,u$. Observe that the edges $ab$ cannot exists, since otherwise $\{va,ab,bc,cu\}$ would be a  type $B$ graph in $G_1$ of size $4$, contradicting $\opt_{1B}=\opt_{1A}=5$. A symmetric construction shows that edge $ac$ cannot exist. Hence every feasible solution restricted to $G_1$ must include the edges $\{au,av\}$ and furthermore at least $3$ more edges incident to $b$ and $c$, so at least $5$ edges altogether. This implies that $C$ is an $\alpha$-contractible subgraph of size $5\leq 1/\eps$: this is excluded by the check of Line \ref{step:checkContractible}.
\end{proof}

Claim \ref{claim:OPT1typeBC} implies $\opt_2\geq \min\{\opt_{2A},\opt_{2B}\}$ and $\opt_1\geq \min\{\opt_{1B},\opt_{1C}\}$. 
By the check of Line \ref{step:checkContractible} we also have that 
$
\opt_{1A}>\frac{1}{\alpha}\opt_1\geq \frac{1}{\alpha}\min\{\opt_{1B},\opt_{1C}\} 
$, 
hence in particular $\OPT_{1min}$ is of type $B$ or $C$ in this case. More precisely, by the tie-breaking rule, $\OPT_{1min}$ is of type C if the condition of Line \ref{step:OPT1Csmall} holds, and of type B otherwise (i.e. if the condition of Line \ref{step:OPT1Clarge} holds).

Suppose that the condition of Line \ref{step:OPT1Csmall} holds. In this case $\OPT_{1min}$ is of type B by the tie-breaking rule and the previous arguments, hence $|H_{1min}|=\opt_{1B}\leq \opt_1$. Notice also that $\OPT_{2A}$ and $\OPT_{2B}$ plus the two dummy edges are feasible solutions for $G'_2$, thus $\opt'_2\leq \min\{\opt_{2A},\opt_{2B}\}+2\leq \opt_2+2$. Thus
\begin{align*}
|H'''| & =|\OPT_{1min}|+|H'''_2|-2\leq \opt_1+\alpha(\opt_2+2)+2\eps(|V_2|+3)-2-2\\
& \leq \alpha\opt +2\eps |V(G)|-2+(2\alpha-2-(\alpha-1)\opt_1)\leq \alpha\opt +2\eps |V(G)|-2,
\end{align*}
where in the last inequality we used the obvious fact that $\opt_1\geq 2$.

It remains to consider the case that the condition of Line \ref{step:OPT1Clarge} holds.
 Recall that in this case $\OPT_{1min}$ is of type $C$. Remember also that $\OPT_2$ must be of type A or B. Suppose first that $\OPT_2$ is of type $A$. Then $\OPT_1$ has to be of type $C$ (otherwise $\OPT(G)$ would not be optimal), implying $\opt=\opt_{1C}+\opt_{2A}$. Notice that $\OPT_2$ is a feasible solution for $G''_2$, hence $|\OPT(G''_2)|\leq \opt_{2A}$. Then
$$
|H''|\leq \opt_{1min}+\alpha \opt_{2A}+2\eps (|V_2|+2)-2-1+1\leq \alpha \opt +2\eps |V(G)|-2.
$$ 
Otherwise, i.e. $\OPT_2$ is of type $B$, $\OPT_1$ must be of type $A$ or $B$, hence of type $B$ by Claim \ref{claim:OPT1typeBC}. Hence $\opt=\opt_{1B}+\opt_{2B}$. Notice that $\OPT_2\cup \{uv\}$ is a feasible solution for $G''_2$, hence $|\OPT(G''_2)|\leq \opt_{2B}+1$. Since $\opt_{1min}= \opt_{1C}\leq \opt_{1B}-1$, one has
\begin{align*}
|H''| & \leq \opt_{1B}-1+\alpha (\opt_{2B}+1)+2\eps (|V_2|+2)-2-1+1\\
& \leq \alpha \opt + 2\eps |V(G)|-2-(\alpha-1)(\opt_{1B}-1)\leq \alpha \opt + 2\eps |V(G)|-2.
\end{align*}
\end{proof}

\section{Finding a Canonical 2-Edge-Cover}\label{sec:canonical}

In this section we prove Lemma \ref{lem:canonical}, \amEdit{which} we restate here together with the definition of canonical 2-edge-cover. 

\defCanonical*

\lemCanonical*

\begin{proof}
Starting with $H'=H$, while $H'$ is not canonical we perform one of the following updates, in the specified order of priority: 
%
%
%

\medskip\noindent {\bf (1)} If there exists $e\in H'$ such that $H'\setminus \{e\}$ is a 2-edge-cover, replace $H'$ with $H'\setminus \{e\}$. 

\medskip\noindent {\bf (2)} If there exist $D\subseteq H'$ and $A\subseteq E(G)\setminus H'$, with $|D|=|A|\leq 3$, such that $H'\setminus D\cup A$ is a  2-edge-cover with fewer components than $H'$, replace $H'$ with $H'\setminus D\cup A$. 

\medskip\noindent {\bf (3)} Suppose that there exists a 2EC component $C$ of $H'$ with fewer than $7$ edges that is not an $|V(C)|$-cycle. Notice that $|E(C)|\geq |V(C)|$ since otherwise $C$ would not be 2EC. This implies $|V(C)|\leq 5$. However $|V(C)|\geq 5$ since case (1) does not apply. Thus $C$ contains precisely $5$ nodes, it is 2EC, \amEdit{minimal,} and not a 5-cycle. The only possibilities are that $C$ is a bowtie (i.e., two triangles that share a common node) or a $K_{2,3}$ (i.e a complete bipartite graph with a side of size 2 and the other of size 3). Depending on which case applies, we perform one of the following steps:

\medskip\noindent {\bf (3.a)} $C$ is a bowtie (see Figure \ref{fig:CanonicalProofCase3}\textcolor{red}{.a}). Let $\{v_1,v_2,u\}$ and $\{v_3,v_4,u\}$ be the two triangles in the bowtie. Assume that there exists some edge in $E(G)\setminus E(H')$ between the nodes of the bowtie, say $v_1v_3$ w.l.o.g. In this case we replace $C$ with the cycle $v_1,v_3,v_4,u,v_2,v_1$. In other words, we  replace $H'$ with $H'\setminus E(C)\cup \{v_1v_3,v_3v_4,v_4u,uv_2,v_2v_1\}$. Notice that the number of edges of $H'$ decreases.
Otherwise, since $u$ is not a 1-vertex-cut of $G$ (as $G$ is structured), there must exist an edge $zw\in E(G)\setminus E(H')$ with $z\in V(C)\setminus \{u\}$ and $w\in V(G)\setminus V(C)$. This is however excluded by case (2) since $H'\setminus \{uz\}\cup \{zw\}$ would be a 2-edge-cover with the same size as $H'$ but with fewer components. 

\begin{figure}
\begin{center}
\begin{tikzpicture}[scale=1.5]


\tikzset{vertex/.style={draw=black, very thick, circle,minimum size=0pt, inner sep=2pt, outer sep=2pt}
}



\node () at (-3.8,1.5) {(a)};
\node () at (0.1,1.5) {(b)};

\node () at (-2,0){$v_1$};
\node () at (-2,1.2){$v_2$};
\node () at (-0.4,0){$v_3$};
\node () at (-0.4,1.2){$v_4$};
\node () at (-1.2,0.8){$u$};
\node () at (-3,-0.2){$w$};
\node () at (1.6,0.45){$u$};
\node () at (2.2,-0.4){$w_1$};
\node () at (2.15,0.4){$w_2$};
\node () at (2.15,1.2){$w_3$};
\node () at (3.62,0){$v_1$};
\node () at (3.62,0.8){$v_2$};

\begin{scope}[every node/.style={vertex}]

\node (1a) at (-3,0) {};
\node (2a) at (-3.5,0.5){};
\node (3a) at (-3.5,-0.5){};
\draw[ultra thick] (1a) to (2a) to (3a) to (1a);

\node (1b) at (-2,0.2) {};
\node (2b) at (-2,1) {};
\node (3b) at (-1.2,0.6){};
\node (4b) at (-0.4,0.2){};
\node (5b) at (-0.4,1){};
\draw[ultra thick] (3b) to (1b) to (2b) to (3b) to (4b) to (5b) to (3b);

\draw[dashed] (1a) to (1b);
\draw[dashed] (1b) to (4b);

\node (1c) at (1,0.8) {};
\node (3c) at (1,-0.2) {};
\node (2c) at (0.5,0.3){};
\node (4c) at (1.5,0.3){};

\draw[ultra thick] (1c) to (2c) to (3c) to (4c) to (1c);

\node (1da) at (2.4,-0.3){};
\node (2da) at (2.4,0.4){};
\node (3da) at (2.4,1.1){};
\node (1db) at (3.4,0){};
\node (2db) at (3.4,0.8){};
\draw[ultra thick] (1db) to (1da) to (2db) to (2da) to (1db) to (3da) to (2db);
\draw[dashed] (4c) to (1da) to (2da);

\end{scope}
\begin{scope}[very thick]
\end{scope}
\begin{scope}[very thick, densely dashed]
\end{scope}
\end{tikzpicture}
\end{center}
\caption{(a) Illustration of case (3.a) of Lemma \ref{lem:canonical}. By adding $v_1v_3$ and removing $\{v_1u,v_3u\}$ we obtain a smaller 2-edge-cover. Alternatively, by adding $v_1w$ and removing $v_1u$ we obtain a 2-edge-cover of the same size with fewer components. (b) Illustration of case (3.b) of Lemma \ref{lem:canonical}. By adding $w_1w_2$ and removing $\{v_2w_1,v_1w_2\}$ we obtain a smaller 2-edge-cover. Alternatively, by adding $w_1u$ and removing $w_1v_2$ we obtain a 2-edge-cover of the same size with fewer components.
}
\label{fig:CanonicalProofCase3}
\end{figure}
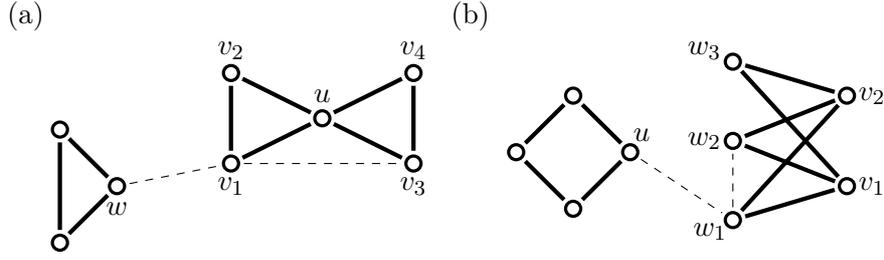 

\medskip\noindent {\bf (3.b)} $C$ is a $K_{2,3}$ (see Figure \ref{fig:CanonicalProofCase3}\textcolor{red}{.b}). Let $\{v_1, v_2\}$ and $\{w_1, w_2, w_3\}$ be the two sides of $C$. Suppose that there exists an edge between the nodes $\{w_1, w_2, w_3\}$, say $w_1w_2$. In this case we replace $C$ with the 5-cycle $w_1,w_2,v_1,w_3,v_2,v_1$, namely we replace $H'$ with $H'\setminus E(C)\cup \{w_1w_2,w_2v_1,v_1w_3,w_3v_2,v_2v_1\}$. Notice that the number of edges of $H'$ decreases. 

Suppose next that the nodes $\{w_1, w_2, w_3\}$ have degree precisely $2$ in $G$. Then every feasible 2EC spanning subgraph contains all the edges of $C$, making $C$ a $5/4$-contractible subgraph on $5$ nodes, which contradicts the assumption that $G$ is structured. 

The only remaining case is that there exists an edge $f\in E(G)\setminus E(H')$ between $\{w_1, w_2, w_3\}$ and $V(G)\setminus V(C)$, w.l.o.g. $f=w_1u$. This is however excluded by case (2) since $H'\setminus \{w_1v_1\}\cup \{w_1u\}$ would be a 2-edge-cover with fewer components. 

\medskip\noindent {\bf (4)} There exists a block $B$ of $H'$ which contains at most $5$ edges and such that all the edges of $H'$ with exactly one endpoint in $V(B)$ are incident to the same node $v_1$ of  $B$. We remark that this includes leaf blocks as a special case. Notice that by Case (1) $B$ must be an $\ell$-cycle $v_1,\ldots,v_\ell,v_1$ for $3\leq \ell\leq 5$.  Since $G$ is structured, by the 3-matching Lemma~\ref{lem:matchingOfSize3} there is a matching $M$ of size $3$ between $V(B)$ and $V\setminus V(B)$ in $G$. We distinguish two subcases:  

\medskip\noindent {\bf (4.a)} There exists an edge $zw\in E(G)$ with $z\in V(B)$ adjacent to $v_1$ in $B$, say $z=v_2$, and $w\notin V(B)$ (see Figure \ref{fig:CanonicalProofCase4}\textcolor{red}{.a}). In this case we replace $H'$ with $H'\setminus \{v_1v_2\}\cup \{v_2w\}$. Notice that $w$ must be in the same connected component of $B$ by case (2). Furthermore, the number of bridges of $H'$ decreases since some original bridge incident to $v_1$ is not a bridge any longer and no new bridge is created. 

\begin{figure}
\begin{center}
\begin{tikzpicture}[scale=1.5]

\tikzset{vertex/.style={draw=black, very thick, circle,minimum size=0pt, inner sep=2pt, outer sep=2pt}
}

\node () at (-4,1.5) {(a)};
\node () at (-0.2,1.5) {(b)};

\node () at (0.2,1.29) {$v_4$};
\node () at (0.9,1.29) {$v_5$};
\node () at (0.15,0.1){$z=v_3$};
\node () at (0.9,0.1){$v_2$};
\node () at (1.35,0.8){$v_1$};

\node () at (-2.25,0.45) {$e$};
\node () at (1.25,-0.3){$e$};
\node () at (-3.5,0.67){$z=v_2$};
\node () at (-3.5,-0.68){$v_3$};
\node () at (-2.82,0.15){$v_1$};
\node () at (-1.5,-0.15){$w$};
\node () at (2.6,-0.36){$w$};

\begin{scope}[every node/.style={vertex}]

\node[fill=gray] (1a) at (-3,0) {};
\node[fill=gray] (2a) at (-3.5,0.5){};
\node[fill=gray] (3a) at (-3.5,-0.5){};
\draw[ultra thick] (1a) to (2a) to (3a) to (1a);

\node (l) at (-2.3,0){};

\node (1b) at (-1.5,0) {};
\node (2b) at (-1.5,0.9) {};
\node (3b) at (-0.6,0.9){};
\node (4b) at (-0.6,0){};

\draw[ultra thick] (1b) to (2b) to (3b) to (4b) to (1b);

\draw[ultra thick] (1a) to (l) to (1b);
\draw[dashed,bend left=10] (2a) to (1b);


\node[fill=gray] (1c) at (1.2,0.7) {};
\node[fill=gray] (2c) at (0.9,0.3){};
\node[fill=gray] (3c) at (0.2,0.3) {};
\node[fill=gray] (4c) at (0.2,1.1){};
\node[fill=gray] (5c) at (0.9,1.1){};

\draw[ultra thick] (1c) to (2c) to (3c) to (4c) to (5c) to (1c);
\draw[red] (2c) to (5c);

\node (1d) at (2.6,-0.2){};
\node (2d) at (2.2,0.6){};
\node (3d) at (3,0.6){};

\draw[ultra thick] (1d) to (2d) to (3d) to (1d);
\draw[ultra thick] (2d) to (1c);
\draw[dashed,bend right=18] (3c) to (1d);

\end{scope}
\begin{scope}[very thick]
\end{scope}
\begin{scope}[very thick, densely dashed]
\end{scope}
\end{tikzpicture}
\end{center}
\caption{(a) Illustration of case (4.a) of Lemma \ref{lem:canonical}. By adding $v_2w$ and removing $v_2v_1$ we obtain a 2-edge-cover with fewer bridges. (b) Illustration of case (4.b) of Lemma \ref{lem:canonical}. Nodes $v_2$ and $v_5$ are adjacent only to nodes in the 5-cycle $B=v_1,v_2,v_3,v_4,v_5,v_1$. Furthermore edge $v_2v_5$ must exist since otherwise $B$ would be a $5/4$-contractible subgraph. By adding $\{v_2v_5,v_3w\}$ and removing $\{v_2v_3,v_1v_5\}$ we obtain a 2-edge-cover with fewer bridges.}
\label{fig:CanonicalProofCase4}
\end{figure}
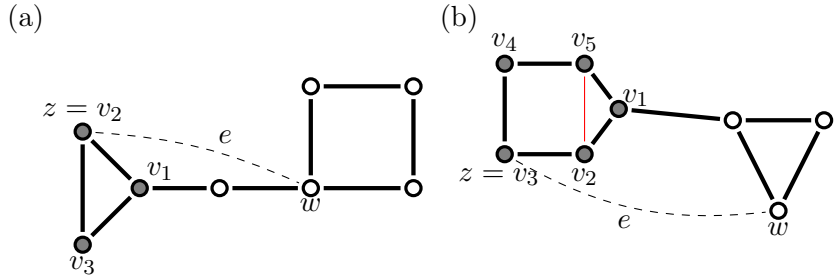

\medskip\noindent {\bf (4.b)} $v_2$ and $v_\ell$ are adjacent only to nodes in $V(B)$ in $G$ (see Figure \ref{fig:CanonicalProofCase4}\textcolor{red}{.b}). Since we have the matching $M$, then $\ell=5$ in this case. Furthermore, we can assume w.l.o.g. that an edge $zw\in E(G)$, with $z\in \{v_3,v_4\}$ and $w\notin V(B)$, exists (otherwise $v_1$ would be a 1-vertex-cut). Assume by contradiction that $v_2v_5\notin E(G)$. In this case any feasible solution would contain $4$ distinct edges (incident to $v_2$ and $v_5$) with both endpoints in $V(B)$: this makes $B$ a $5/4$-contractible subgraph on $5$ nodes, contradicting the assumption that $G$ is structured. Thus we can assume $v_2v_5\in E(G)$. In this case replace $H'$ with $H'\setminus \{v_1v_5,v_2v_3\}\cup \{zw,v_2v_5\}$. As in Case (4.a), the number of connected components of $H'$ decreases or the number of bridges of $H'$ decreases and its number of connected components does not change.

In Steps (1) and (3) the number of edges decreases strictly, while in Steps (2) and (4) it remains unchanged. Furthermore in Steps (2) and (4) either the number of connected components decreases or the number of bridges decreases and the number of connected components does not change. It is \amEdit{then} easy to check that the process ends within a polynomial number of steps (where each step can be executed in polynomial time).

At the end of the process clearly $H'$ satisfies properties (i) (by case (3)) and (iii) (by case (2)). By case (4), leaf blocks contain at least $6$ edges. The same holds for inner blocks $B$ such that all the bridge edges incident to $B$ are incident to the same node $v\in V(B)$. Consider any remaining inner block $B$ with $3$ edges, in particular $B$ is a triangle $v_1,v_2,v_3,v_1$. By assumption there are at least two bridge edges incident to distinct nodes of $B$, say $wv_1$ and $zv_2$. However in this case we would apply Step (1) and remove the edge $v_1v_2$, a contradiction. 
\end{proof}

\section{Case of Many Triangles}\label{sec:manyTriangles}

In this section we prove Lemma \ref{lem:caseB:gluingStep}, \amEdit{which} we next restate together with the definition of merging cycle for \amEdit{the} reader's convenience. We recall that $S$ is the current infeasible solution, whose connected components are 2EC. 

\lemmaCaseBgluingStep*
\definitionMergingCycle*

Recall that we build an auxiliary forest $G'=(\calC(S),F)$ on the components $\calC(S)$ of $S$, where initially $F=\emptyset$. The main tool to add edges to $F$ is the following lemma, whose proof is postponed.

\begin{lemma}\label{lem:caseB:mergingCycle}
If $G'$ consists of at least $2$ trees, in polynomial time one can compute a merging cycle $M$ incident to the components of at least $2$ distinct trees of $G'$.
\end{lemma}
If the merging cycle $M$ obtained via Lemma \ref{lem:caseB:mergingCycle} is cheap, we \fab{derive the desired $S'$ from $S$ and $M$ as described earlier.}. Otherwise, we add to $F$ the edge between the two 2EC components of $S$ incident to $M$. We recall that $M$ is incident to one light and one heavy component. 
Clearly the process ends in a polynomial number of steps by returning the desired cheap merging cycle or by obtaining a single tree $G'$. 

In the latter case, we check whether there exists an edge $e$ in $G$ between two connected components of $S$ which are \emph{not} adjacent in $G'$. If such an $e$ exists, we can find a cheap merging cycle via the following lemma. 
\begin{lemma}\label{lem:caseB:nonTreeEdge}
Suppose that $G'$ consists of a single tree and there exists an edge $e_0$ in $G$ with endpoints in two distinct 2EC components  of $S$ which are not adjacent in $G'$. Then in polynomial time one can compute a cheap merging cycle of $S$.   
\end{lemma}
\begin{proof}
Assume that $e_0$ is incident to the components $C_1$ and $C_k$, and let $C_1,C_2,\ldots,C_k$, $k\geq 3$, be the components along the (unique) path in $G'$ between $C_1$ and $C_k$. We will define edges $e_1,\ldots,e_{k-1}$ such that each $e_i$, $i\geq 1$, has one endpoint in $C_i$ and the other one in $C_{i+1}$, and such that $M:=\{e_0,\ldots,e_{k-1}\}$ is a merging cycle. The claim follows since $k\geq 3$ (hence $M$ has to be cheap).  

Remember that for each pair $C_iC_{i+1}$ there is an expensive merging cycle incident to exactly $C_i$ and $C_{i+1}$ (by the construction of $G'$). In particular there will be two edges $a_i$ and $b_i$ with one endpoint in $C_i$ and the other in $C_{i+1}$. Furthermore, if $C_i$ is light, the endpoints of $a_i$ and $b_i$ in $C_i$ are distinct. By construction we also have that $C_i$ is light and $C_{i+1}$ is heavy or vice versa. 

Given the edges $e_0,\ldots,e_{i}$, $i\in\{0,\ldots,k-2\}$, we choose $e_{i+1}$ as follows. If $C_{i+1}$ is light, we choose any edge in $\{a_{i+1},b_{i+1}\}$ which is not incident to $e_i$ (this edge must exist since $a_{i+1}$ and $b_{i+1}$ are incident to distinct nodes of $C_{i+1}$). Otherwise, i.e. if $C_{i+1}$ is heavy, we set $e_{i+1}$ to any edge in $\{a_{i+1},b_{i+1}\}$ which is not incident to $e_0$ (the latter condition is always satisfied unless $i+1=k-1$).
\end{proof} 
If there is no edge $e_{0}$ as in the previous lemma, then we apply the following lemma to compute a desired $S'$.
\begin{lemma}\label{lem:caseB:CoreTriangle}
Suppose that $G'$ consists of a single tree, all the edges $e$ of $G$ between distinct 2EC components $C_1$ and $C_2$ of $S$ involve two adjacent components in $G'$, and $S$ is not core-triangle. Then in polynomial time one can compute an $S'\subseteq E$ such that: 
(1) the connected components of $S'$ are 2EC, (2) $S'$ has fewer connected components than $S$, and (3) $cost(S')\leq cost(S)$.
\end{lemma}
\begin{proof}
By assumption all the edges of $G$ between different components of $S$ involve adjacent components in $G'$. Furthermore, as remarked before, by construction if $C_1$ and $C_2$ are adjacent in $G'$, then exactly one of them is light and the other is heavy. Notice that there must exist a light component $C=\{v_1,v_2,v_3\}$ of degree at least $2$  in $G'$: indeed otherwise $G'$ would be a star graph whose center is a heavy component and whose leaves are light components (hence in particular triangles). Since by assumption there are no edges between the leaf components of $G'$, this implies that $S$ is core-triangle, contradicting the assumptions. 

Suppose by contradiction that there exists a component $C'$ adjacent to $C$ in $G'$ such that $C'$ is adjacent only to two nodes in $V(C)$, say $v_1$ and $v_2$. Then $\{v_1,v_2\}$ is a 2-vertex-cut and the edge $v_1v_2$ is irrelevant, contradicting the assumption that $G$ is structured. Therefore we can choose any two distinct components $C_{1}$ and $C_{2}$ adjacent to $C$ in $G'$, and two edges $\{a_{1}v_1,b_1v_2\}$ (resp., $\{a_2v_2,b_2v_3\}$) between $C_{1}$ (resp., $C_2$) and $C$. Consider $S':=(S\setminus \{v_1v_2,v_2v_3\})\cup \{a_{1}v_1,b_1v_2,a_2v_2,b_2v_3\}$. Notice that $S'$ satisfies (1) and (2), in particular it contains a 2EC component $C'$ spanning the nodes $V(C)\cup V(C_{1})\cup V(C_{2})$. Furthermore $S'$ satisfies (3) since 
\begin{align*}
& cost(S)-cost(S') =|S|-|S'|+cr(S)-cr(S')\\  = & -2+cr(C_{1})+cr(C_{2})+cr(C)-cr(C')
 = -2+2+2+\frac{1}{2}-2>0.
\end{align*}
\end{proof}

\begin{proof}[Proof of Lemma \ref{lem:caseB:gluingStep}]
The claim follows directly from the above construction and Lemmas \ref{lem:caseB:mergingCycle}, \ref{lem:caseB:nonTreeEdge}, and \ref{lem:caseB:CoreTriangle}. 
\end{proof}


\begin{figure}
\begin{center}
\begin{tikzpicture}[scale=1.5]

\tikzset{vertex/.style={draw=black, very thick, circle,minimum size=0pt, inner sep=2pt, outer sep=2pt}
}

\node[right=0.5pt] () at (0,0.5){$e$};

\begin{scope}[every node/.style={vertex}]

\node[fill=red] (u) at (0,0){};
\node[fill=red] (v) at (0,1){};
\node [fill=red] (1) at (-0.5,1.5){};
\node [fill=red] (2) at (-1,1){};
\draw[ultra thick] (1) to (2) to (v) to (1);

\node[fill=red] (3) at (-2.5,1){};
\node[fill=red] (4) at (-3,0.7){};
\node[fill=red] (5) at (-3,0){};
\node[fill=red] (6) at (-2,0){};
\node[fill=red] (7) at (-2,0.7){};
\draw[ultra thick] (3) to (4) to (5) to (6) to (7) to (3);

\node[fill=blue!50!white] (11) at (2,0){};
\node[fill=blue!50!white] (12) at (1.2,0){};
\node[fill=blue!50!white] (13) at (1.2,0.6){};
\node[fill=blue!50!white] (14) at (2,0.6){};
\draw[ultra thick] (11) to (12) to (13) to (14) to (11);

\node[fill=black!10!white] (8) at (1.2,1.3){};
\node[fill=black!10!white] (9) at (2.2,1.3){};
\node[fill=black!10!white] (10) at (1.7,1.8){};
\draw[ultra thick] (8) to (9) to (10) to (8);

\node[fill=black!50!green] (20) at (3,0.9){};
\node[fill=black!50!green] (21) at (3.5,1.4){};
\node[fill=black!50!green] (22) at (4,0.9){};
\node[fill=black!50!green] (23) at (3.5,0.4){};
\draw[ultra thick] (20) to (21) to (22) to (23) to (20);

\node[fill=violet] (24) at (0.5,1.8){};

\draw[dashed,red] (u) to (v);
\draw [dashed,red] (u) to (2);
\draw [dashed,red] (2) to (6);
\draw [dashed,red] (1) to (3);

\draw[dashed,blue] (1) to (24) to (10);
\draw[dashed,blue] (8) to (14);
\draw[dashed,blue] (12) to (u);

\end{scope}
\begin{scope}[very thick]
\end{scope}

\begin{scope}[very thick, densely dashed]
\end{scope}

\end{tikzpicture}
\end{center}
\caption{The black edges are the current solution $S$. The forest $G'$ has $5$ trees, which are identified by nodes of distinct colors. The red edges correspond to expensive merging cycles identified in previous steps of the construction. Let $\Pi$ the partition of the node set induced by distinct colors. The blue dashed edges form a nice cycle $N$ of $\Pi$. In particular if two edges in $N$ touch a given set $V_i$ in the partition, this happens at distinct nodes of $V_i$ unless $|V_i|=1$. A merging cycle $M$ can be easily derived from $N$ using the red edges. In particular, $M=N\cup \{e\}$ in this case, which is cheap since it involves at least $3$ components of $S$.}
\label{fig:NiceCycle}

\end{figure}
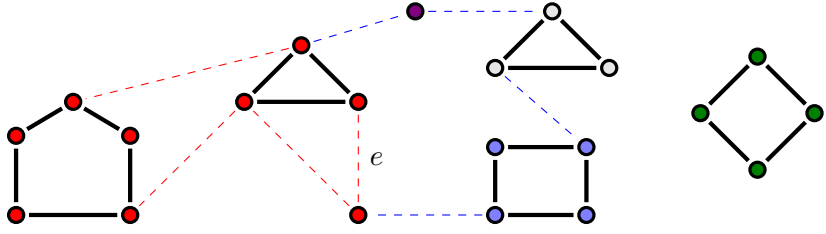

It remains to prove Lemma \ref{lem:caseB:mergingCycle}. It is convenient to introduce the following definition (see Figure \ref{fig:NiceCycle}).
\begin{definition}[Nice Cycle and Path]\label{def:nicePathCycle}
Let $\Pi=(V_1,\ldots,V_k)$, $k\geq 2$, be a partition of the node-set of a graph $G$. A nice cycle (resp. nice path) $N$ of $\Pi$ is a subset of edges with endpoints in distinct subsets of $\Pi$ such that: (1) $N$ induces one cycle of length at least $2$ (resp., one path of length at least $1$) in the graph obtained from $G$ by collapsing each $V_i$ into a single node; (2) given any two edges of $N$ incident to some $V_i$, these edges are incident to distinct nodes of $V_i$ unless $|V_i|=1$.  
\end{definition} 
Notice that, considering the partition $\Pi$ induced by the node sets of the 2EC components of $S$, a nice cycle for $\Pi$ is a merging cycle for $S$. More important for our goals, if we consider the partition $\Pi(G')$ induced by the nodes of the components in each tree of $G'$ (assuming that it contains at least $2$ trees), then we can derive a merging cycle $M$ for $S$ from a nice cycle $N$ for $\Pi(G')$ via the following lemma (see Figure \ref{fig:NiceCycle}).  
\begin{lemma}\label{lem:caseB:NiceToMerging}
Given a nice cycle $N$ of $\Pi(G')=\{V_1,\ldots,V_k\}$, $k\geq 2$, in polynomial time one can compute a merging cycle $M$ of $S$ incident to the components of at least two distinct trees of $G'$. 
\end{lemma}  
\begin{proof}
Let $U_1,\ldots,U_q$ be the subsets in the partition incident to $N$. The set $M$ consists of $N$ plus a subset of edges $M_j$ for each set $U_j$ which is constructed as follows. Let $e'$ and $e''$ be the two edges of $N$ incident to $U_j$. Let $C_1$ and $C_k$ be the components of $S$ incident to $e'$ and $e''$, resp., (possibly $k=1$) and let $C_1,C_2,\ldots,C_k$ the set of components along the unique path between $C_1$ and $C_k$ in $G'$. 
If $k=1$, set $M_i=\emptyset$ (Notice that, if $C_1$ is light, the edges $e'$ and $e''$ are incident to distinct nodes of $C_1$ by the definition of nice cycle). Otherwise $M_i$ is constructed analogously to the construction of the edge set $\{e_1,\ldots,e_{k-1}\}$ in the proof of Lemma \ref{lem:caseB:nonTreeEdge} with the following difference: In the definition of $e_1$ the role of $e_0$ is played by $e'$ and in the definition of $e_{k-1}$ the role of $e_0$ is replaced by $e''$ (notice that $e_0$ did not play any role in the definition of the remaining $e_i$'s).
\end{proof}
A nice cycle can be computed efficiently via the following lemma. 
\begin{lemma}\label{lem:caseB:NiceCycle}
Let $\Pi=(V_1,\ldots,V_k)$, $k\geq 2$, be a partition of the node-set of a 2VC simple graph $G$. In polynomial time one can compute a nice cycle $N$ of $\Pi$. 
\end{lemma}
Given the above claims, it is easy to conclude the proof of Lemma \ref{lem:caseB:mergingCycle}.
\begin{proof}[Proof of Lemma \ref{lem:caseB:mergingCycle}]
It is sufficient to compute a nice cycle $N$ of $\Pi(G')$ via Lemma \ref{lem:caseB:NiceCycle}, and then convert it into a merging cycle $M$ via Lemma \ref{lem:caseB:NiceToMerging}.
\end{proof}
It remains to prove Lemma \ref{lem:caseB:NiceCycle} (see Figure \ref{fig:ProofOfNiceCycle}). We remark that the claim of this lemma sounds rather basic and natural, hence it is reasonable that it is already known in the literature (possibly with a simpler/more elegant proof).

%

\begin{proof}[Proof of Lemma \ref{lem:caseB:NiceCycle}]
We use induction on $|\Pi|$ to obtain the desired cycle. The base case is $|\Pi|=2$. Since $G$ is 2VC, if $|V_1|,|V_2|>1$, there is a matching $N$ of size $2$ between $V_1$ and $V_2$, which induces a nice cycle. If one of $V_1$ and $V_2$, say $V_1$, has cardinality $1$, then there are at least two edges $e_1$ and $e_2$ between $V_1$ and $V_2$ (with distinct endpoints in $V_2$ since $G$ is simple). These edges form a nice cycle $N$.

Now assume $|\Pi|>2$. We need the following (rather involved) definition of almost-nice cycle of a partition $\Pi'$: 
\begin{definition}
Let $\Pi'$ be a partition of $V$, $U_0,U_1,\ldots,U_r\in \Pi'$ be distinct sets, $|U_0|\geq 2$, and $u_0\in U_0$. Furthermore for every $1\le i\le r$ there exist two nice-paths (not necessarily distinct) $P_{i,1}$ and $P_{i,2}$ from $U_0$ to $U_i$ such that:
\begin{itemize}
    \item $P_{i,1}$ and $P_{i,2}$ are incident to $u_0$ and all their nodes belong to $U_0\cup \dots \cup U_{r}$.
    \item If $|U_i|>1$ then $P_{i,1}$ and $P_{i,2}$ are incident to different nodes of $U_i$.
\end{itemize}
In this case we say that $A=(U_0,u_0, (U_1,P_{1,1},P_{1,2}), ...,(U_r,P_{r,1},P_{r,2}))$ is an almost-nice cycle of $\Pi'$.
\end{definition}
We start by constructing an almost-nice cycle for $\Pi$ involving at least two subsets of $\Pi$  (or finding a nice cycle, in which case \fab{the claim holds}). To this aim, we build a sequence $V_0,V_1,\ldots,V_k$ of distinct sets in $\Pi$, and a sequence of edges $e_1,\ldots,e_{k}$ such that $e_i$ has one endpoint in $V_{i-1}$ and the other in $V_i$. Furthermore, the two edges $e_i$ and $e_{i+1}$, $1\leq i<k$, are incident to distinct nodes of $V_i$ if $|V_i|>1$. We set $V_0$ to be any set in $\Pi$. Then we take as $e_1$ any edge with exactly one endpoint in $V_0$ and we let $V_1$ be the set containing the other endpoint of $e_1$. Given $V_0,\ldots,V_k$ and $e_1,\ldots,e_k$, we proceed as follows. We take any edge $e_{k+1}=vu$ with $v\in V_k$ and $u\notin V_k$, with the extra condition that $v$ is not an endpoint of $e_k$ if $|V_k|> 1$. Notice that such edge $e_{k+1}$ must exist since the graph is 2VC. Now if $u\notin V_0\cup \ldots\cup V_{k-1}$, we expand the sequence by adding the edge $e_{k+1}$ and the set $V_{k+1}\in \Pi$ containing $u$. Notice that this can happen at most $|\Pi|-1$ times since $k+1\leq |\Pi|$. Otherwise, let $u\in V_{j}$ for some $0\leq j<k$. If $|V_j|=1$, $e_{j+1},\ldots,e_{k+1}$ is a nice cycle and \fab{the claim holds}. The same holds if $|V_j|>1$ and $u$ is not incident to $e_{j+1}$. The only remaining case is that $u$ is incident to $e_{j+1}$. In this case we set $u_0=u$ and $(U_0,\ldots,U_{r})=(V_j,\ldots,V_k)$. For each $U_{i-j}$, the edges $\{e_{j+1},\cdots,e_{k}\}$ can be partitioned into two nice paths $P_{i-j,1}$ and $P_{i-j,2}$. These paths involve only nodes in $U_0\cup \ldots \cup U_r$, they have 
one endpoint in $u_0$ and the other in $U_{i-j}$. Furthermore these two paths have distinct endpoints in $U_{i-j}$ if $|U_{i-j}|> 1$. Thus $A=(U_0,u_0,(U_1,P_{1,1},P_{1,2}),...,(U_{r},P_{r,1},P_{r,2}))$ is an almost-nice cycle involving at least $2$ sets in $\Pi$.

Given the current almost-nice cycle $A=(U_0,u_0,(U_1,P_{1,1},P_{1,2}),...,(U_r,P_{r,1},P_{r,2}))$, $r\geq 1$, we either find a nice cycle $N$, hence \fab{the claim holds}, or we find a larger almost-nice cycle $A'$ (incident to more subsets of $\Pi$). Clearly this process ends within a polynomial number of steps. 

Given $A$, we proceed as follows. From $\Pi$ we obtain a new partition $\Pi'$ of $V$ by replacing $U_0, \ldots, U_r$ with their union $U'=U_0\cup \cdots \cup U_r$. We observe that $|\Pi'|=k-r<k$. We have the following cases:

\smallskip\noindent {\bf (1)} $|\Pi'|=1$. In this case since $u_0$ is not a cut-node in $G$, then there must be an edge $e=vw\in E$ such that $v\in U_0\setminus\{u_0\}$ and $w\in U_i$ for some $i\in \{1,2,...,r\}$. If $|U_i|=1$ then $P_{i,1}\cup \{e\}$ forms a nice cycle. Otherwise, i.e. $|U_i|>1$, at least one path among $P_{i,1}$ and $P_{i,2}$, say $P_{i,1}$, is not incident to $w$. Therefore $P_{i,1}\cup \{e\}$ is a nice cycle.

\smallskip\noindent {\bf (2)} $|\Pi'|>1$. By inductive hypothesis we can compute a nice cycle $N'$ of $\Pi'$. We distinguish the following subcases:

\smallskip\noindent {\bf (2.a)} $N'$ is not incident to $U'$ or is incident to a unique $U_i\subseteq U'$. In this case $N'$ is also a nice cycle for $\Pi$. 

\smallskip\noindent {\bf (2.b)} $N'$ is incident to $U'$ and it is incident to a node $v\in U_0\setminus\{u_0\}$. Let $w\in U_i\subseteq U'$ be the other node of $U'$ which $N'$ is incident to. Notice that $i\neq 0$ by case (2.a). If $|U_i|=1$, then $N:=P_{i,1}\cup N'$ is a nice cycle w.r.t. $\Pi$. If $|U_i|\ge 2$ at least one of $P_{i,1}$ and $P_{i,2}$, say $P_{i,1}$, is not incident to $w$. Then $N:=P_{i,1}\cup N'$ is a nice cycle w.r.t. $\Pi$.

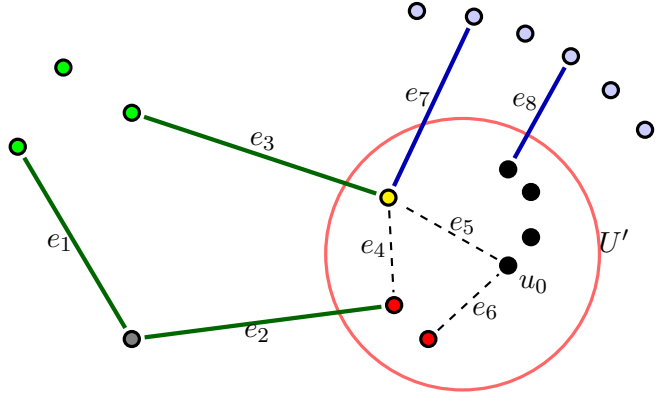
\begin{figure}
\begin{center}
\begin{tikzpicture}[scale=1.5]

\tikzset{vertex/.style={draw=black, very thick, circle,minimum size=0pt, inner sep=2pt, outer sep=2pt}
}


\begin{scope}[every node/.style={vertex}]
\filldraw[color=red!60, fill=red!0, very thick](3.9,0.75) circle (1.2);
\node[fill=gray] (a1) at (1,0) {};

\node [fill=green] (b1) at (0,1.7) {};
\node [fill=green] (b2) at (1,2) {};
\node [fill=green] (b3) at (0.4,2.4) {};

\node [fill=red] (c1) at (3.3,0.3){};
\node [fill=red] (c2) at (3.6,0){};

\node [fill=yellow] (d1) at (3.25,1.25) {};

\node [fill=black] (e1) at (4.3,0.65) {};
\node [fill=black] (e2) at (4.5,0.9){};
\node [fill=black] (e3) at (4.5,1.3){};
\node [fill=black] (e4) at (4.3,1.5){};

\node [fill=blue!20!white] (f0) at (5.5,1.85){};
\node [fill=blue!20!white] (f1) at (5.2,2.2){};
\node [fill=blue!20!white] (f3) at (4.85,2.5){};
\node [fill=blue!20!white] (f2) at (4.45,2.7){};
\node [fill=blue!20!white] (f4) at (4,2.85){};
\node [fill=blue!20!white] (f5) at (3.5,2.9){};

\draw[ultra thick, green!40!black] (d1) to (b2);
\draw[ultra thick, green!40!black] (b1) to (a1) to (c1);

\draw[ultra thick, blue!70!black] (f3) to (e4);
\draw[ultra thick, blue!70!black] (f4) to (d1);

\draw[black,dashed,thick] (c2) to (e1) to (d1) to (c1);

\end{scope}
\node[below right=0.1pt] () at (e1) {$u_0$};
\node () at (0.37,0.85) {$e_1$};
\node () at (2.1,0.05) {$e_2$};
\node () at (2.15,1.75) {$e_3$};

\node () at (3.12,0.8){$e_4$};
\node () at (3.9,1){$e_5$};
\node () at (4.1,0.25){$e_6$};

\node () at (3.51,2.14){$e_7$};
\node () at (4.45,2.1){$e_8$};

\node () at (5.23,0.87){$U'$};

\begin{scope}[red!80!black, very thick]

\end{scope}

\begin{scope}[very thick]
\end{scope}

\begin{scope}[densely dashed]
\end{scope}

\end{tikzpicture}
\end{center}
\caption{
 The initial partition $\Pi$ can be identified by the colors. Observe that $|\Pi|=6$. The circle corresponds to items that form an almost-nice cycle $A=(U_0,u_0,(U_1,P_{1,1},P_{1,2}),(U_2,P_{2,1},P_{2,2}))$, with $U_0$ being the set of black nodes, $U_1$ being the set of red nodes and $U_2$ is a singleton set containing the yellow node. Here $P_{1,1}=e_6$, $P_{1,2}=e_4,e_5$ and $P_{2,1}=P_{2,2}=e_5$. Now we obtain the new partition $\Pi'$ from $\Pi$ by simply merging the sets of $A$ into a single set $U'$. Observe that $|\Pi'|=4>1$. The blue (resp., green) edges correspond to a nice cycle $N_1$(resp., $N_2$) of $|\Pi'|$; The edge $e_8\in N_1$ is incident to a node in $U_0\setminus \{u_0\}$ and hence using $N_1$ and $A$, we can find the nice cycle $N'_1=\{e_8,e_5,e_7\}$ of $\Pi$ (This is case (2.b) in the proof of Lemma~\ref{lem:caseB:NiceCycle}). $N_2$ is not incident on any node of $U_0\setminus \{u_0\}$ and it is incident on $U'$ (this is case (2.c) in the proof of Lemma~\ref{lem:caseB:NiceCycle}). In this case we can form a larger almost-nice cycle $A'$ which includes the set $U_3$ of green nodes and the singleton set $U_4$ containing the gray node. In particular $P_{4,1}=P_{4,2}=e_2,e_6$, $P_{3,1}=e_1,e_2,e_6$, and $P_{3,2}=e_3,e_5$.
\label{fig:ProofOfNiceCycle}
}
\end{figure}

\smallskip\noindent {\bf (2.c)} $N'$ is incident to $U'$ and it is not incident to any node in $U_0\setminus\{u_0\}$. Let $U',U'_1,...,U'_s$ be the subsets of $\Pi'$ incident to $N'$. We next show how to obtain a larger almost-nice cycle 
$$
A'=(U_0,u_0,(U_1,P_{1,1},P_{1,2}),...,(U_r,P_{r,1},P_{r,2}),(U'_1,P'_{1,1},P'_{1,2}),...,(U'_s,P'_{s,1},P'_{s,2}).
$$ 
For every  
$j\in \{1,...,s\}$, we construct two nice paths $P'_{j,1}$ and $P'_{j,2}$ of $\Pi$ between $U_0$ and $U'_j$, which are incident to $u_0$ and are not incident to the same node of $U'_j$ if $|U'_j|>1$ as follows. Suppose first that $U'_j=\{w\}$. In this case we set $P'_{j,1}=P'_{j,2}=Q'\circ Q''$, defined as follows. We let $Q'\subset N'$ be any nice path from $U'_j$ to $U'$. Let $U_q\subseteq U'$ be the subset of $\Pi$ incident to $Q'$. If $U_q=U_0$ (in particular $Q'$ must be incident to $u_0$), $Q''=\emptyset$. Otherwise, we choose $Q''\in \{P_{q,1},P_{q,2}\}$ so that, if $|U_q|>1$, the last edge of $Q'$ and the first edge of $Q''$ are incident to distinct nodes of $U_q$. 

Suppose next that $|U'_j|>1$. In this case we consider the two nice paths $Q'_1$ and $Q'_2$ from $U'_j$ to $U'$ starting at two distinct nodes of $U'_j$ which are naturally induced by $N'$. For each such $Q'_i$, we define $Q''_i$ in the same way as in the previous case. Finally we set $P'_{j,1}=Q'_1\circ Q''_1$ and $P'_{j,2}=Q'_2\circ Q''_2$. 
\end{proof}

\section{Case of Few Triangles}
\label{sec:fewTriangles}

Let us start with the omitted proof of Lemma \ref{lem:caseA:initialCost}, \amEdit{which} we restate here for reader's convenience.
\lemmaCaseAinitialCost*
\begin{proof}
Let us initially assign $\frac{1}{3}$, $\frac{1}{4}$ and $\frac{3}{10}$ credits to each edge in $H$ which belongs to a triangle 2EC component, which is a bridge, and to the remaining edges, resp. We show how to use these credits to assign the credits according to the \fab{assignment scheme from Section \ref{sec:fewTrianglesOverview}}. The claim then follows since
$$
cost(H)=|H|+cr(H)\leq |H|+(\frac{1}{3}t+\frac{1}{4}b+\frac{3}{10}(1-t-b))|H|.
$$ 
Each 2EC component $C$ of $H$ retains the credits of its edges. The number of such credits is $\frac{1}{3}\cdot 3=1$ if $C$ is a triangle, and $\frac{3}{10}|E(C)|$ credits otherwise. This is sufficient to satisfy case (1) of the credit assignment scheme. Each bridge retains its credits: this satisfies case (3) of the credit assignment scheme.

Each 2EC block $B$ retains $1$ credit of the total amount of credits assigned to its edges, and the remaining credits of the mentioned edges are assigned to the connected component $C$ containing $B$. Notice that $|E(B)|\geq 4$ since $H$ is canonical, hence the edges of $B$ have at least $\frac{3}{10}\cdot 4>1$ credits, which is sufficient to assign one credit to $B$: this satisfies case (4) of the credit assignment scheme.  

Next consider any connected component $C$ of $H$ which is not 2EC. Notice that $C$ contains at least $2$ leaf blocks, namely blocks with exactly one edge of $C$ (which is a bridge) incident to them. Let $B_1$ and $B_2$ be any two such blocks. Since $H$ is canonical, each such $B_i$ contains at least $6$ edges. Therefore $C$ collects from $B_1$ and $B_2$ at least $2\cdot 6\cdot \frac{3}{10}-2>1$ credits: hence also case (2) of the credit assignment scheme is satisfied. 
\end{proof}

The proofs of the Bridge-Covering Lemma \ref{lem:caseA:bridgeCovering} and the Gluing Lemma \ref{lem:caseA:gluing} are given in the next two subsections.

\subsection{Bridge-Covering}

In this section we prove Lemma \ref{lem:caseA:bridgeCovering}, which we restate next.
 \lemmaCaseAbridgeCovering*


Let $C$ be any connected component of $S$ containing at least one bridge (i.e., $C$ is not 2EC). Let $G_C$ be the multi-graph obtained from $G$ by contracting into a single node each block $B$ of $C$ and each connected component $C'$ of $S$ other than $C$. Let $T_C$ be the tree in $G_C$ induced by the bridges of $C$: we call the nodes of $T_C$ corresponding to blocks \emph{block nodes}, and the remaining nodes of $T_C$ \emph{lonely nodes}. Observe that the leaves of $T_C$ are necessarily block nodes (otherwise $S$ would not be a 2-edge-cover). 

\amEdit{at a high leve}l, in this stage of our construction we will transform $S$ into a new solution $S'$ containing a component $C'$ spanning the nodes of $C$ (and possibly the nodes of some other components of $S$). Furthermore, no new bridge is created and at least one bridge $e$ of $C$ is not a bridge of $C'$ (intuitively, the bridge $e$ is covered). We remark that each component of the initial (canonical) 2-edge-cover $H$ which is not 2EC contains at least 12 nodes, and we only possibly merge together components in this stage of the construction. As a consequence Invariant \ref{inv:2ECcomponent} is preserved (in particular, if $C'$ happens to be 2EC, it contains at least 12 edges).


A \emph{bridge-covering path} $P_C$ is any path in $G_C\setminus E(T_C)$ with its (distinct) endpoints $u$ and $v$ in $T_C$, and the remaining (internal) nodes outside $T_C$. Notice that $P_C$ might consist of a single edge, possibly parallel to some edge in $E(T_C)$. Augmenting $S$ along $P_C$ means adding the edges of $P_C$ to $S$, hence obtaining a new 2-edge-cover $S'$. Notice that $S'$ obviously has fewer bridges than $S$: in particular all the bridges \fab{of $S$} along the $u$-$v$ path in $T_C$ are \fab{not bridges in $S'$ (we also informally say that such bridges are \emph{removed})}, and \fab{the bridges of $S'$ are a subset of the bridges of $S$}. Let us analyze $cost(S')$. Suppose that the distance between $u$ and $v$ in $T_C$ is $br$ and such path contains $bl$ blocks. 
Then the number of edges w.r.t. $S$ grows by $|E(P_C)|$. The number of
credits w.r.t. $S$ decreases by at least $\frac{1}{4}br+bl+|E(P_C)|-1$ since we remove $br$ bridges, $bl$ blocks and $|E(P_C)|-1$ components (each one bringing at least one credit). However the number of credits also grows by $1$ since we create a new block $B'$ (which needs $1$ credit) or a new 2EC component $C'$ (which needs $1$ additional credit w.r.t. the credit of $C$). Altogether $cost(S)-cost(S')\geq \frac{1}{4}br+bl-2$. We say that $P_C$ is cheap if the latter quantity is non-negative, and expensive otherwise. In particular $P_C$ is cheap if it involves at least $2$ block nodes or $1$ block node and at least $4$ bridges. Notice that a bridge-covering path, if \fab{at least one such path exists}, can be computed in polynomial time. 

Before proving Lemma \ref{lem:caseA:bridgeCovering} we need the following two technical lemmas (see Figures \ref{fig:LBonSizeOfR(W)} and \ref{fig:BridgeCoveringPathsIntersection}). We say that a node $v\in V(T_C)\setminus \{u\}$ is reachable from $u\in V(T_C)$ if there exists a bridge-covering path between $v$ and $u$. Let $R(W)$ be the nodes in $V(T_C)\setminus W$ reachable from some node in $W\subseteq V(T_C)$, and let us use the shortcut $R(u)=R(\{u\})$ for $u\in V(T_C)$. Notice that $v\in R(u)$ iff $u\in R(v)$.

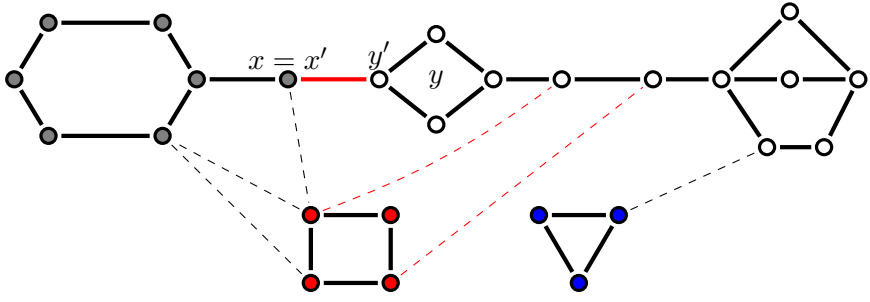
\begin{figure}
\begin{center}
\begin{tikzpicture}[scale=1.5]

\tikzset{vertex/.style={draw=black, very thick, circle,minimum size=0pt, inner sep=2pt, outer sep=2pt}
}


\begin{scope}[every node/.style={vertex}]
\node[fill=gray] (u1) at (0,0) {};
\node[fill=gray] (u4) at (-1.6,0){};
\node[fill=gray] (u2) at (-0.3,-0.5) {};
\node[fill=gray] (u6) at (-0.3,0.5){};
\node[fill=gray] (u3) at (-1.3,-0.5){};
\node[fill=gray] (u5) at (-1.3,0.5){};
\draw[ultra thick] (u1) to (u2) to (u3) to (u4) to (u5) to (u6) to (u1);

\node[fill=gray] (x) at (0.8,0) {};
\node (y) at (1.6,0){};
\node (y') at (2.6,0){};
\node (z) at (2.1,-0.4){};
\node (z') at (2.1,0.4){};
\node (w1) at (3.2,0) {};
\node (w2) at (4,0) {};

\node (a1) at (4.6,0) {};
\node (a2) at (5.2,0) {};
\node (a3) at (5.8,0){};
\node (a4) at (5.2,0.6){};
\node (a5) at (5,-0.6){};
\node (a6) at (5.5,-0.6){};

\node[fill=red] (r1) at (1,-1.2){};
\node[fill=red] (r2) at (1.7,-1.2){};
\node[fill=red] (r3) at (1.7,-1.8){};
\node[fill=red] (r4) at (1,-1.8){};

\draw[ultra thick] (r1) to (r2) to (r3) to (r4) to (r1);

\draw[dashed,bend right=8,red] (r1) to (w1);
\draw[dashed,red] (r3) to (w2);
\draw[dashed] (x) to (r1) to (u2) to (r4);

\draw[ultra thick] (u1) to (x);
\draw[ultra thick,red] (x) to (y);
\draw[ultra thick] (y) to (z') to (y') to (w1) to (w2) to (a1) to (a4) to (a3) to (a2) to (a1) to (a5) to (a6) to (a3);
\draw[ultra thick] (y) to (z) to (y');

\node[fill=blue] (s1) at (3,-1.2){};
\node[fill=blue] (s2) at (3.7,-1.2){};
\node[fill=blue] (s3) at (3.35,-1.8){};

\draw[ultra thick] (s1) to (s2) to (s3) to (s1);

\draw[dashed] (s2) to (a5);

\end{scope}

\node[above] () at (x){$x=x'$};
\node[above] () at (y){$y'$};
\node () at (2.1,0){$y$};

\begin{scope}[red!80!black, very thick]

\end{scope}

\begin{scope}[very thick]
\end{scope}

\begin{scope}[densely dashed]
\end{scope}

\end{tikzpicture}
\end{center}
\caption{Illustration of Lemma \ref{lem:reachable}. The solid edges define $S$ and the dashed ones are part of the remaining edges. The gray and white nodes induce the sets $X_C$ and $Y_C$, resp. The set $X$ is induced by the gray and red nodes. Notice that $y$ corresponds to the block with $4$ nodes (belonging to $C$). The red edges form the 3-matching used in the proof of the mentioned lemma. 
}
\label{fig:LBonSizeOfR(W)}
\end{figure}

\begin{lemma}\label{lem:reachable}
Let $e=xy\in E(T_C)$ and let $X_C$ and $Y_C$ be the nodes of the two trees obtained from $T_C$ after removing the edge $e$, where $x\in X_C$ and $y\in Y_C$. Then $R(X_C)$ contains a block node or $R(X_C)\setminus\{y\}$ contains at least $2$ lonely nodes.
\end{lemma}
\begin{proof}
Let us assume that $R(X_C)$ contains only lonely nodes, otherwise \fab{the claim holds}. Let $X$ be the nodes in $G_C$ which are connected to $X_C$ after removing $Y_C$, and let $Y$ be the remaining nodes in $G_C$. Let $X'$ and $Y'$ be the nodes in $G$ corresponding to $X$ and $Y$, resp. Let $e'$ be the edge in $G$ that corresponds to $e$ and let $y'$ be the node of $Y'$ that is incident to $e'$. In particular if $y$ is a lonely node then $y'=y$ and otherwise $y'$ belongs to the block of $C$ corresponding to $y$.

Observe that both $X_C$ and $Y_C$ (hence $X$ and $Y$) contain at least one (leaf) block node, hence both $X'$ and $Y'$ contain at least 6 nodes\footnote{Leaf blocks initially contain at least 6 nodes being $H$ canonical, and the bridge-covering stage does not create smaller leaf blocks.}. Therefore we can apply the 3-matching Lemma \ref{lem:matchingOfSize3} to $X'$ and obtain a matching $M=\{u'_1v'_1,u'_2v'_2,u'_3v'_3\}$ between $X'$ and $Y'$, where $u'_i\in X'$ and $v'_i\in Y'$. Let $u_i$ and $v_i$ be the nodes in $G_C$ corresponding to $u'_i$ and $v'_i$, resp. We remark that $v_i\in Y_C$: indeed otherwise $v_i$ would be connected to $X_C$ in $G_C\setminus Y_C$, contradicting the definition of $X$. Let us show that the $v_i$'s are all distinct. Assume by contradiction that $v_i=v_j$ for $i\neq j$. Since $v'_i\neq v'_j$ (being $M$ a matching) and $v'_i,v'_j$ are both associated with $v_i$, this means that $v_i$ is a block node in $R(X_C)$, a contradiction.   

For each $u_i$ there exists a path $P_{w_iu_i}$ in $G_C\setminus E(T_C)$ between $u_i$ and some $w_i\in X_C$ (possibly $w_i=u_i$). Observe that $P_{w_iu_i}\circ u_iv_i$ is a bridge-covering path between $w_i\in X_C$ and $v_i\in Y_C$ unless $u_iv_i=xy$ (recall that the edges $E(T_C)$ cannot be used in a bridge-covering path). In particular, since the $v_i$'s are all distinct, at least two such paths are bridge-covering paths from $X_C$ to distinct (lonely) nodes of $Y_C\setminus \{y\}$, implying $|R(X_C)\setminus\{y\}|\geq 2$.
%
\end{proof}

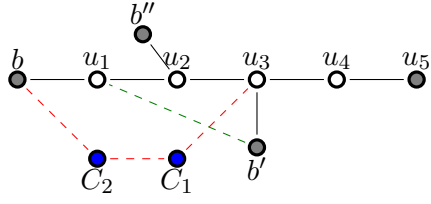
\begin{figure}
\begin{center}
\begin{tikzpicture}[scale=1.5]

\tikzset{vertex/.style={draw=black, very thick, circle,minimum size=0pt, inner sep=2pt, outer sep=2pt}
}


\begin{scope}[every node/.style={vertex}]

\node[fill=gray] (b) at (0.4,0) {};
\node (u1) at (1.1,0){};
\node (u2) at (1.8,0){};
\node (u3) at (2.5,0) {};
\node (u4) at (3.2,0){};
\node[fill=gray] (u5) at (3.9,0){};
\node[fill=gray] (x) at (1.5, 0.4){};
\node [fill=gray] (b') at (2.5,-0.6){};

\node[fill=blue] (c1) at (1.8,-0.7){};
\node[fill=blue] (c2) at (1.1,-0.7){};

\draw (b) to (u1) to (u2) to (u3) to (u4) to (u5);
\draw (u2) to (x);
\draw (u3) to (b');

\draw[dashed,red] (b) to (c2) to (c1) to (u3);
\draw[dashed,green!50!black] (b') to (u1);
\end{scope}

\node[above] () at (x) {$b''$};
\node[above] () at (b) {$b$};
\node[above] () at (u1) {$u_1$};
\node[above] () at (u2) {$u_2$};
\node[above] () at (u3) {$u_3$};
\node[above] () at (u4) {$u_4$};
\node[above] () at (u5) {$u_5$};
\node[below] () at (b') {$b'$};
\node[below] () at (c2) {$C_2$};
\node[below] () at (c1) {$C_1$};

\begin{scope}[red!80!black, very thick]

\end{scope}

\begin{scope}[very thick]
\end{scope}

\begin{scope}[densely dashed]
\end{scope}

\end{tikzpicture}
\end{center}
\caption{The solid edges induce $T_C$. The blue nodes correspond to contracted connected components other than $C$. The black and white nodes are block and lonely nodes of $T_C$, resp. The red edges induce an (expensive) bridge-covering path covering the bridges $u_1u_2$, $u_2u_3$ and a bridge corresponding to $bu_1$. The green edge induces another (expensive) bridge-covering path. An edge $bb''$ or $bu_4$ would induce a cheap bridge-covering path. Lemma \ref{lem:merge2paths} can be applied with $u=u_3\in R(b)$ and $u'=u_1\in R(b')$ (both $T_C(b,u_3)$ and $T_C(b',u_1)$ contain $w=u_2$). In particular, one obtains $S'$ by adding to $S$ the red and green edges. Notice that the new component $C'$ of $S'$ which contains $V(C)$ has $4$ fewer bridges and $1$ fewer block than $C$.
}
\label{fig:BridgeCoveringPathsIntersection}
\end{figure}

Let $T_C(u,v)$ denote the path in $T_C$ between nodes $u$ and $v$. 
\begin{lemma}\label{lem:merge2paths}
Let $b$ and $b'$ be two leaf (block) nodes of $T_C$. Let $u\in R(b)$ and $u'\in R(b')$ be nodes of $V(T_C)\setminus \{b,b'\}$. Suppose that $T_C(b,u)$ and $T_C(b',u')$ both contain some node $w$ (possibly $w=u=u'$) and $|E(T_{C}(b,u))\cup E(T_C(b,'u'))|\geq 4$. Then in polynomial time one can find a 2-edge-cover $S'$ satisfying the conditions (1) and (2) of Lemma \ref{lem:caseA:bridgeCovering}.  
\end{lemma}
\begin{proof}
Let $P_{bu}$ (resp., $P_{b'u'}$) be a bridge-covering path between $b$ and $u$ (resp., $b'$ and $u'$). Suppose that $P_{bu}$ and $P_{b'u'}$ share an internal node. Then there exists a (cheap) bridge-covering path between $b$ and $b'$ and we can find the desired $S'$ by the previous discussion. So next assume that $P_{bu}$ and $P_{b'u'}$ are internally node disjoint. Let $S'=S\cup E(P_{bu})\cup E(P_{b'u'})$. All the nodes and bridges induced by $E(P_{bu})\cup E(P_{b'u'})\cup E(T_C(b,u))\cup E(T_C(b',u'))$ become part of the same block or 2EC component $C'$ of $S'$. Furthermore $C'$ contains at least $4$ bridges of $C$ and the two blocks $B$ and $B'$ corresponding to $b$ and $b'$, resp. One has $|S'|=|S|+|E(P_{bu})|+|E(P_{b'u'})|$. Moreover $cr(S')\leq cr(S)+1-(|E(P_{bu})|-1)-(|E(P_{b'u'})|-1)-2-\frac{1}{4}4$. In the latter inequality the $+1$ is due to the extra credit required by $C'$, the $-2$ to the removed blocks $B$ and $B'$ \fab{(i.e., blocks of $S$ which are not blocks of $S'$)}, and the $-\frac{1}{4}4$ to the removed bridges. Altogether $cost(S)-cost(S')\geq \frac{1}{4}\cdot 4+2-3=0$, hence (2) holds. Notice that $S'$ satisfies (1) too.
\end{proof}

We are now ready to prove Lemma \ref{lem:caseA:bridgeCovering}.
\begin{proof}[Proof of Lemma \ref{lem:caseA:bridgeCovering}]
If there exists a cheap bridge-covering path $P_C$ (condition that we can check in polynomial time), we simply augment $S$ along $P_C$ hence obtaining the desired $S'$. Thus we next assume that no such path exists. 

Let $P'=b,u_1,\ldots,u_\ell$ in $T_C$ be a longest path in $T_C$ (interpreted as a sequence of nodes).  Notice that $b$ must be a leaf of $T_C$, hence a block node (corresponding to some leaf block $B$ of $C$). Let us consider $R(b)$. Since by assumption there is no cheap bridge-covering path, $R(b)$ does not contain any block node. Hence $|R(b)\setminus\{u_1\}|\geq 2$ and $R(b)$ contains only lonely nodes by Lemma \ref{lem:reachable} (applied to $xy=bu_1$). 

We next distinguish a few subcases depending on $R(b)$. Let $V_i$, $i\geq 1$, be the nodes in $V(T_C)\setminus V(P')$ such that their path to $b$ in $T_C$ passes through $u_i$. Notice that $\{V(P'),V_1, \ldots,V_\ell\}$ is a partition of $V(T_C)$. We observe that any node in $V_i$ is at distance at most $i$ from $u_i$ in $T_C$ (otherwise $P'$ would not be a longest path in $T_C$). We also observe that as usual the leaves of $T_C$ in $V_i$ are block nodes. 
\begin{remark}\label{rem:V1V2}
This implies that all the nodes in $V_1$ are block nodes, and all the nodes in $V_2$ are block nodes or are lonely non-leaf nodes at distance $1$ from $u_2$ in $T_C$.
\end{remark}

\medskip\noindent{\bf (1)} There exists $u\in R(b)$ with $u\notin \{u_1,u_2,u_3\}\cup V_1\cup V_2$. By definition there exists a bridge-covering path between $b$ and $u$ containing at least $4$ bridges, hence cheap. This is excluded by the previous steps of the construction. 

\medskip\noindent{\bf (2)} There exists $u\in R(b)$ with $u\in V_1\cup V_2$. Since $u$ is not a block node, by Remark \ref{rem:V1V2} $u$ must be a lonely non-leaf node in $V_2$ at distance $1$ from $u_2$. Furthermore $V_2$ must contain at least one leaf block node $b'$ adjacent to $u$. Consider $R(b')$. By the assumption that there are no cheap bridge-covering paths and Lemma~\ref{lem:reachable} (applied to $xy=b'u$), $|R(b')\setminus \{u\}|\geq 2$.
In particular $R(b')$ contains at least one lonely node $u'\notin \{b',b,u\}$. The tuple $(b,b',u,u')$ satisfies the conditions of Lemma \ref{lem:merge2paths} (in particular both $T_C(b,u)$ and $T_C(b',u')$ contain $w=u_2$), hence we can obtain the desired $S'$. 
%
%

\medskip\noindent{\bf (3)} $R(b)\setminus\{u_1\}=\{u_2,u_3\}$. Recall that $u_2$ and $u_3$ are lonely nodes. We distinguish $2$ subcases:

\medskip\noindent{\bf (3.a)} $V_1\cup V_2\neq \emptyset$. Take any leaf (block) node $b'\in V_1\cup V_2$, say $b'\in V_i$. Let $\ell'$ be the node adjacent to $b'$. By the assumption that there are no cheap bridge-covering paths and Lemma~\ref{lem:reachable} (applied to $xy=b'\ell'$), $R(b')\setminus \{\ell'\}$ has cardinality at least $2$ and contains only lonely nodes. Choose any $u'\in R(b')\setminus \{\ell'\}$. Notice that $u'\notin V_1$ by Remark \ref{rem:V1V2} (but it could be a lonely node in $V_2$ other than $\ell'$). 
The tuple $(b,b',u_3,u')$ satisfies the conditions of Lemma \ref{lem:merge2paths} (in particular both $T_C(b,u_3)$ and $T_C(b',u')$ contain $w=u_2$), hence we can compute the desired $S'$.

\medskip\noindent{\bf (3.b)} $V_1\cup V_2= \emptyset$. By Lemma \ref{lem:reachable} (applied to $xy=u_1u_2$) the set $R(\{b,u_1\})$ contains a block node or $R(\{b,u_1\})\setminus\{u_2\}$ contains at least 2 lonely nodes. Suppose first that $R(\{b,u_1\})$ contains a block node $b'$. Notice that $b'\notin R(b)$ by the assumption that there are no cheap bridge-covering paths, hence $u_1 \in R(b')$. Notice also that $b'\notin \{u_2,u_3\}$ since those are lonely nodes. Thus the tuple $(b,b',u_2,u_1)$ satisfies the conditions of Lemma \ref{lem:merge2paths} (in particular both $T_C(b,u_2)$ and $T_C(b',u_1)$ contain $w=u_2$), hence we can obtain the desired $S'$. 

The remaining case is that $R(\{b,u_1\})\setminus\{u_2\}$ contains at least $2$ lonely nodes. Let us choose $u'\in R(\{b,u_1\})\setminus\{u_2\}$ with $u'\neq u_3$. Let $P_{bu_2}$ (resp., $P_{u'u_1}$) be a bridge-covering path between $b$ and $u_2$ (resp., $u'$ and $u_1$). Notice that $P_{bu_2}$ and $P_{u'u_1}$ must be internally node disjoint otherwise $u'\in R(b)$, which is excluded since $R(b)\subseteq \{u_1,u_2,u_3\}$ by assumption. Consider $S'=S\cup E(P_{bu_2})\cup E(P_{u'u_1})\setminus \{u_1u_2\}$. Notice that $S'$ satisfies (1). One has $|S'|=|S|+|E(P_{bu_2})|+|E(P_{u'u_1})|-1$, where the $-1$ comes from the edge $u_1u_2$. Furthermore $cr(S')\leq cr(S)+1-(|E(P_{bu_2})|-1)-(|E(P_{u'u_1})|-1)-\frac{1}{4}{4}-1$, where the $+1$ comes from the extra credit needed for the block or 2EC component $C'$ containing $V(B)$, the final $-1$ from the removed block $B$, and the $-\frac{1}{4}{4}$ from the at least $4$ bridges removed from $S$ (namely a bridge corresponding to $bu_1$, the edges $u_1u_2$ and $u_2u_3$, and one more bridge incident to $u'$). Altogether $cost(S)-cost(S')\geq 0$, hence $S'$ satisfies (2).
\end{proof}

\subsection{Gluing}

In this section we complete the description of the Gluing stage by proving Lemma \ref{lem:caseA:gluing}, \amEdit{which} we next restate.
\lemmaCaseAgluing*

Recall that the component graph $\hat{G}_S$ has one node per 2EC component of $S$ and one edge $C_1C_2$ iff $G$ contains at least one edge between the components $C_1$ and $C_2$ of $S$. In the gluing step associated with $S$ we will check if $\hat{G}_S$ satisfies certain properties, and based on that we derive a convenient $S'$. We start by dealing with certain \emph{local configurations} involving $2$ or $3$ 2EC components that can be merged together in Section \ref{sec:localConfigurations}. Then we consider the case that $\hat{G}_S$ is a tree in Section \ref{sec:caseA:tree}. Finally we consider the case that $\hat{G}_S$ is \emph{not} a tree in Section \ref{sec:caseA:nonTree}.

\subsubsection{Local Merges}
\label{sec:localConfigurations}

If one of the following cases applies (which we can check in polynomial time), then we obtain the desired $S'$. We remark that in each case $S'$ has the same 2EC components of $S$ a part for a new 2EC component $C'$ that spans the nodes of two or three 2EC components of $S$. In particular conditions (1) and (2) of Lemma \ref{lem:caseA:gluing} immediately hold. We also remark that $C'$ contains at least $7$ edges, hence Invariant \ref{inv:2ECcomponent} is preserved by a simple induction.

The next $3$ lemmas can be used to merge two 2EC components of $S$ into one.
\begin{lemma}\label{lem:4cycle}
Suppose that there exists two components $C_1$ and $C_2$ of $S$ where $C_1$ is a 4-cycle, $C_2$ is \emph{not} a $5$-cycle, and there exists a 3-matching $M$ between $C_1$ and $C_2$. Then in polynomial time one can compute an $S'$ satisfying the conditions (1), (2) and (3) of Lemma \ref{lem:caseA:gluing}. 
\end{lemma}
\begin{proof}
Let $M=\{u_1v_1,u_2v_2,u_3v_3\}$ with $u_i\in C_1$ and $v_j\in C_2$. 

Assume by contradiction that $C_2$ is a triangle or a $4$-cycle. By a simple case analysis we can always choose $2$ edges in $M$, say $u_1v_1$ and $u_2v_2$, such that $u_1$ and $u_2$ (resp., $v_1$ and $v_2$) and adjacent in $C_1$ (resp., in $C_2$). Let $S'=S\cup \{u_1v_1,u_2v_2\}\setminus \{u_1u_2,v_1v_2\}$. Notice that $S'$ is a 2-edge-cover of the same size as $S$ but with fewer connected components than $S$. By Invariant \ref{inv:2ECcomponent}, both $C_1$ and $C_2$ are 2EC components of the initial canonical 2-edge-cover $H$. Hence the existence of $S'$ contradicts the 3-optimality of $H$.


Therefore $C_2$ is a $6$-cycle or large. We remark that $cr(C_2)\geq \frac{3}{10}6$. In this case we choose any two edges in $M$, say 
$u_1v_1$ and $u_2v_2$, such that $u_1$ and $u_2$ are adjacent in $C_1$. Let $S'=S\setminus \{u_1u_2\}\cup \{u_1v_1,u_2v_2\}$. Notice that $S'$ satisfies (1) and (2). In particular the nodes of $C_1$ and $C_2$ belong to a unique (large) 2EC component $C'$ of $S'$. One has $|S'|=|S|+1$ and 
\begin{align*}
cr(S)-cr(S')=cr(C_1)+cr(C_2)-cr(C')\geq \frac{3}{10}4+\frac{3}{10}6-2=1. 
\end{align*}
Hence $cost(S')\leq cost(S)$, i.e. $S'$ satisfies (3). 
\end{proof}



\begin{lemma}\label{lem:triangle}
Suppose that there exists two components $C_1$ and $C_2$ of $S$ where $C_1$ is a triangle, $C_2$ is \emph{not} a $6$-cycle, and there exists a 3-matching $M$ between $C_1$ and $C_2$. Then in polynomial time one can compute an $S'$ satisfying the conditions (1), (2) and (3) of Lemma \ref{lem:caseA:gluing}. 
\end{lemma}
\begin{proof}
Let $M=\{u_1v_1,u_2v_2,u_3v_3\}$, where $u_i\in C_1$ and $v_j\in C_2$. 
We can exclude the case where $C_2$ is a $4$-cycle since it is captured by Lemma \ref{lem:4cycle}. If $C_2$ is a triangle or a $5$-cycle, we can choose two edges in $M$, say $u_1v_1$ and $u_2v_2$, such that $v_1$ and $v_2$ are adjacent in $C_2$ ($u_1$ and $u_2$ are obviously adjacent in $C_1$). Consider $S'=S\cup \{u_1v_1,u_2v_2\}\setminus \{u_1u_2,v_1v_2\}$. The same argument as in the proof of Lemma \ref{lem:4cycle} leads to a contradiction (of the 3-optimality of the initial $H$). 

The remaining case is that $C_2$ is large. In this case consider $S'=S\setminus \{u_1u_2\}\cup \{u_1v_1,u_2v_2\}$. Notice that $S'$ satisfies (1) and (2). Let $C'$ be the (large) 2EC component of $S'$ containing the nodes of $C_1$ and $C_2$. One has $|S'|=|S|+1$ and 
\begin{align*}
cr(S)-cr(S')=cr(C_1)+cr(C_2)-cr(C')\geq 1+2-2=1. 
\end{align*}
Hence $cost(S')\leq cost(S)$, i.e. $S'$ satisfies (3).
\end{proof}

\begin{lemma}\label{lem:5cycle}
Suppose that there exists two components $C_1$ and $C_2$ of $S$ where $C_1$ is a $5$-cycle, $C_2$ is not a $4$-cycle, and that there exists a 3-matching $M$ between $C_1$ and $C_2$.
Then in polynomial time one can compute an $S'$ satisfying the conditions (1), (2) and (3) of Lemma \ref{lem:caseA:gluing}. 
\end{lemma}

\begin{proof}
Let $M=\{u_1v_1,u_2v_2,u_3v_3\}$, with $u_i\in C_1$ and $v_j\in C_2$. If $C_2$ is a triangle then we can apply Lemma~\ref{lem:triangle} and obtain the desired $S'$. Therefore we assume $C_2$ is $5$-cycle, $6$-cycle or large. Select any two edges in $M$, say $u_1v_1$ and $u_2v_2$, such that $u_1$ and $u_2$ are adjacent in $C_1$. Consider $S'=S\cup \{u_1v_1,u_2v_2\}\setminus \{u_1u_2\}$ and define $C'$ analogously to the previous case. Also in this case $S'$ satisfies (1) and (2). One has $|S'|=|S|+1$ and 
\begin{align*}
cr(S)-cr(S')=cr(C_1)+cr(C_2)-cr(C')\geq \frac{3}{10}5+\frac{3}{10}5-2=1. 
\end{align*}
Hence $cost(S')\leq cost(S)$, i.e. $S'$ satisfies (3).

\end{proof}

\begin{lemma}\label{lem:large}
Suppose that there exists two large components $C_1$ and $C_2$ and there exists a 3-matching $M$ between $C_1$ and $C_2$.
Then in polynomial time one can compute an $S'$ satisfying the conditions (1), (2) and (3) of Lemma \ref{lem:caseA:gluing}. 
\end{lemma}

\begin{proof} 

In this case take any two edges $e,f\in M$, and consider $S'=S\cup \{e,f\}$. Notice that $S'$ satisfies (1) and (2). Let $C'$ be the 2EC component of $S'$ containing the nodes of $C_1$ and $C_2$. One has $|S'|=|S|+2$ and 
\begin{align*}
cr(S)-cr(S')=cr(C_1)+cr(C_2)-cr(C')=2+2-2=2. 
\end{align*}
Hence $cost(S')\leq cost(S)$, i.e. $S'$ satisfies (3).  

\end{proof}

We next describe two cases when we can merge $3$ 2EC components of $S$ into one.
\begin{lemma}\label{lem:3componentsC1_4cycle}
Let $C_1$, $C_2$ and $C_3$ be 2EC components of $S$ such that $C_1$ is a $4$-cycle, $C_2$ and $C_3$ are $5$-cycles, and there is a 3-matching $M_{12}$ (resp., $M_{13}$) between $C_1$ and $C_2$ (resp., $C_3$).
Then in polynomial time one can compute an $S'$ satisfying the conditions (1), (2) and (3) of Lemma \ref{lem:caseA:gluing}. 
\end{lemma}
\begin{proof}
Let us choose two edges $u_1v_1,u_2v_2\in M_{12}$ with $u_i\in C_2$ and $v_j\in C_1$, such that $u_1$ and $u_2$ are adjacent in $C_2$. Similarly, let us choose two edges $w_1v'_1,w_2v'_2\in M_{12}$ with $w_i\in C_3$ and $v'_j\in C_1$, such that $w_1$ and $w_2$ are adjacent in $C_3$. Consider $S'=S\setminus \{u_1u_2,w_1w_2\}\cup \{u_1v_1,u_2v_2,w_1v'_1,w_2v'_2\}$. Observe that $S'$ satisfies (1) and (2). In particular it contains a 2EC component $C'$ spanning the nodes $V(C_1)\cup V(C_2)\cup V(C_3)$. $S'$ satisfies (3) since
\begin{align*}
cost(S)-cost(S') & =|S|-|S'|+cr(C_1)+cr(C_2)+cr(C_3)-cr(C')\\
& =-2+\frac{3}{10}4+\frac{3}{10}5+\frac{3}{10}5-2>0
\end{align*}
\end{proof}

We need the following technical lemma (see Figure \ref{fig:c5c4c9}).

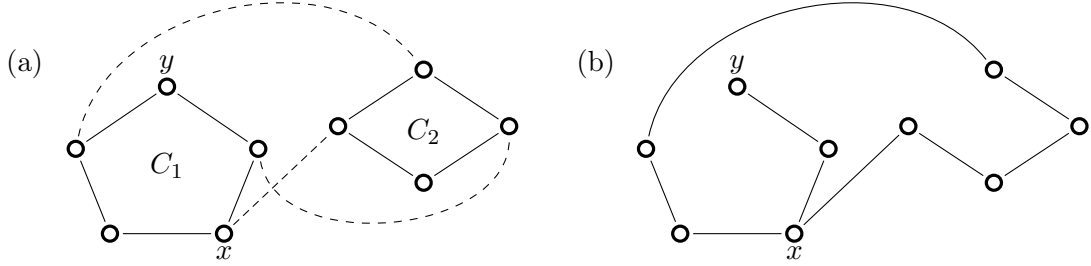
\begin{figure}
\begin{center}
\begin{tikzpicture}[scale=1.5]

\tikzset{vertex/.style={draw=black, very thick, circle,minimum size=0pt, inner sep=2pt, outer sep=2pt}
}


\node[right] () at (-1,1.5) {(a)};
\node[right] () at (4,1.5) {(b)};

\node[below] () at (0.5,0.8) {$C_1$};
\node[below] () at (2.75,1.1) {$C_2$};

\begin{scope}[every node/.style={vertex}]

\node (0) at (0,0) {};
\node (1) at (1,0) {};
\node (2) at (1.3,0.75) {};

\node (4) at (-0.3,0.75) {};
\node (3) at (0.5,1.3) {};

\node (0') at (5,0) {};
\node (1') at (6,0) {};
\node (2') at (6.3,0.75) {};

\node (4') at (4.7,0.75) {};
\node (3') at (5.5,1.3) {};

\draw (1) to (2);
\draw (2) to (3);
\draw (3) to (4);
\draw (4) to (0);
\draw (0) to (1);

\draw (1') to (2');
\draw (2') to (3');
\draw (4') to (0');
\draw (0') to (1');


\node (a) at (2,0.95) {};

\node (b) at (2.75,0.45) {};
\node (d) at (2.75,1.45) {};
\node (c) at (3.5,0.95) {};

\node (a') at (7,0.95) {};

\node (b') at (7.75,0.45) {};
\node (d') at (7.75,1.45) {};
\node (c') at (8.5,0.95) {};



\draw (d) to (a);
\draw (c) to (d);
\draw (b) to (c);
\draw (a) to (b);

\draw (c') to (d');
\draw (b') to (c');
\draw (a') to (b');

\draw [bend left=65, dashed](4) to (d);

\draw [bend right=85,dashed](2) to (c);

\draw [bend left=65, ](4') to (d');

\draw [](1') to (a');

\end{scope}


\begin{scope}[red!80!black, very thick]
\end{scope}


\begin{scope}[very thick]
\end{scope}

\node[below=1pt] () at (1) {$x$};
\node[below=1pt] () at (1') {$x$};
\node[above=1pt] () at (3) {$y$};
\node[above=1pt] () at (3') {$y$};

\begin{scope}[densely dashed]
\draw (1) -- (a);

\end{scope}

\end{tikzpicture}
\end{center}
\caption{
(a) Illustration of Lemma~\ref{lem:c5c4to9}, where the dashed edges correspond to the edges of $M$.
(b) A valid subset $F$ of $9$ edges. Note that $F\cup\{xy\}$ forms a 2EC graph on $V(C_1)\cup V(C_2)$.
\label{fig:c5c4c9}
}
\end{figure}

\begin{lemma}\label{lem:c5c4to9}
Let $C_1$ and $C_2$ be two 2EC components of $S$ such that $C_1$ is a 5-cycle, $C_2$ is a 4-cycle, and there is a 3-matching $M$ between $C_1$ and $C_2$. Then for any two nodes $x,y\in V(C_1)$ there exists a set $F$ of $9$ edges in $G$ such that $F\cup \{xy\}$ induces a 2EC spanning subgraph of $V(C_1)\cup V(C_2)$.
\end{lemma}
\begin{proof}
Let $C_1=v_1,v_2,v_3,v_4,v_5,v_1$ and $C_2=u_1,u_2,u_3,u_4,u_1$. 
We consider the following cases, each time assuming that the previous cases do not apply:

\medskip\noindent {\bf (a)} There exist $e_1,e_2\in M$ which are incident to adjacent nodes of $C_1$, say $v',v''$, and adjacent nodes of $C_2$, say $u',u''$, the cycle $F=(E(C_1)\setminus \{v'v''\}) \cup (E(C_2)\setminus \{u'u''\})\cup\{e_1,e_2\}$ satisfies the claim.

In the following we can assume w.l.o.g. that $M=\{u_1v_1,u_2v_3,u_4v_4\}$.

\medskip\noindent {\bf (b)} $x$ and $y$ are adjacent nodes of $C_1$. Consider $F= (E(C_1)\setminus \{xy\}) \cup (E(C_2)\setminus \{u_1u_2\})\cup\{u_1v_1,u_2v_3\}$. Clearly $F\cup \{xy\}$ is 2EC and spans all the considered nodes.

In the following we can assume $\{x,y\}\neq \{v_1,v_3\}$, since $\{x,y\} = \{v_1,v_4\}$ is a symmetric case. 

\medskip\noindent {\bf (c)} $\{x,y\}\cap\{v_1,v_3\}=\emptyset$. The case $\{x,y\}=\{v_4,v_5\}$ is excluded by case (b). Thus, $\{x,y\}=\{v_2,v_4\}$ or $\{x,y\}=\{v_2,v_5\}$ may occur. In both cases we will not use the edge $v_4u_4$ in $F$. Under tris restriction, the two mentioned cases are symmetric hence assume w.l.o.g. $\{x,y\}=\{v_2,v_5\}$. In this case $F=(E(C_1)\setminus \{v_1v_5\})\cup (E(C_2)\setminus \{u_1u_2\})\cup\{u_1v_1,u_2v_3\}$ satisfies the claim. 

\medskip\noindent {\bf (d)} $|\{x,y\}\cap\{v_1,v_3\}|=1$. W.l.o.g. $x\in \{v_1,v_3\}$, therefore either $(x,y)=(v_1,v_4)$ or $(x,y)=(v_3,v_5)$ by case (b). In both cases we will not use the edge $u_4v_4$ in $F$. Under this restriction, the two mentioned cases are symmetric hence assume w.l.o.g. $(x,y)=(v_1,v_4)$. In this case  $F=(E(C_1)\setminus \{v_3v_4\})\cup (E(C_2)\setminus \{u_1u_2\})\cup\{u_1v_1,u_2,v_3\}$ satisfies the claim.
\end{proof}

\begin{lemma}\label{lem:3componentsC1_5cycle}
Let $C_1$, $C_2$ and $C_3$ be 2EC components of $S$ such that $C_1$ is a $5$-cycle, $C_2$ and $C_3$ are $4$-cycles, and there is a 3-matching $M_{12}$ (resp., $M_{13}$) between $C_1$ and $C_2$ (resp., $C_3$).
Then in polynomial time one can compute an $S'$ satisfying the conditions (1), (2) and (3) of Lemma \ref{lem:caseA:gluing}.  
\end{lemma}
\begin{proof}
 Let $C_1=v_1,v_2,v_3,v_4,v_5,v_1$ and $C_2=z_1,z_2,z_3,z_4,z_1$. Notice that $M_{12}$ must contain two edges with adjacent nodes in $C_2$, say $v_iz_1,v_jz_2\in M_{12}$. We apply Lemma \ref{lem:c5c4to9} to the tuple $(C_1,C_3,v_i,v_j)$, hence obtaining a set $F$ of $9$ edges such that $F\cup \{v_iv_j\}$ is a 2EC spanning subgraph of $V(C_1)\cup V(C_3)$. As a consequence $F':=F\cup E(C_2)\cup \{v_iz_1\}\cup \{v_jz_2\}\setminus \{z_1z_2\}$ induces a 2EC component $C'$ spanning the nodes $V(C_1)\cup V(C_2)\cup V(C_3)$. Define $S'=S\setminus (E(C_1)\cup E(C_2)\cup E(C_3))\cup F'$. Notice that $S'$ satisfies (1) and (2) trivially, and satisfies (3) since
 \begin{align*}
 cost(S)-cost(S') & =|S|-|S'|+\credit(C_1)+\credit(C_2)+\credit(C_3)-\credit(C')\\
 & =13-14+\frac{3}{10}5+\frac{3}{10}4+\frac{3}{10}4-2>0
 \end{align*}
 \end{proof}

\subsubsection{When the Component Graph is a Tree}
\label{sec:caseA:tree}

Assuming that no local merge of the above kind can be performed, in this section we show how to obtain the desired $S'$ when $\hat{G}_S$ (which is a connected graph) is a tree. A critical observation here is that each edge $C_1C_2\in E(\hat{G}_S)$ naturally induces a partition $(V_1,V_2)$ of the node set (with both sides containing at least $3$ nodes). The 3-matching Lemma \ref{lem:matchingOfSize3} then guarantees that there is a 3-matching between $V_1$ and $V_2$, hence between $C_1$ and $C_2$ by the definition of $\hat{G}_S$. The next lemma follows immediately:
\begin{lemma}\label{lem:treeFiltering}
Assume that the conditions of Lemmas \ref{lem:4cycle}-\ref{lem:3componentsC1_5cycle} do not hold and $\hat{G}_S$ is a tree. Then, for any edge $C_1C_2$ and any two edges $C'_1C'_2,C'_1C'_3\in E(\hat{G}_S)$: 
\begin{enumerate}\itemsep0pt
\item\label{lem:treeFiltering:1} If $C_1$ is a $4$-cycle, then $C_2$ is a $5$-cycle.
\item\label{lem:treeFiltering:2} If $C_1$ is a $5$-cycle then $C_2$ is a $4$-cycle.
\item\label{lem:treeFiltering:3} If $C_1$ is a triangle, then $C_2$ is a $6$-cycle.
\item\label{lem:treeFiltering:4} If $C_1$ is a large then $C_2$ is not large.
\item\label{lem:treeFiltering:5} If $C'_1$ is a $4$-cycle, then $C'_2$ and $C'_3$ cannot be both $5$-cycles.
\item\label{lem:treeFiltering:6} If $C'_1$ is a $5$-cycle, then $C'_2$ and $C'_3$ cannot be both $4$-cycles.

\end{enumerate}
\end{lemma}

\begin{lemma}\label{lem:caseA:treeCase}
Assume that the conditions of Lemmas \ref{lem:4cycle}-\ref{lem:3componentsC1_5cycle} do not hold and $\hat{G}_S$ is a tree. Then in polynomial time one can compute an $S'$ satisfying the conditions (1), (2) and (3) of Lemma \ref{lem:caseA:gluing}.
\end{lemma}
\begin{proof}
Assume by contradiction that $S$ contains a component $C$ which is a $4$-cycle or a $5$-cycle. By Lemma \ref{lem:treeFiltering}.\ref{lem:treeFiltering:1}-\ref{lem:treeFiltering:2} all the neighbors of a $4$-cycle (resp., $5$-cycle) $C'$ are $5$-cycles (resp., $4$-cycles). In particular, $S$ must consist only of $4$-cycles and $5$-cycles. Notice that $S$ must contain at least $3$ components (otherwise $G$ would contain at most $9$ nodes), hence there exists a node in $\hat G_S$ of degree at least $2$. However a $5$-cycle $C_1$ cannot have degree at least $2$ in $\hat G_S$, otherwise $C_1$ would have $2$ neighbouring $4$-cycles $C_2$ and $C_3$ which is excluded by Lemma \ref{lem:treeFiltering}.\ref{lem:treeFiltering:6}.  Similarly, a $4$-cycle $C_1$ cannot have degree at least $2$ by Lemma \ref{lem:treeFiltering}.\ref{lem:treeFiltering:5}. This leads to the desired contradiction. 

Thus the components of $S$ are triangles, 6-cycles or large. Assume by contradiction that there is \emph{not} a $6$-cycle in $S$. Then the existence of a triangle would contradict Lemma \ref{lem:treeFiltering}.\ref{lem:treeFiltering:2}. However in the complementary case all the components of $S$ would be large, contradicting \ref{lem:treeFiltering}.\ref{lem:treeFiltering:3}. Thus we can assume that $S$ contains a $6$-cycle $C$. By the 3-matching Lemma \ref{lem:matchingOfSize3} and the assumptions, for each neighbor $C'$ of $C$ in $\hat{G}_S$, there is a 3-matching between $C$ and $C'$. We distinguish three subcases:  

\medskip\noindent {\bf (a)} For some triangle component $C_1$ of $S$ adjacent to $C$ in $\hat{G}_S$, there exists a 2-matching $M=\{u_1v_1,u_2v_2\}$ with $u_i\in C_1$ and $v_j\in C$ such that $v_1$ and $v_2$ are adjacent nodes of $C$. Consider $S'=S\setminus \{v_1v_2,u_1u_2\}\cup \{u_1v_1,u_2v_2\}$. The same argument as in the proof of Lemma \ref{lem:4cycle} leads to a contradiction, hence this case cannot happen.


\medskip\noindent {\bf (b)} For some 6-cycle or large component $C_1$ of $S$ adjacent to $C$ in $\hat{G}_S$, there exists two edges $u_1v_1,u_2v_2$ with $u_i\in C_1$ and $v_j\in C$ such that $v_1$ and $v_2$ are adjacent nodes of $C$. Consider $S'=S\setminus \{v_1v_2\}\cup \{u_1v_1,u_2v_2\}$, which trivially satisfies (1) and (2). Let $C'$ be the 2EC component in $S'$ with node set $V(C)\cup V(C_1)$. Notice that $S'$ satisfies (3) since
\begin{align*}
cost(S)-cost(S')=|S|-|S'|+cr(C)+cr(C_1)-cr(C')\geq -1+\frac{3}{10}6+\frac{3}{10}6-2>0.
\end{align*}    
 
\medskip\noindent {\bf (c)} None of the above cases applies. Let $C=v_1,v_2,v_3,v_4,v_5,v_6,v_1$. 
Notice that for each neighbour $C'$ of $C$ in $\hat{G}_S$ there is a 3-matching between $C$ and $C'$. In particular $C'$ is adjacent (in $G$) to at least $3$ nodes in $C$. Furthermore the subset of nodes in $V(C)$ adjacent to $C'$ are not adjacent in $C$, namely they are either $V_{odd}:=\{v_1,v_3,v_5\}$ or $V_{even}:=\{v_2,v_4,v_6\}$. This allows us to partition the neighbours of $C$ in $\hat G_S$ into the components $\calC_{odd}$ and $\calC_{even}$ adjacent to $V_{odd}$ and $V_{even}$, resp. We distinguish two subcases:

\medskip\noindent {\bf (c.1)} One of the two sets $\calC_{odd}$ and $\calC_{even}$ is empty, say $\calC_{odd}=\emptyset$. Notice that any neighbour of $\{v_1,v_3,v_5\}$ in $G$ must belong to $C$ in this case. Assume that $v_1v_3\in E(G)$. Consider any $C_1\in \calC_{even}$ (which cannot be empty). Let $u_2v_2,u_4v_4\in E(G)$, with $u_2,u_4\in V(C_1)$ and $u_2\neq u_4$. If $C_1$ is \emph{not} a triangle (hence it is a 6-cycle or large by the previous cases), let $S':=S\cup \{u_2v_2,u_4v_4,v_1v_3\}\setminus \{v_1v_2,v_3v_4\}$. Clearly $S'$ satisfies (1) and (2), and it satisfies (3) since 
\begin{align*}
cost(S)-cost(S')\geq -1+\frac{3}{10}6+\frac{3}{10}6-2>0. 
\end{align*}
If $C_1$ is a triangle, consider 
$S':=S\cup \{u_2v_2,u_4v_4,v_1v_3\}\setminus \{v_1v_2,v_3v_4,u_2u_4\}$ (i.e., the same $S'$ as before, but without the edge $u_2u_4$). The same argument as in the proof of Lemma \ref{lem:4cycle} leads to a contradiction, hence this case cannot happen.

A symmetric argument holds if $v_1v_5\in E(G)$ or $v_3v_5\in E(G)$, hence we next assume $v_1v_3,v_1v_5,v_3v_5\notin E(G)$. This however implies that any feasible solution must include at least $6$ edges with endpoints in $V(C)$, making $C$ a contractible subgraph. This contradicts the assumption that $G$ is structured. 

\medskip\noindent {\bf (c.2)} Let $C_{odd}\in \calC_{odd}$ and $C_{even}\in \calC_{even}$. Let $M_1=\{u_1v_1,u_2v_3,u_3v_5\}$ be a $3$-matching between $C$ and $C_{odd}$, with $u_i\in C_{odd}$. Similarly, let 
$M_2=\{w_1v_2,w_2v_4,w_3v_6\}$ be a $3$-matching between $C$ and $C_{even}$, with $w_i\in C_{even}$. We initially set $S'=S\setminus \{v_1v_2,v_3v_4\}\cup \{u_1v_1,u_3v_3,w_1v_2,w_2v_4\}$. Then, if $C_{odd}$ (resp., $C_{even}$) is a triangle, we also remove $u_1u_3$ (resp., $w_1w_2$). In any case, $S'$ satisfies (1) and (2). In particular it contains a 2EC component $C'$ spanning the nodes set $V(C)\cup V(C_{odd})\cup V(C_{even})$. If $C_{odd}$ and $C_{even}$ are both $6$-cycles or large (in any combination), one has 
\begin{align*}
cost(S)-cost(S')& =|S|-|S'|+cr(C)+cr(C_{odd})+cr(C_{even})-cr(C')\\
& \geq -2+\frac{3}{10}6+\frac{3}{10}6+\frac{3}{10}6-2>0. 
\end{align*}
If exactly one of the two components, say $C_{even}$, is a triangle, then 
\begin{align*}
cost(S)-cost(S')& =|S|-|S'|+cr(C)+cr(C_{odd})+cr(C_{even})-cr(C')\\
& \geq -1+\frac{3}{10}6+\frac{3}{10}6+1-2>0. 
\end{align*}
If both the components are triangles, one has 
\begin{align*}
cost(S)-cost(S')& =|S|-|S'|+cr(C)+cr(C_{odd})+cr(C_{even})-cr(C')\\
& \geq \frac{3}{10}6+1+1-2>0. 
\end{align*}
In all the cases $S'$ satisfies (3). 
\end{proof}

\subsubsection{When the Component Graph is Not a Tree}
\label{sec:caseA:nonTree}

It remains to consider the case where $\hat G_S$ is not a tree. In the rest of this section we will prove the following lemma. 
\begin{lemma}\label{lem:caseA:nonTreeCase}
Assume that the conditions of Lemmas \ref{lem:4cycle}-\ref{lem:3componentsC1_5cycle} do not hold and $\hat{G}_S$ is not a tree. Then in polynomial time one can compute an $S'$ satisfying the conditions (1), (2) and (3) of Lemma \ref{lem:caseA:gluing}.
\end{lemma}


We first need some more notation. Let $\calC(S^*)$ be a subset of $\calC(S)$ obtained as follows: iteratively remove from $\calC(S)$ any component which has degree $1$ in $\hat{G}_S$, and update $\hat{G}_S$ consequently. Let $S^*$ be the set of edges in $\calC(S^*)$, and $\hat{G}_{S^*}$ the subgraph of $\hat{G}_{S}$ induced by $S^*$. Notice that by assumption $\hat{G}_{S^*}$ is not empty and by construction each node in $\hat{G}_{S^*}$ has degree at least $2$. We also denote by $V^*\subseteq V$ the nodes set of $\calC(S^*)$. Notice that each removed component $C$ belongs to a subtree of $\hat{G}_S$ rooted at some component $C^*$ of $S^*$: we say that $C$ is rooted at $C^*$.


We start by showing a variant of the 3-matching Lemma restricted to $S^*$. 
\begin{lemma}\label{lem:restricted3matching}
Let $\calC_1,\calC_2\neq \emptyset$ be a partition of $\calC(S^*)$, and let $V_i$ be the node set of $\calC_i$. Then there is a $3$-matching between $V_1$ and $V_2$.
\end{lemma}
\begin{proof}
Let $\calC'_i$ be the (removed) components in $\calC(S)\setminus \calC(S^*)$ which are rooted at some component in $\calC_i$, and $V'_i$ be the corresponding node set. The 3-matching Lemma \ref{lem:matchingOfSize3} guarantees that there is a $3$-matching $M$ between $V_1\cup V'_1$ and $V_2\cup V'_2$. However by construction nodes in $V'_1$ (resp., $V'_2$) are not adjacent to nodes in $V_2\cup V'_2$ (resp., $V_2\cup V'_2$). Consequently the endpoints of $M$ must belong to $V_1\cup V_2$. The claim follows. 
\end{proof}
In order to prove Lemma \ref{lem:caseA:nonTreeCase}, we introduce the notion of \emph{gluing path} of $S^*$ which has some similarities with the notion of nice path (Definition \ref{def:nicePathCycle} in Section \ref{sec:manyTriangles}), however with some critical differences (see Figure \ref{fig:GluingPath}). 
\begin{definition} 
A gluing path $P^*$ of $S^*$ is a sequence of edges with endpoints in distinct components of $S^*$ such that:
(1) $P^*$ induces a simple path in the graph obtained by contracting each connected component of $S^*$ into a single node; (2) If $e_1$ and $e_2$ are two edges of $P^*$ incident to a component $C^*\in \calC(S^*)$ which is a triangle or a 4-cycle, then the two endpoints of these edges in $V(C^*)$ are adjacent in $C^*$. 
\end{definition}
The intuition behind this definition is as follows (see Figure \ref{fig:GluingPath}). Suppose that we are given a gluing path $P_\ell=e_1,\ldots,e_\ell$, with $e_i=out_{i-1}in_i$. Let $C_{i-1}$ and $C_i$ be the two components incident to $e_i$. Notice that $C_i$ contains $in_i$ for $2\leq i\leq \ell$ and $out_i$ for $1\leq i\leq \ell-1$. Suppose now that we are also given an edge $e'=v'u'$ with $v'\in C_\ell$ and $u'\in C_{\ell'}$ for some $\ell'<\ell$ (in Figure \ref{fig:GluingPath} one has $\ell'=0$). Consider $S'=S\cup \{e_{\ell'+1},\ldots,e_\ell,e'\}$. Notice that $S'$ contains a 2EC component $C'$ spanning the nodes $V(C_{\ell'})\cup \ldots \cup V(C_{\ell})$. 
This construction can be refined as follows (which motivated the definition of gluing path). For each component $C_j$, $\ell'< j< \ell$, which is a triangle or a 4-cycle, remove from $S'$ (and $C'$) the edge $in_jout_j$. Notice that $C'$ remains 2EC. Observe also that $S'$ satisfying the conditions (1) and (2) of Lemma \ref{lem:caseA:gluing}. 

It remains to check condition (3) of Lemma \ref{lem:caseA:gluing}, namely $cost(S')\leq cost(S)$. Our goal is to show that $cost(S)-cost(S')=|S|-|S'|+\sum_{j=\ell'}^{\ell}cr(C_j)-cr(C')\geq 0$. Excluding the $-2$ coming from $e_\ell$ and $e'$, and the $-2$ coming from $cr(C')=2$, we can decompose the above difference as follows. We associate to each pair $(C_j,e_{j+1})$, $\ell'< j< \ell$, the following contribution: $-1$ due to $e_{j+1}$, $+1$ due to $in_jout_j$ if $C_j$ is a triangle or 4-cycle, and $+cr(C_j)$. Hence the neat contribution of the pair is $+1$, $+\frac{12}{10}$, $+\frac{5}{10}$,  $+\frac{8}{10}$ and $+1$ if $C_j$ is a triangle, $4$-cycle, $5$-cycle, $6$-cycle, or large, resp. In particular these contributions are always positive, which implies that $cost(S)\geq cost(S')$ trivially if $\ell'$ is sufficiently smaller than $\ell$. In the remaining cases we need to be more careful in the analysis. In particular, we will need to carefully choose $e'$. Furthermore, sometimes we will need to remove one edge from $C_{\ell'}$ and/or $C_\ell$. In one case we will need also to involve some component not in $S^*$. We remark that in the above construction the new 2EC component $C'$ that we create always contains at least $7$ nodes, hence Invariant \ref{inv:2ECcomponent} is maintained by a simple induction.

In more details, the construction proceeds as follows: we start with a gluing path $P_\ell$ of \emph{length} $|P_\ell|=\ell\geq 1$. The base case $\ell=1$ is obtained by taking any component $C_0\in \calC(S^*)$, and selecting any edge $e_1$ with one endpoint in $C_0$ and the other endpoint in some other component $C_1\in \calC(S^*)$ (notice that $e_1$ must exist since $C_0$ is not an isolated node in $\hat G_{S^*}$). Given such a $P_\ell$, we either build a new gluing path $P_{\ell+1}$ of length $\ell+1$ (which replaced $P_\ell$ in the next steps), or we build an $S'$ satisfying the properties of Lemma \ref{lem:caseA:nonTreeCase}. Clearly we obtained a desired $S'$ within a polynomial number of steps (since $\ell\leq |\calC(S^*)|-1$).

Consider the given gluing path $P_\ell$ with the above notation. We say that $P_\ell$ ends at $C_\ell$. We sometimes will need to modify $P_\ell$ as in the next lemma.  
\begin{lemma}\label{lem:change_e_ell}
Given a gluing path $P_\ell=e_1,\ldots,e_\ell$ ending at $C_\ell$, in polynomial time one can compute a gluing path $P'_\ell=e_1,\ldots,e'_\ell$ ending at $C_\ell$ and a 3-matching $M_\ell$ between $V(C_\ell)$ and $V^*\setminus V(C_\ell)$ with $e'_\ell\in M_\ell$.  
\end{lemma}
\begin{proof}
By the 3-matching Lemma \ref{lem:restricted3matching}, there exists a $3$-matching $M_\ell$ between $V(C_\ell)$ and $V^*\setminus V(C_{\ell})$. Suppose that $M_{\ell}$ is not incident to $out_{\ell-1}$. In this case we add $e_\ell$ to $M_\ell$ and remove some other edge from $M_\ell$ giving priority to the edge containing $in_\ell$ if any. Then $P'_\ell=P_\ell$ and (the modified) $M_\ell$ satisfy the claim. Otherwise, let $e'_\ell$ be the edge in $M_\ell$ incident to $out_{\ell-1}$. In this case $P'_\ell=e_1,\ldots,e'_\ell$ and $M_\ell$ satisfy the claim.  
\end{proof}

In the next few lemmas we perform a case analysis based on the type of $C_{\ell}$, namely if it is an $i$-cycle for some $3\leq i\leq 6$ or it is large (i.e., if it contains at least $7$ edges). We start with the case that $C_\ell$ is a $5$-cycle, the most complicated case in our analysis. We need to distinguish between the subcases when $C_\ell$ has an adjacent component in $S\setminus S^*$ or not w.r.t. $\hat{G}_S$. 
\begin{lemma}\label{lem:expand5cycle_root}
Suppose that $C_\ell$ is a $5$-cycle and it is adjacent to some component $C'_\ell$ in $S\setminus S^*$. Then in polynomial time one can either compute a gluing path $P_{\ell+1}$ of length $\ell+1$ or a solution $S'$ satisfying the conditions (1), (2) and (3) of Lemma \ref{lem:caseA:gluing}.
\end{lemma}
\begin{proof}
By construction $C'_\ell C_\ell$ is a bridge edge in $\hat G_S$. Hence by the 3-matching Lemma \ref{lem:matchingOfSize3} there exists a 3-matching $M'_\ell$ between $V(C'_\ell)$ and $V(C_\ell)$. Since by assumption the conditions of Lemma~\ref{lem:5cycle} do not hold, it must be the case that $C'_\ell$ is a $4$-cycle. 

Let us apply Lemma \ref{lem:change_e_ell} to $P_\ell$, in order to obtain a pair $(P'_\ell,M_\ell)$ as in the claim, and set  $P_\ell$ to $P'_\ell$. Suppose that there exists an edge $e_{\ell+1}$ with one endpoint in $V(C_\ell)$ and the other endpoint in $V^*\setminus (V(C_0)\cup \ldots \cup V(C_{\ell}))$. Then $P_{\ell+1}=P_\ell,e_{\ell+1}$ satisfies the claim. 

Otherwise, all the edges in $M_\ell$ must have one endpoint in $V(C_\ell)$ and the other one in $V(C_0)\cup \ldots \cup V(C_{\ell-1})$. We distinguish two cases:

\medskip\noindent {\bf (a)} All the edges in $M_\ell$ are between $V(C_\ell)$ and $V(C_{\ell-1})$. Since by assumption the conditions of Lemma~\ref{lem:5cycle} do not hold, $C_{\ell-1}$ must be a $4$-cycle. However in this case the triple $(C_\ell,C_{\ell-1},C'_\ell)$ satisfies the conditions of Lemma \ref{lem:3componentsC1_5cycle}, a contradiction. Hence this case cannot happen. 

\medskip\noindent {\bf (b)} There exists an edge $e'=u'v'\in M_\ell$ with $u'\in V(C_\ell)\setminus \{in_\ell\}$ and $v'\in V(C_{\ell'})$, $\ell'\leq \ell-2$. In this case we invoke Lemma \ref{lem:c5c4to9}  on the pair $(C_\ell,C'_\ell)$ to obtain a subset $F$ of $9$ edges such that $F\cup \{in_\ell u'\}$ is a 2EC spanning subgraph of $V(C_\ell)\cup V(C'_\ell)$. Let us initially set $S'=S\setminus (E(C_\ell)\cup E(C'_\ell))\cup F\cup \{e_{\ell'},\ldots,e_\ell,e'\}$. We then remove from $S'$ any edge $in_{j}out_{j}\in E(C_j)$, $\ell'+1\leq j\leq \ell-1$, whenever $C_j$ is a triangle or a $4$-cycle (recall that $in_{j}$ and $out_j$ must be adjacent in $C_j$ since $P_\ell$ is a gluing path). Let $n'$ be the number of such removed edges and $n''=\ell-\ell'-1-n'$.  

Observe that $S'$ satisfies (1) and (2). In particular, it contains a 2EC component $C'$ spanning the nodes $V(C_{\ell'})\cup\ldots \cup V(C_\ell)\cup V(C'_{\ell})$.
One has $|S'|=|S|+n''+2$. Furthermore 
\begin{align*}
cr(S)-cr(S') & =cr(C'_{\ell})+cr(C_\ell)+cr(C_{\ell'})+\sum_{j=\ell'+1}^{\ell-1}cr(C_j)-cr(C')\\
& \geq \frac{3}{10}4+\frac{3}{10}5+1+n'+\frac{3}{10}5n''-2.
\end{align*}
Altogether $cost(S)-cost(S')\geq \frac{1}{2}n''+n'-\frac{3}{10}>0$, hence $S'$ satisfies (3).
\end{proof}

\begin{lemma}\label{lem:expand5cycle_nonRoot}
Suppose that $C_\ell$ is a $5$-cycle and it is \emph{not} adjacent to any component in $S\setminus S^*$. Then in polynomial time one can either compute a gluing path $P_{\ell+1}$ of length $\ell+1$ or a solution $S'$ satisfying the conditions (1), (2) and (3) of Lemma \ref{lem:caseA:gluing}.
\end{lemma}
\begin{proof}
Let us apply Lemma \ref{lem:change_e_ell} to $P_\ell$, in order to obtain a pair $(P'_\ell,M_\ell)$ as in the claim, and set  $P_\ell$ to $P'_\ell$. Suppose that there exists an edge $e'=u'v'$ with $u'\in V(C_\ell)$ and $v'\in V^*\setminus (V(C_0)\cup \ldots \cup V(C_{\ell}))$. In this case $P_{\ell+1}:=P_\ell,e'$ satisfies the claim. 

Otherwise, let us prove the following technical claim, which exploits a simple case analysis and the fact that $G$ is structured.
\begin{claim}
There exists a $2$-matching $M=\{wout_{\ell-1},v'u'\}$, with $v',w\in V(C_\ell)$, between $V(C_\ell)$ and $V^*\setminus V(C_\ell)$ such that there exists a $v'$-$w$ Hamiltonian path $H_\ell$ restricted to $V(C_\ell)$. 
\end{claim}
\begin{proof}
Let $C_\ell=v_1,v_2,v_3,v_4,v_5,v_1$. Consider the 3-matching $M_\ell=\{e_\ell,e_3,e_4\}$. W.l.o.g. assume $e_\ell=v_1out_{\ell-1}$. 
If one of the nodes $v_2$ or $v_5$, say $v_2$, is adjacent in $G$ to some node $u_2\in V^*\setminus (V(C_\ell)\cup \{out_{\ell-1}\})$, then $M=\{e_\ell,v_2u_2\}$ satisfies the claim. We next assume that $v_2$ and $v_5$ are adjacent only to nodes in $V(C_\ell)\cup \{out_{\ell-1}\}$. Notice that by the previous case the endpoints of $e_3$ and $e_4$ are distinct from $v_2$ and $v_5$.
W.l.o.g. assume $e_3=v_3u_3$ and $e_4=v_4u_4$. If one of the nodes $v_2$ and $v_5$, say $v_2$, is adjacent to $out_{\ell-1}$, then $M=\{v_2out_{\ell-1},e_3\}$ satisfies the claim with $H_\ell=v_2,v_1,v_5,v_4,v_3$. If $v_2v_5\in E(G)$, then $M=\{e_\ell,e_3\}$ satisfies the claim with $H_\ell=v_1,v_2,v_5,v_4,v_3$. Otherwise every feasible solution must contain at least $4$ distinct edges incident to $v_2$ or $v_5$ with the other endpoint in $C_\ell$. Hence $C_\ell$ is a contractible subgraph of size $5$, a contradiction.
\end{proof}
We next update $P_\ell$ such that its final edge $e_\ell$ is the first edge in the 2-matching $M$ (in particular, $w=in_\ell$) guaranteed by the above lemma. Notice that $|E(H_\ell)|=4=|E(C_\ell)|-1$. We let $e'=v'u'$ be the remaining edge in $M$. W.l.o..g. assume $in_{\ell}=v_1$. By the previous case analysis we know that $u'\in V(C_{\ell'})$ for some $0\leq \ell'<\ell$. We distinguish a few cases:

\medskip\noindent {\bf (a)} $\ell'\leq \ell-2$. Define $S'=S\setminus E(C_\ell)\cup H_\ell \cup \{e_{\ell'},\ldots,e_\ell,e'\}$. Remove from $S'$ any edge $in_j out_j$ for $\ell'+1\leq j\leq \ell-1$ such that $C_j$ is a triangle or a $4$-cycle. Observe that $S'$ trivially satisfies (1) and (2). With the usual definition of $C'$, $n'$, and $n''$, $S'$ also satisfies (3) since 
\begin{align*}
cost(S)-cost(S') & =|S|-|S'|+cr(C_{\ell'})+cr(C_{\ell})+\sum_{j=\ell'+1}^{\ell-1}cr(C_j)-cr(C')\\ 
& \geq -(n''+2-1)+1+\frac{3}{10}5+n'+\frac{3}{10}5n''-2=\frac{1}{2}n''+n'-\frac{1}{2}\geq 0.
\end{align*}

\medskip\noindent {\bf (b)} $\ell'=\ell-1$ and $C_{\ell-1}$ is a 5-cycle, 6-cycle, or large.  Define $S'$ and $C'$ as in case (a) (in particular $S'$ satisfies (1) and (2)). $S'$ satisfies (3) since 
\begin{align*}
cost(S)-cost(S') & =|S|-|S'|+cr(C_{\ell-1})+cr(C_{\ell})-cr(C')\\ 
& \geq -(2-1)+\frac{3}{10}5+\frac{3}{10}5-2= 0.
\end{align*}

\medskip\noindent {\bf (c)} $\ell'=\ell-1$ and $C_{\ell-1}$ is a triangle. In this case consider $S'=S\setminus E(C_\ell)\cup H_\ell \setminus \{u'out_{\ell-1}\} \cup \{e_\ell,e'\}$, which satisfies (1) and (2). With the usual definition of $C'$, one has that $S'$ satisfies (3) since 
\begin{align*}
cost(S)-cost(S') & =|S|-|S'|+cr(C_{\ell-1})+cr(C_{\ell})-cr(C')\\
& \geq -(2-2)+1+\frac{3}{10}5-2> 0.
\end{align*}

\medskip\noindent {\bf (d)} $\ell'=\ell-1$ and $C_{\ell-1}$ is a $4$-cycle. By the definition of $S^*$, there must exist an edge $e''=v''u''$ with $v''\in V(C_\ell)$ and $u''\in V(C_{\ell''})$ for some $\ell''\leq \ell-2$. 

\medskip\noindent {\bf (d.1)} $out_{\ell-1}$ and $u'$ are adjacent in $C_{\ell-1}$. Then essentially the same construction and analysis as in case (c) gives the desired $S'$. In particular one has 
\begin{align*}
cost(S)-cost(S')\geq -(2-2)+\frac{3}{10}4+\frac{3}{10}5-2>0.
\end{align*} 

We next assume w.l.o.g. that $C_{\ell-1}=u_1,u_2,u_3,u_4,u_1$ where $u_1=out_{\ell-1}$, $u_3=u'$, and $\ell-1=0$ or $in_{\ell-1}=u_2$. 

\medskip\noindent {\bf (d.2)} $v'=v_3$ (the case $v'=v_4$ is symmetric). If $v''\in \{v_2,v_3,v_5\}$, we can exploit the analysis of case (a) where the role of $e'$ is replaced by $e''$. This is because there exists a Hamiltonian $v''$-$v_1$ path over $V(C_\ell)$. Otherwise, consider the alternative gluing path $P'_\ell=e_1,\ldots,e_{\ell-1},e'$ of length $\ell$. If $v''\in \{v_1,v_4\}$, thanks to the existence of edge $e''$, we can apply to $P'_{\ell}$ the analysis of case (a).

\medskip\noindent {\bf (d.3)} $v'=v_2$ (the case $v'=v_5$ is symmetric). If $v''\in \{v_2,v_5\}$, we can exploit the analysis of case (a) where the role of $e'$ is replaced by $e''$. If $v''\in \{v_1,v_3\}$, consider the alternative gluing path $P'_\ell=e_1,\ldots,e_{\ell-1},e'$ of length $\ell$. Thanks to the existence of edge $e''$, we can apply to $P'_{\ell}$ the analysis of case (a). 

The remaining case is $v''=v_4$. If $v_3v_5\in E(G)$, then there exists a $v_1$-$v_4$ Hamiltonian path over $V(C_\ell)$, namely $v_1,v_2,v_3,v_5,v_4$. Hence again case (a) applies with $e'$ replaced by $e''$. Otherwise, assume by contradiction that $v_3$ and $v_5$ are not adjacent to any node outside $V(C_\ell)$. This would make $V(C_\ell)$ a contractible subgraph on $5$ nodes (at least $4$ edges with endpoints in $V(C_\ell)$ belong to every feasible solution), which is a contradiction since $G$ is structured. Thus there must exist an edge $e'''=v'''u'''$ with $v'''\in \{v_3,v_5\}$ and $u'''\notin V(C_{\ell})$. If $u'''\notin V(C_{\ell-1})$, again we can reduce to case (a) applied to $P_\ell$ or $P'_\ell$ with edge $e'''$ playing the role of $e'$. If $u'''\in \{u_2,u_4\}$, we can exploit the analysis of case (d.1) applied to $P_\ell$ (if $v'''=v_5$) or to $P'_\ell$ (if $v'''=v_3$). The remaining case is $u'''\in \{u_1,u_3\}$. Consider $P''_\ell=e_1,\ldots,e_{\ell-1},e'''$. We can apply to $P''_\ell$ the analysis of case (a) where the role of $e'$ is played by $e''$ (there is a $v'''$-$v_4$ Hamiltonian path over $V(C_\ell)$).
\end{proof}

In the next three lemmas we address the case when $C_\ell$ is a $4$-cycle, a 6-cycle or large, or a triangle, resp.

\begin{lemma}\label{lem:expand4cycle}
Suppose that $C_\ell$ is a $4$-cycle. Then in polynomial time one can either compute a gluing path $P_{\ell+1}$ of length $\ell+1$ or a solution $S'$ satisfying the conditions (1), (2) and (3) of Lemma \ref{lem:caseA:gluing}.
\end{lemma}
\begin{proof}
We distinguish a few subcases (each time excluding the previous cases):

\medskip\noindent {\bf (a)} There exists an edge $e'=v'u'$ with $v'\in V(C_\ell)$ adjacent to $in_{\ell}$ in $C_\ell$ and $u'\in V^*\setminus (V(C_0)\cup \ldots \cup V(C_\ell))$. Then $P_{\ell+1}:=P_\ell,e_{\ell+1}$ satisfies the claim. 

\medskip\noindent {\bf (b)} There exists an edge $e'=u'v'\in E(G)$ with $u'\in V(C_\ell)\setminus \{in_\ell\}$ adjacent to $in_\ell$ and $v'\in V(C_{\ell'})$, $\ell'\leq \ell-2$. Let $S'=S\cup \{e_{\ell'},\ldots,e_{\ell},e'\}\setminus \{in_\ell'u'\}$. Next remove from $S'$ any edge $in_j out_j$ with $\ell'+1\leq j\leq \ell-1$ where $C_j$ is a triangle or a $4$-cycle. Let $n'$ be the number of removed edges and $n''=\ell-\ell'-1-n'$. Observe that $S'$ satisfies conditions (1) and (2). In particular, $S'$ has a 2EC component $C'$ spanning the nodes $V(C_{\ell'})\cup \ldots \cup V(C_{\ell})$. Notice also that $|S'|=|S|+n''+1$. We further distinguish a few subcases (each time excluding the previous subcases):

\medskip\noindent {\bf (b.1)} $C_{\ell-1}$ is \emph{not} a 5-cycle.  If $C_{\ell-1}$ is a triangle or 4-cycle (which implies $n'\geq 1$), one has 
\begin{align*}
cost(S)-cost(S') & =-n''-1+cr(C_{\ell'})+cr(C_\ell)+\sum_{j=\ell'+1}^{\ell-1}ct(C_j)-cr(C')\\
& \geq -n''-1+1+\frac{3}{10}4+n'+\frac{3}{10}5n''-2=\frac{1}{2}n''+n'-\frac{8}{10}>0.  
\end{align*}
Otherwise $C_{\ell-1}$ is a 6-cycle or large, hence
\begin{align*}
& cost(S)-cost(S')  =-n''-1+cr(C_{\ell-1})+cr(C_{\ell'})+cr(C_\ell)+\sum_{j=\ell'+1}^{\ell-2}cr(C_j)-cr(C')\\
\geq & -n''-1+\frac{3}{10}6+1+\frac{3}{10}4+n'+\frac{3}{10}5(n''-1)-2=\frac{1}{2}n''+n'-\frac{1}{2}\geq 0.
\end{align*}
In both cases (3) holds.

\medskip\noindent {\bf (b.2)} $\ell'\leq \ell-3$. $S'$ satisfies (3) since 
\begin{align*}
cost(S)-cost(S') & =-n''-1+cr(C_{\ell'})+cr(C_\ell)+\sum_{j=\ell'+1}^{\ell-1}cr(C_j)-cr(C')\\
& \geq -n''-1+1+\frac{3}{10}4+n'+\frac{3}{10}5n''-2\geq \frac{1}{2}n''+n'-\frac{8}{10}>0, 
\end{align*}
where in the last inequality we used $n'+n''\geq 2$. We next assume $\ell'=\ell-2$.

\medskip\noindent {\bf (b.3)} $C_{\ell-1}$ is a 5-cycle and $v'$ is adjacent to $out_{\ell'}$. In this case modify $S'$ (and $C'$) by removing $v'out_{\ell'}$. Notice that $S'$ still satisfies (1) and (2). Furthermore it satisfies (3) since 
\begin{align*}
cost(S)-cost(S') & =-n''+cr(C_{\ell'})+cr(C_\ell)+\sum_{j=\ell'+1}^{\ell-1}ct(C_j)-cr(C')\\
& \geq -n''+1+\frac{3}{10}4+n'+\frac{3}{10}5n''-2=\frac{1}{2}n''+n'+\frac{1}{5}>0.
\end{align*}

\medskip\noindent {\bf (b.4)} $C_{\ell-1}$ is a 5-cycle and $v'$ is not adjacent to $out_{\ell'}$. Consider $P'_\ell=e_1,\ldots,e_{\ell'},e',e_\ell$. Notice that this is a gluing path of length $\ell$. This is trivial if $C_{\ell'}$ is not a triangle nor a 4-cycle. Otherwise, by Case (b.3), $v'=out_{\ell'}$ if $C_{\ell'}$ is a triangle and $v'=out_{\ell'}$ or $v'$ is the node opposite to $out_{\ell'}$ in $C_{\ell'}$ if $C_{\ell'}$ is a 4-cycle. In both cases $v'$ is adjacent to $in_{\ell'}$ unless the latter node is not defined (i.e., if $\ell'=0$). Hence $P'_\ell$ satisfies all the desired properties. Since $P'_{\ell}$ ends with a 5-cycle by construction, we can apply Lemma \ref{lem:expand5cycle_root} or \ref{lem:expand5cycle_nonRoot} to infer the claim. 

In the following cases we assume $\ell'=\ell-1$.

\medskip\noindent {\bf (c)} There exists an edge $e'=u'v'$ with $u'\in V(C_\ell)$ adjacent to $in_\ell$ in $C_\ell$ and $v'\in V(C_{\ell-1})$ adjacent to $out_{\ell-1}$ in $C_{\ell-1}$. Consider $S''=S\setminus \{in_\ell u',out_{\ell-1}v'\}\cup \{e_\ell,e'\}$. Notice that $S''$ satisfies (1) and (2). In particular $S''$ has a 2EC component $C'$ spanning the nodes $V(C_\ell)\cup V(C_{\ell-1})$. $S''$ also satisfies (3) since \begin{align*}
cost(S)-cost(S'') =cr(C_\ell)+cr(C_{\ell-1})-cr(C')\geq \frac{3}{10}4+1-2> 0
\end{align*}

\medskip\noindent {\bf (d)} All the edges $e'=u'v'$, with $u'\in V(C_\ell)\setminus \{in_\ell\}$ adjacent to $in_\ell$ in $C_\ell$ and $v'\notin V(C_\ell)$, have $v'\in V(C_{\ell-1})$ not adjacent to $out_{\ell-1}$ in $C_{\ell-1}$. Let $e'=u'v'$ be one such edge. Notice that such an edge must exist since there is a 3-matching between $V(C_\ell)$ and $V^*\setminus V(C_\ell)$. Let $C_\ell=v_1,v_2,v_3,v_4,v_1$ with $v_1=in_{\ell}$.   By the definition of $S^*$, one among $v_1$ and $v_3$ must be adjacent to some node in $V^*\setminus (V(C_\ell)\cup V(C_{\ell-1}))$. Consider $P'_\ell=e_1,\ldots,e_{\ell-1},e'$. This happens to be a gluing path of length $\ell$. In particular, if $C_{\ell-1}$ is a triangle or a 4-cycle, then $v'$ is adjacent to $in_{\ell-1}$ in $C_{\ell-1}$ or $\ell-1=0$. Observe that $u'$ is adjacent to both $v_1$ and $v_3$, and at least one of them is incident to an edge with the other endpoint in $V^*\setminus (V(C_\ell)\cup V(C_{\ell-1}))$. The claim then follows by applying to $P'_\ell$ the construction of cases (a) or (b). 
\end{proof}

\begin{lemma}\label{lem:expand6cycleLarge}
Suppose that $C_\ell$ is a $6$-cycle or large. Then in polynomial time one can either compute a gluing path $P_{\ell+1}$ of length $\ell+1$ or a solution $S'$ satisfying the conditions (1), (2) and (3) of Lemma \ref{lem:caseA:gluing}.
\end{lemma}
\begin{proof}
Suppose there exists an edge $e'=u'v'$ with $u'\in V(C_\ell)$ and $v'\in V^*\setminus (V(C_0)\cup \ldots \cup V(C_{\ell}))$. Then $P_{\ell+1}:=P_\ell,e_{\ell+1}$ satisfies the claim. We next assume that such $e'$ does not exist. 

By the definition of $S^*$ there exists an edge $e'=u'v'\in E(G)$ with $u'\in V(C_\ell)$ and $v'\in V(C_{\ell'})$ with $\ell'\leq \ell-2$. In case of ties, we choose $e'$ so that $\ell'$ is as small as possible, and as a second choice so that $v'$ is adjacent to $out_{\ell'}$ in $C_{\ell'}$ if possible.

We define $S'$ as follows. Set initially $S'=S\cup \{e_{\ell'+1},\ldots,e_{\ell},e'\}$. Then remove from $S'$ any edge $in_jout_j$ with $\ell'+1\leq j\leq \ell-1$ such that $C_j$ is a triangle or $4$-cycle. Let $n'$ be the number of such removed edges and $n''=\ell-\ell'-1+1-n'$. Observe that $S'$ satisfies (1) and (2). In particular, there is a component $C'$ spanning the nodes $V(C_{\ell'})\cup \ldots \cup V(C_\ell)$. Notice also that $|S'|=|S|+n''+2$ and $cr(C')=2$. 

We distinguish a few subcases, each time assuming that the previous cases do not apply:

\medskip\noindent {\bf (a)} $\ell'\leq \ell-3$. One has 

 \begin{align*}
& cost(S)-cost(S')  = |S|-|S'|+cr(C_{\ell'})+cr(C_{\ell})+\sum_{j=\ell'+1}^{\ell-1}cr(C_j)-cr(C')\\
\geq & -(n''+2)+cr(C_{\ell'})+\frac{3}{10}6+n'+\frac{3}{10}5n''-2=\frac{1}{2}n''+n'+cr(C_{\ell'})-\frac{22}{10}.
\end{align*}

Notice that $n'+n''\geq 2$. 
If $C_{\ell'}$ is not a triangle, then $\frac{1}{2}n''+n'+cr(C_{\ell'})\geq 1+\frac{3}{10}4=\frac{22}{10}$ and (3) holds. Similarly, if $n'\geq 1$ or $n''\geq 3$ since in that case $\frac{1}{2}n''+n'+cr(C_{\ell'})\geq \frac{3}{2}+1>\frac{22}{10}$. Therefore we can assume that $C_{\ell'}$ is a triangle, $n'=0$ and $n''=2$. Notice that $\ell'=\ell-3$ in this case. Observe that $C_{\ell-1}$ and $C_{\ell-2}$ are not triangles nor 4-cycles. If at least one of them is a 6-cycle or large, we can refine the above inequality to 
$$
cost(S)-cost(S')\geq -(n''+2)+1+\frac{3}{10}6+n'+\frac{3}{10}5(n''-1)+\frac{3}{10}6-2=\frac{1}{2}n''+n'-\frac{9}{10}>0,
$$ 
where we used $n''=2$. Thus we can assume that both $C_{\ell-1}$ and $C_{\ell-2}$ are $5$-cycles. 

Suppose now that $v'\neq out_{\ell'}$. In this case consider the same $S'$ as before, but removing the edge $v'out_{\ell}$ (this leads to a modified $C'$ which is still 2EC, in particular (1) and (2) still hold). With a similar notation we obtain 
$$
cost(S)-cost(S')\geq -(n''+2-1)+1+\frac{3}{10}6+n'+\frac{3}{10}5n''-2=\frac{1}{2}n''+n'-\frac{2}{10}>0.
$$

In the remaining case $v'= out_{\ell'}$, consider the alternative gluing path $P'_\ell=\\e_1,\ldots,e_{\ell-3},e',e_{\ell},e_{\ell-1}$ of length $\ell$. This path ends at $C_{\ell-2}$, which is a $5$-cycle. The claim follows by applying to $P'_\ell$ one of the Lemmas \ref{lem:expand5cycle_root} and \ref{lem:expand5cycle_nonRoot}. In the following cases we will assume $\ell'=\ell-2$, hence in particular $v'\in C_{\ell-2}$.

\medskip\noindent {\bf (b)} $v'$ is adjacent to $out_{\ell-2}$. Modify $S'$ (and $C'$) by removing the edge $v'out_{\ell-2}$. Notice that $C'$ remains 2EC, in particular (1) and (2) still hold. If $C_{\ell-dit{1}}$ is a triangle or a 4-cycle, one has 
\begin{align*}
cost(S)-cost(S') & =|S|-|S'|+cr(C_{\ell-2})+cr(C_{\ell-1})+cr(C_{\ell})-cr(C')\\
& \geq -1+1+1+\frac{3}{10}6-2>0. 
\end{align*}
Otherwise 
\begin{align*}
cost(S)-cost(S')\geq -2+1+\frac{3}{10}5+\frac{3}{10}6-2>0. 
\end{align*}
In both cases (3) holds.

\medskip\noindent {\bf (c)} $C_{\ell-2}$ is a 5-cycle, a 6-cycle or large.  

\medskip\noindent {\bf (c.1)} $C_{\ell-1}$ is a 4-cycle or a 5-cycle. Consider the alternative gluing path  $P'_\ell=\\e_1,\ldots,e_{\ell-2},e',e_{\ell}$ of length $\ell$. Notice that the last component of $P'_{\ell}$ is $C_{\ell-1}$, which is a 4-cycle or a 5-cycle. The claim follows by applying to $P'_\ell$ one of the Lemmas \ref{lem:expand5cycle_root}, \ref{lem:expand5cycle_nonRoot}, and \ref{lem:expand4cycle}. 

\medskip\noindent {\bf (c.2)} $C_{\ell-1}$ is a triangle. $S'$ satisfies (3) since
\begin{align*}
cost(S)-cost(S') & =|S|-|S'|+cr(C_{\ell-2})+cr(C_{\ell-1})+ cr(C_{\ell})-cr(C')\\
& \geq -2+\frac{3}{10}5+1+\frac{3}{10}6-2>0.
\end{align*}

\medskip\noindent {\bf (c.3)} $C_{\ell-1}$ is a 6-cycle, or large. $S'$ satisfies (3) since 
\begin{align*}
cost(S)-cost(S') & =|S|-|S'|+cr(C_{\ell-2})+cr(C_{\ell-1})+cr(C_{\ell})-cr(C')\\
& \geq -3+\frac{3}{10}5+\frac{3}{10}6+\frac{3}{10}6-2>0.
\end{align*}

\medskip\noindent {\bf (d)} $C_{\ell-2}$ is a triangle or a 4-cycle. Notice that by case (2), $v'$ is not adjacent to $out_{\ell-2}$, hence it must be adjacent to $in_{\ell-2}$ or $\ell-2=0$.

\medskip\noindent {\bf (d.1)} $C_{\ell-1}$ is a 4-cycle or a 5-cycle. The claim follows analogously to case (c.1) by considering the same $P'_{\ell}$ and applying  
Lemma \ref{lem:expand5cycle_root}, \ref{lem:expand5cycle_nonRoot}, or \ref{lem:expand4cycle}.


\medskip\noindent {\bf (d.2)} $C_{\ell-1}$ is a triangle. Suppose first that $C_{\ell-2}$ is a 4-cycle. Then $S'$ satisfies (3) since
\begin{align*}
cost(S)-cost(S') & =|S|-|S'|+cr(C_{\ell-2})+cr(C_{\ell-1})+cr(C_{\ell})-cr(C')\\
& \geq -2+\frac{3}{10}4+1+\frac{3}{10}6-2=0.
\end{align*}

Otherwise (i.e., $C_{\ell-2}$ is a triangle), by Lemma \ref{lem:restricted3matching}, there exists a 3-matching $M$ between $V(C_{\ell})\cup V(C_{\ell-1})$ and $V^*\setminus (V(C_{\ell})\cup V(C_{\ell-1}))$. By the previous cases $M$ does not contain an edge $e''=u''v''$ with $u''\in V(C_\ell)$ and $v''\in V(C_{\ell-2})$ adjacent to $out_{\ell-2}$. Hence $M$ must contain an edge $e''=u''v''$ with $u''\in V(C_{\ell-1})\setminus \{out_{\ell-1}\}$ and $v''\neq out_{\ell-2}$. If $v''\notin V(C_0)\cup \ldots \cup V(C_{\ell-2})$, then $P_{\ell+1}=(e_1,\ldots,e_{\ell-2},e',e_{\ell},e'')$ satisfies the claim. If $v''\in V(C_{\ell-2})$, then consider $S''=S\cup \{e_{\ell},e',e''\}\setminus \{v''out_{\ell-2},u''out_{\ell-1}\}$ and let $C''$ be the 2EC component spanning the nodes $V(C_{\ell-2})\cup V(C_{\ell-1})\cup V(C_{\ell})$. 
Clearly $S''$ satisfies (1) and (2), and it satisfies (3) since: 
\begin{align*}
cost(S)-cost(S'') & =|S|-|S''|+cr(C_{\ell-2})+cr(C_{\ell-1})+cr(C_\ell)-cr(C'')\\
& =-1+1+1+\frac{3}{10}6-2>0.
\end{align*}

Otherwise $v''\in V(C_{\ell''})$ for some $\ell''<\ell-2$. In this case let $S''=\\S\cup \{e_{\ell''+1},\ldots,e_{\ell-1},e',e_\ell,e''\}\setminus \{u''out_{\ell-1}\}$. Remove from $S''$ any edge $in_jout_j$ with $\ell''+1\leq j\leq \ell-2$ where $C_j$ is a triangle or a 4-cycle, and let $n'$ be the number of such removed edges. Notice that $S''$ satisfies (1) and (2). Observe that $n'\geq 1$ due to $C_{\ell-2}$. Let $n''=\ell-2-\ell''-n'$. Then one has $|S''|=|S|+n''+2$. Observe that $S''$ contains a 2EC component $C''$ spanning the nodes $V(C_{\ell''})\cup \ldots  \cup V(C_{\ell})$. Then $S'$ satisfies (3) since
\begin{align*}
cost(S)-cost(S'') & =|S|-|S''|+cr(C_{\ell''})+cr(C_{\ell-1})+cr(C_{\ell})+\sum_{j=\ell''+1}^{\ell-2} cr(C_{j})-cr(C'')\\
& \geq -(n''+2)+1+1+\frac{3}{10}6+n'+\frac{3}{10}5n''-2=\frac{1}{2}n''+n'-\frac{2}{10}> 0,
\end{align*}
where in the last inequality we used $n'\geq 1$.

\medskip\noindent {\bf (d.3)} $C_{\ell-1}$ is a $6$-cycle or large. By Lemma \ref{lem:restricted3matching}, there exists a 3-matching $M$ between $V(C_{\ell})\cup V(C_{\ell-1})$ and $V^*\setminus (V(C_{\ell})\cup V(C_{\ell-1}))$. Since by the previous cases $V(C_\ell)$ can be adjacent only to $V(C_{\ell-1})\cup \{out_{\ell-2}\}$, there must exist one edge $e''=u''v''\in E(G)$ with $u''\in V(C_{\ell-1})$ and $v''\neq out_{\ell-2}$. If there exists one such $e''$ with $v''\notin V(C_0)\cup \ldots \cup V(C_{\ell-2})$, then $P_{\ell+1}:=(e_1,\ldots,e_{\ell-2},e',e_{\ell},e'')$ satisfies the claim. 

Otherwise, consider $P'_{\ell}:=(e_1,\ldots,e_{\ell-2},e',e_{\ell})$. Notice that $P'_{\ell}$ is a gluing path of length $\ell$ ending at $C_{\ell-1}$. In particular, the endpoint $v'$ of $e'$ in $V(C_{\ell-2})$ is not adjacent to $out_{\ell-2}$, hence it must be adjacent to $in_{\ell-2}$. If there exists one such $e''$ with $v''\in V(C_{\ell''})$ and $\ell''\le \ell-3$, then we can apply case (a) to $P'_{\ell}$. Otherwise the $3$ edges of $M$ must be all incident to $V(C_{\ell-2})$, and the ones  incident to $V(C_\ell)$ (if any) cannot be adjacent to the two neighbours of $out_{\ell-2}$. Thus there exists one edge $e''=u''v''\in M$ with $u''\in V(C_{\ell-1})$ and $v''\in V(C_{\ell-2})$ adjacent to $out_{\ell-2}$. Then we can apply case (b) to $P'_{\ell}$.
\end{proof}

\begin{lemma}\label{lem:expand3cycle}
Suppose that $C_\ell$ is a triangle. Then in polynomial time one can either compute a gluing path $P_{\ell+1}$ of length $\ell+1$ or a solution $S'$ satisfying the conditions (1), (2) and (3) of Lemma \ref{lem:caseA:gluing}.
\end{lemma}
\begin{proof}
Let us apply Lemma \ref{lem:change_e_ell} to $P_\ell$, in order to obtain a pair $(P'_\ell,M_\ell)$ as in the claim, and set  $P_\ell$ to $P'_\ell$. Suppose there exists an edge $e_{\ell+1}$ with one endpoint in $V(C_\ell)\setminus \{in_\ell\}$ (in particular not incident to $e_\ell$) and the other in $V^*\setminus (V(C_0)\cup \ldots \cup V(C_{\ell}))$. Then $P_{\ell+1}:=P_\ell,e_{\ell+1}$ satisfies the claim.

Otherwise all the edges in $M_\ell$ have one endpoint in $V(C_\ell)$ and the other one in $V(C_0)\cup \ldots \cup V(C_{\ell-1})$. We distinguish two cases:

\medskip\noindent {\bf (a)} There exists an edge $e'=u'v'\in E(G)$ with $u'\in V(C_\ell)\setminus \{in_\ell\}$ and $v'\in V(C_{\ell'})$, $\ell'\leq \ell-2$. We further distinguish some subcases:

\medskip\noindent {\bf (a.1)} $C_{\ell'}$ is not a triangle nor a $4$-cycle. In this case consider $S'=S\setminus \{in_\ell u'\}\cup \{e_{\ell'},\ldots,e_\ell,e'\}$. Let us remove from $S'$ any edge of type $in_{j}out_{j}$ where $\ell'+1\leq j \leq \ell-1$ and $C_j$ is a triangle or a $4$-cycle. Let $n'$ be the number of such removed edges, and $n''=\ell-\ell'-1-n'$. Observe that $S'$ satisfies (1) and (2). In particular it contains a 2EC component $C'$ spanning the nodes $V(C_{\ell'})\cup \ldots \cup V(C_\ell)$. $S'$ satisfies (3) since 
\begin{align*}
cost(S)-cost(S') & =|S|-|S'|+cr(C_{\ell'})+cr(C_\ell)+\sum_{j=\ell'+1}^{\ell-1}cr(C_j)-cr(C')\\
& \geq -n''-1+\frac{3}{10}5+1+n'+\frac{3}{10}5n''-2=\frac{1}{2}n''+n'-\frac{1}{2}\geq 0.
\end{align*}

\medskip\noindent {\bf (a.2)} $C_{\ell'}$ is a triangle or a $4$-cycle and $C_{\ell'+1}$ is a triangle or a $4$-cycle. In this case we can take the same $S'$ as in case (2.a), which satisfies (1) and (2). With the same notation, we observe that in this case the difference is that $n'\geq 1$ (due to $C_{\ell'+1}$) and $cr(C_{\ell'})\geq 1$. Thus $S'$ satisfies (3) since 
\begin{align*}
cost(S)-cost(S') & =|S|-|S'|+cr(C_{\ell'})+cr(C_\ell)+\sum_{j=\ell'+1}^{\ell-1}cr(C_j)-cr(C')\\
& \geq -n''-1+1+1+n'+\frac{3}{10}5n''-2=\frac{1}{2}n''+n'-1\geq 0.
\end{align*}

\medskip\noindent {\bf (a.3)} $C_{\ell'}$ is a triangle or a $4$-cycle and $C_{\ell'+1}$ is not a triangle nor a $4$-cycle. If $v'$ is a neighbour of $in_{\ell'}$ in $C_{\ell'}$ or $\ell'=0$ (in which case $in_{\ell'}$ is not defined), consider $P'_\ell:=e_1,\ldots,e_{\ell'},e',e_{\ell},e_{\ell-1},\ldots,e_{\ell'+2}$. Observe that $P'_{\ell}$ is an alternative  gluing path of length $\ell$ whose ending component $C_{\ell'+1}$ is not a triangle nor a $4$-cycle. The claim then follows by applying one of the Lemmas \ref{lem:expand5cycle_root}, \ref{lem:expand5cycle_nonRoot}, and \ref{lem:expand6cycleLarge}. 

The remaining case is when $v'$ is not a neighbour of $in_{\ell'}$ in $C_\ell'$ and $\ell'\neq 0$. Observe that, since $P_\ell$ is a gluing path and $C_{\ell'}$ is a triangle or a $4$-cycle, then $v'$ must be adjacent to $out_{\ell'}$. Consider $S'=S\setminus \{in_\ell u',v' out_{\ell'}\}\cup \{e_{\ell'+1},\ldots,e_{\ell},e'\}$. Remove as usual from $S'$ any edge $in_{j}out_{j}$ where $\ell'+1\leq j \leq \ell-1$ and $C_j$ is a triangle or a $4$-cycle. $S'$ satisfies (1) and (2). With the usual notation for $n'$, $n''$ and $C'$, $S'$ satisfies (3) since 
\begin{align*}
cost(S)-cost(S') & =|S|-|S'|+cr(C_{\ell'})+cr(C_\ell)+\sum_{j=\ell'+1}^{\ell-1}cr(C_j)-cr(C')\\
& \geq -n''+1+1+n'+\frac{3}{10}5n''-2=\frac{1}{2}n''+n'\geq 0.
\end{align*}

\medskip\noindent {\bf (b)} All the edges in $M_\ell$ are between $V(C_\ell)$ and $V(C_{\ell-1})$ (notice that this must be the case when case (1) does not hold). Since by assumption the conditions of Lemma \ref{lem:triangle} do not hold, $C_{\ell-1}$ must be a $6$-cycle. Notice that by the definition of $S^*$, $C_\ell$ is adjacent to at least two components in $S^*$, hence to at least one component $\tilde{C}$ other than $C_{\ell-1}$. This implies that there exists an edge $e'=vz$ with $v\in V(C_\ell)$ and $z\in V(\tilde{C})$. More precisely, $v=in_\ell$ by the previous cases. In this case we choose any edge $e''\in M_\ell\setminus \{e_\ell\}$. 
If $\tilde{C}\notin \{C_1,\ldots,C_{\ell-2}\}$, then $P_{\ell+1}=e_1,\ldots,e_{\ell-1},e'',e'$ satisfies the claim. Otherwise $P'_{\ell}=e_1,\ldots,e_{\ell-1},e''$ is another gluing path of length $\ell$ whose ending component is $C_\ell$. We can now apply case (a) to $P'_\ell$.

\end{proof}

We now have all the ingredients to prove Lemma \ref{lem:caseA:nonTreeCase}.
\begin{proof}[Proof of Lemma \ref{lem:caseA:nonTreeCase}]
It follows easily from the previous discussion and Lemmas \ref{lem:expand5cycle_root}-\ref{lem:expand3cycle}.
\end{proof}

\section{A Refined Approximation}
\label{sec:refinedApproximation}

In this section we present the following refined version of Lemma \ref{lem:manyTriangles:main} for the case of many triangles. 
\begin{lemma}\label{lem:manyTriangles:refined}
Given a canonical minimum-size 2-edge-cover $H$ of a structured graph $G$ with $b|H|$ bridges and $t|H|$ edges belonging to triangle 2EC components. In polynomial time one can compute a $\frac{14}{9}-\frac{8}{27}t+\frac{4}{9}b$ approximate solution for 2-ECSS on $G$.
\end{lemma}
The proof of Theorem \ref{thr:refinedApx} follows easily.
\begin{proof}[Proof of Theorem \ref{thr:refinedApx}]
It is sufficient to describe a $\frac{118}{89}$-approximation algorithm for the case of a structured graph $G$. The claim then follows from Lemma \ref{preprocessingLemma}. Let $H$ be obtained from a minimum-size 2-edge-cover of $G$ via Lemma \ref{lem:canonical}. Taking the best solution among the ones guaranteed by Lemmas \ref{lem:manyTriangles:refined} and \ref{lem:fewTriangles:main} w.r.t $H$, with the usual notation for $t$ and $b$ one obtains an approximation factor
\begin{align*}
apx(t,b):=min\{apx_{mny}(t,b),apx_{few}(t,b)\},\\
apx_{mny}(t,b):=\frac{14}{9}-\frac{8}{27}t+\frac{4}{9}b, \quad apx_{few}:=\frac{13}{10}+\frac{1}{30}t-\frac{1}{20}b.
\end{align*}
For a fixed value of $b$, $apx_{few}(t,b)$ is an increasing function of $t$, while $apx_{mny}(t,b)$ is a decreasing function of $t$. The two factors are equal for $t=t(b):=\frac{69}{89}+\frac{3}{2}b$. Notice that $t(b)\geq 0$. For the values of $b$ such that $t(b)\leq 1-b$, namely for $b\leq \frac{8}{89}$, the worst case is achieved for $t=t(b)$, leading to the approximation factor
$$
apx(b):=\frac{13}{10}+\frac{1}{30}(\frac{69}{89}+\frac{3}{2}b)-\frac{1}{20}b=\frac{118}{89}.
$$
For $b\geq \frac{8}{89}$, the worst case is achieved for $t=1-b$ and it is given by $apx_{few}(t,b)$, leading to the approximation factor
$$
apx(b):=apx_{few}(1-b,b)=\frac{13}{10}+\frac{1}{30}(1-b)-\frac{1}{20}b=\frac{4}{3}-\frac{1}{12}b\leq \frac{118}{89}. 
$$
\end{proof}

It remains to prove Lemma \ref{lem:manyTriangles:refined}. The initial part of the construction, leading to a core-triangle 2-edge-cover, is the same as in Lemma \ref{lem:manyTriangles:main}. In particular, we will obtain a 2-edge-cover $S$ consisting of a core $C$ (which is 2EC) and $k$ triangles $T_1,\ldots,T_k$, such that nodes of distinct triangles are not adjacent in $G$. Recall that at this point in Lemma \ref{lem:manyTriangles:main} we consider a solution $APX$ obtained from $S$ by adding any two edges between two distinct nodes $v_i,u_i\in V(T_i)$ and $C$, and removing the edge $v_iu_i$: let $Q_i$ be the (four) edges with at least one endpoint in $V(T_i)$. Our final solution is $APX=E(C)\cup \bigcup_{i=1}^{k}Q_i$. We recall that
$$
|APX|=|E(C)|+4k=|S|+k\leq (\frac{3}{2}-\frac{1}{3}t+\frac{1}{2}b)|H|+\frac{k}{2}.
$$
The main idea here is to choose the $Q_i$'s more carefully among the possible options. To do that, we use an approach very similar in spirit to the one in \cite{SV14}. 

In more detail, let $Q=\cup_{i=1}^{k}Q_i$. For a subset of edges $F$, let $\alpha(F)$ be the number of connected components of the graph $(V,F)$. Define $\alpha(S)$ to be the minimum possible value of $\alpha(Q)$ for any valid choice of $Q$ for a given core-triangle 2-edge-cover $S$. The following fact follows similarly to the reduction to matroid intersection in \cite{SV14} for the computation of a maximum earmuff:
\begin{fact}
Given a core-triangle 2-edge-cover $S$ of $G$, there exists a polynomial time algorithm that computes $Q^*$ minimizing $\alpha(Q)$ over the set of edges $Q$ of the above type, i.e. $\alpha(Q^*)=\alpha(S)$.
\end{fact}
\begin{proof}
Consider the following algorithm. For each $Q_i$ that contains a 2-matching $u_1v_1,u_2v_2$ between $T_i$ and $C$ (with $v_i\in V(C)$) we define a pseudo-edge $v_1v_2$ of color $i$. We next compute a maximum size set $F$ of pseudo-edges such that $(V(C),F)$ is a forest and no two edges of $F$ have the same color. This can be done in polynomial time as follows. Let $M_1$ be the graphic matroid induced by the pseudo-edges $E'$, and $M_2$ be the partition matroid induced by the colors of the pseudo-edges $E'$. Then $F$ is a maximum cardinality independent set in the intersection of $M_1$ and $M_2$, which can be computed in polynomial time (see, e.g., \cite{E01,S03}).

Given $F$, we build a feasible $Q$ as follows. For each pseudo-edge $f\in F$ of color $i$, we add the corresponding $Q_i$ to $Q$. Notice that at this point $\alpha(Q)$ is precisely the number of connected components induced by $F$, namely $|V(C)|-|F|$. Then, for each missing color $j$ in $F$, we add an arbitrary $Q_j$ to $Q$. Notice that this does not reduce $\alpha(Q)$ by the optimality of $F$. In particular, $\alpha(Q)=|V(C)|-|F|$.

It remains to argue that the above $Q$ minimizes $\alpha(Q)$. To see that, consider the optimal $Q^*=\cup_{i=1}^{k}Q^*_i$. We next iteratively remove from $Q^*$ any $Q^*_j$ whose removal does not increase $\alpha(Q^*)$ (in particular this always happens when $Q^*_j$ is incident to a unique node of $V(C)$). Let $F^*$ be the pseudo-edges corresponding to the remaining $Q^*_i$'s in $Q^*$. The number of connected components $|V(C)|-|F^*|$ induced by $F^*$ is precisely $\alpha(Q^*)$. Furthermore, $F^*$ is in the intersection of $M_1$ and $M_2$, hence $|F^*|\leq |F|$ (implying $\alpha(Q)\leq \alpha(Q^*)$).
\end{proof}

We next assume that $Q$ is obtained as in the above fact in the computation of $\APX$, a define the $Q_i$'s accordingly. Analogously to \cite{SV14}, we next define an alternative solution $\APX'$ of size $\apx'=|\APX'|$ starting from $S$ as follows. We start from $Q$, and add a minimal number of edges $F\subseteq G[V(C)]$ so that $Q\cup F$ becomes a connected spanning subgraph. Obviously $|F|=\alpha(S)-1$. Then we compute a minimum-size $T$-join\footnote{Recall that, for a subset of nodes $T$ of even cardinality, a $T$-join is a subgraph such that the nodes in $T$ have odd degree and the other nodes even degree.} $J$ over the odd degree nodes $T$ of $Q\cup F$. Notice that $T\subseteq V(C)$ since the nodes of the triangles have even degree in $Q$ and this degree is not modified by $F$. 
We initially set $\APX':=Q\cup F\cup J$. We observe that $\APX'$ is a connected Eulerian graph, hence in particular it is 2EC. Notice however that $\APX'$ might contain parallel edges: for any two parallel edges $f$ and $g$ (which are effectively copies of the same edge $e$), we remove $g$ and add some other edge $g'$ connecting the two sides of the cut induced by $f$ in $\APX'$. Notice that such $g'$ must exist since $G$ is 2EC. Furthermore, this process keeps $\APX'$ 2EC and does not change its size. By repeating this process we turn $\APX'$ into a simple 2EC spanning subgraph of size $\apx'=4k+\alpha(S)-1+|J|$.

We will next show that $|J|\leq \frac{1}{2}|E(C)|$. To this aim we will exploit a classical result by Frank \cite{F93}. We recall that an ear decomposition of a graph is obtained incrementally as follows. The first ear $E_1$ is a cycle. Given $F_i=E_1\cup \ldots \cup E_i$, the ear $E_{i+1}$ is a path with its endpoints in $V(F_i)$ and the other nodes not in $V(F_i)$, or alternatively a cycle with exactly one node in $V(F_i)$. An ear is even (resp., odd) if its number of edges is so. A well-known fact is that every 2EC graph $G$ admits an ear decomposition spanning all its edges. By $\phi(G)$ we denote the minimum possible number of even ears in any ear decomposition of $G$.   
\begin{lemma}\label{Join} \cite{F93}
Let $G$ be a 2EC, and $T\subseteq V(G)$, $|T|$ even. Then there exists a $T$-join of $G$ of size at most $\frac{L_{\phi(G)}}{2}$, where $L_{\phi(G)}=|V(G)|+\phi(G)-1$. Furthermore $\opt(G)\geq L_{\phi(G)}$.
\end{lemma}
Recall that the core $C$ is 2EC and that, for the $T$-join $J$, $T\subseteq V(C)$. Therefore we can apply Lemma~\ref{Join} to infer that $|E(C)|\geq \opt(G[V(C)])\geq L_{\phi(G[V(C)])}\geq 2|J|$ as desired. As a consequence 
$$
\apx'=|APX'|=|Q|+|F|+|J|=4k+\alpha(S)-1+|J|\leq 4k+\alpha(S)-1+\frac{1}{2}|E(C)|.
$$

The next step is to define a refined lower bound on $\opt$ which takes $\alpha(S)$ into account. 
\begin{lemma}\label{minimumComponents}
Let $F^*$ be a 2EC spanning subgraph of $G$ including $E(C)$, and $Q^*$ be the edges of $F^*$ with at least one endpoint not in $C$. Then $\alpha(Q^*)\geq 4k-|Q^*|+\alpha(S)$.
\end{lemma}
\begin{proof}
We prove the claim by induction on $|Q^*|$. As already observed, $F^*$ must include at least $4$ edges with at least one endpoint in each $T_i$, hence
$|Q^*|\geq 4k$. The base case is $|Q^*|=4k$. In this case there are precisely $4$ edges incident to the nodes $V(T_i)=\{u_i,v_i,w_i\}$ of each $T_i$. These edges must induce either a path of type $a_i,u_i,v_i,w_i,b_i$ or a cycle of type $a_i,u_i,v_i,w_i,a_i$, where $a_i,b_i\in V(C)$. In particular $Q^*$ is a valid choice for $Q$, hence $\alpha(Q^*)\geq \alpha(S)=4k-|Q^*|+\alpha(S)$.

In the inductive case $|Q^*|=h>4k$ there must exist a triangle $T_i$ such that at least $5$ edges of $Q^*$ are incident to $T_i$. We distinguish $4$ cases depending on the number $t_i\in \{0,1,2,3\}$ of these edges with both endpoints in $T_i$.

\smallskip \paragraph{($\mathbf{t_i=3}$)} Here we distinguish two subcases. Suppose first that there exist two nodes in $T_i$, say $u_i$ and $v_i$, with degree at least $3$ in $Q^*$. Let $Q'=Q^*\setminus \{u_iv_i\}$. Notice that $F':=E(C)\cup Q'$ is a feasible 2EC spanning subgraph and $\alpha(Q')=\alpha(Q^*)$. By inductive hypothesis $\alpha(Q')\geq 4k-|Q^*|+\alpha(S)+1$, hence the claim. 
        
 It remains to consider the case that there exists exactly one node of $T_i$, say $u_i$, with degree at least $3$ in $Q^*$. Since there must be at least $2$ edges in $Q^*$ with exactly one endpoint in $T_i$, this implies $\deg_{Q^*}(u_i) \geq 4$. Let $a_i$ be a node in $C$ that is adjacent to $u_i$ in $Q^*$. Notice that, since $u_i$ is not a cut node, there must exists some edge from some other node in $T_i$ to $C$, say $v_ib_i$. Let $Q'=Q^*\cup \{v_ib_i\}\setminus \{u_ia_i,u_iv_i\}$. Notice that $F'=E(C)\cup Q'$ is a feasible 2EC spanning subgraph, $|Q'|=|Q^*|-1$, and $\alpha(Q')\leq \alpha(Q^*)+1$.
 Therefore, by the inductive hypothesis on $Q'$, $\alpha(Q^*)\geq \alpha(Q')-1\geq 4k-|Q^*|+\alpha(S)$.

\smallskip \paragraph{($\mathbf{t_i=2}$)} The $2$ edges in $E(T_i)\cap Q^*$ form a path $P$ of length $2$, say $v_i,u_i,w_i$. There must be a node $v$ in $T_i$ such that $\deg_{Q^*}(v)\geq 3$. Indeed there are at least $3$ edges in $Q^*$ incident to nodes of $T_i$ and to $C$. If one of them is incident to $u_i$, say $a_iu_i$, then consider $Q'=Q^*\setminus \{a_iu_i\}$. Again $F'=E(C)\cup Q'$ is a feasible 2EC spanning subgraph and the claim follows similarly to previous cases. Otherwise w.l.o.g assume there are two edges in $Q^*\setminus E(T_i)$ incident to $v_i$. Let $Q'=Q^*\setminus \{e\}$, where $e$ is an edge incident to $v_i$ not in $P$. As usual $F'=E(C)\cup Q'$ is a feasible 2EC spanning subgraph and $\alpha(Q')\leq \alpha(Q^*)-1$. Thus, by the inductive hypothesis on $Q'$, $\alpha(Q^*)\geq \alpha(Q')-1\geq 4k-|Q^*|+\alpha(S)$.


\smallskip \paragraph{($\mathbf{t_i=1}$)} Assume w.l.o.g. that $v_iw_i \in Q^*$, hence $u_i$ is not adjacent to $v_i$ nor to $w_i$ in $Q^*$. We distinguish two sub-cases. If $v_i$ does not belong to the same connected component of $u_i$ in $(V,Q^*)$, then we let $Q'=Q^*\cup \{u_iv_i\}\setminus \{e'\}$ where $e'$ is any other edge $e'$ incident to $u_i$. Otherwise, we define $Q'$ similarly but taking care of selecting an $e'\neq u_iv_i$ incident to $u_i$ along a path from $u_i$ to $v_i$ in $(V,Q^*)$. In both cases we have $|Q'|=|Q^*|$, $\alpha(Q')=\alpha(Q^*)$, and $F'=E(C)\cup Q'$ is a feasible 2EC spanning subgraph. Now $Q'$ contains $2$ edges of $E(T_i)$, hence the claim follows from the previous analysis.

\smallskip \paragraph{($\mathbf{t_i=0}$)} We use the same construction as in the case $t_i=1$, considering again nodes $u_i$ and $v_i$. In particular, by adding the edge $u_iv_i$ we reduce to the case $t_i=1$, which is addressed above.
\end{proof}

Recall that we used the lower bound $\opt(G)\geq 4k$. We can next use the following stronger lower bound. 
\begin{corollary}\label{cor:minimumComponents}
Let $S$ be a core-triangle 2-edge-cover of a structured graph $G$ with $k$ triangles. Then 
$\opt(G)\geq 4k+\alpha(S)-1$.
\end{corollary}
\begin{proof}
Consider the edges $Q^*$ of $\OPT(G)$ with at least one endpoint not in the core $C$ of $S$. By Lemma~\ref{minimumComponents}, $(V(G),Q^*)$ has at least $\alpha(Q^*)\geq 4k-|Q^*|+\alpha(S)$ connected components. Therefore $\OPT(G)$, which is connected, must contain at least $4k-|Q^*|+\alpha(S)-1$ edges, besides $Q^*$ to connect the connected components of $(V(G),Q^*)$. The claim follows.
\end{proof}

We are now ready to prove Lemma \ref{lem:manyTriangles:refined}.
\begin{proof}[Proof of Lemma \ref{lem:manyTriangles:refined}]
Recall that we computed two feasible solutions $\APX$ and $\APX'$ satisfying
$$
|APX|= |E(C)|+4k \leq (\frac{3}{2}-\frac{1}{3}t+\frac{1}{2}b)|H|+\frac{k}{2}
\quad \text{and} \quad
|APX'|\leq 4k+\alpha(S)-1+\frac{1}{2}|E(C)|.
$$
By Corollary \ref{cor:minimumComponents} we have that $\opt=\opt(G)\geq 4k+\alpha(S)-1$. We will also use $\opt\geq |H|$ as usual. Notice that 
$$
|S|=3k+|E(C)|\leq (\frac{3}{2}-\frac{1}{3}t+\frac{1}{2}b)|H|-\frac{k}{2}=:f(t,b)|H|-\frac{k}{2},
$$
hence
$$
\opt\geq |H|\geq \frac{\frac{7}{2}k+|E(C)|}{f(t,b)}
$$
Observe that $\frac{7}{6}\leq f(t,b)\leq 2$. Recall that $|APX_B| =|S|+k=|C|+4k$. The overall approximation factor is upper bounded by
$$
apx(t,b):=\frac{\min\{|APX|,|APX'|\}}{\max\{|H|,4k+\alpha(S)-1\}}\leq \frac{4k+\min\{|E(C)|,\alpha(S)-1+\frac{1}{2}|E(C)|\}}{\max\{\frac{\frac{7}{2}k+|E(C)|}{f(t,b)},4k+\alpha(S)-1\}}.
$$

We distinguish $4$ cases depending on which is the smallest upper bound on the cost of the approximate solution and the largest lower bound on the optimal cost:

\medskip\noindent{\bf (1) $\alpha(S)-1\geq \frac{1}{2}|E(C)|$ and $4k+\alpha(S)-1\geq \frac{\frac{7}{2}k+|E(C)|}{f(t,b)}$.} In this case:
$$
apx(t,b)\leq \frac{4k+|E(C)|}{4k+\alpha(S)-1}\leq \frac{4k+|E(C)|}{\max\{4k+\frac{1}{2}|E(C)|,\frac{\frac{7}{2}k+|E(C)|}{f(t,b)}\}}.
$$
The righthand side of the above inequality is an increasing function of $|E(C)|$ for $4k+\frac{1}{2}|E(C)|\geq \frac{\frac{7}{2}k+|E(C)|}{f(t,b)}$, i.e. for $|E(C)|\leq \frac{8f(t,b)-7}{2-f(t,b)}k$, and a decreasing function of $|E(C)|$ in the complementary case. Thus the worst case is achieved for $|E(C)|=\frac{8f(t,b)-7}{2-f(t,b)}k$, leading to
$$
apx(t,b)\leq \frac{4k+\frac{8f(t,b)-7}{2-f(t,b)}k}{4k+\frac{1}{2}\frac{8f(t,b)-7}{2-f(t,b)}k}=\frac{2+8f(t,b)}{9}.
$$

\medskip\noindent{\bf (2) $\alpha(S)-1\geq \frac{1}{2}|E(C)|$ and $4k+\alpha(S)-1< \frac{\frac{7}{2}k+|E(C)|}{f(t,b)}$.} This imposes $\frac{\frac{7}{2}k+|E(C)|}{f(t,b)}-4k>\frac{1}{2}|E(C)|$, i.e. $|E(C)|> \frac{8f(t,b)-7}{2-f(t,b)}k$. In this case 
$$
apx(t,b)\leq \frac{4k+|E(C)|}{\frac{\frac{7}{2}k+|E(C)|}{f(t,b)}}.
$$
The right hand side of the above inequality is a growing function of $|E(C)|$ due to the fact that $f(t,b)\geq \frac{7}{6}\geq \frac{8}{7}$. Hence the worst case is achieved for $|E(C)|$ going to infinity, leading to 
$$
apx(t,b)\leq f(t,b).
$$
Notice that $\frac{2+8f(t,b)}{9}\geq f(t,b)$ iff $2\geq f(t,b)$. Hence case (1) dominates case (2) in terms of \amEdit{the} worst-case approximation factor.

\medskip\noindent{\bf (3) $\alpha(S)-1< \frac{1}{2}|E(C)|$ and $4k+\alpha(S)-1\geq \frac{\frac{7}{2}k+|E(C)|}{f(t,b)}$.} This imposes $\frac{\frac{7}{2}k+|E(C)|}{f(t,b)}-4k<\frac{1}{2}|E(C)|$, i.e. $\frac{2-f(t,b)}{8f(t,b)-7}|E(C)|< k$. In this case one has
\begin{equation}\label{eqn:case3}
apx(t,b)\leq \frac{\alpha(S)-1+\frac{1}{2}|E(C)|}{\alpha(S)-1+4k}\leq \frac{\alpha(S)-1+\frac{1}{2}|E(C)|}{\alpha(S)-1+4\frac{2-f(t,b)}{8f(t,b)-7}|E(C)|}. 
\end{equation}
For $\frac{1/2}{4\frac{2-f(t,b)}{8f(t,b)-7}}=\frac{8f(t,b)-7}{8(2-f(t,b))}\geq 1$, i.e. for $f(t,b)\geq \frac{23}{16}$, the righthand side of \eqref{eqn:case3} is decreasing in $\alpha(S)-1$. Hence the worst case is achieved for $\alpha(S)-1=\frac{\frac{7}{2}k+|E(C)|}{f(t,b)}-4k$ leading to 
$$
apx(t,b)\leq \frac{\frac{\frac{7}{2}k+|E(C)|}{f(t,b)}-4k+\frac{1}{2}|E(C)|}{\frac{\frac{7}{2}k+|E(C)|}{f(t,b)}-4k+4\frac{2-f(t,b)}{8f(t,b)-7}|E(C)|}.
$$
For $f(t,b)\geq \frac{23}{16}$ the righthand side of the above inequality if increasing in $|E(C)|$, hence that quantity is upper bounded by setting $|E(C)|=\frac{8f(t,b)-7}{2-f(t,b)}k$, leading to
\begin{align*}
apx(t,b) \leq \frac{\frac{7k}{2f(t,b)}-4k+\frac{2+f(t,b)}{2f(t,b)}\frac{8f(t,b)-7}{2-f(t,b)}k}{\frac{7k}{2f(t,b)}+\frac{1}{f(t,b)}\frac{8f(t,b)-7}{2-f(t,b)}k} = \frac{16f(t,b)-14}{9}.
\end{align*}

For $f(t,b)<\frac{23}{16}$, the righthand side of \eqref{eqn:case3} is increasing in $\alpha(S)-1$. Hence we can upper bound it by setting $\alpha(S)-1=\frac{1}{2}|E(C)|$, leading to 
$$
apx(t,b) \leq \frac{|E(C)|}{\frac{1}{2}|E(C)|+4\frac{2-f(t,b)}{8f(t,b)-7}|E(C)|}
= \frac{2(8f(t,b)-7)}{8f(t,b)-7+8(2-f(t,b))} 
= \frac{16f(t,b)-14}{9}
$$
We observe that $\frac{16f(t,b)-14}{9}\leq \frac{2+8f(t,b)}{9}$ for $f(t,b)\leq 2$. Hence case (3) is dominated by case (1).

\medskip\noindent{\bf (4) $\alpha(S)-1< \frac{1}{2}|E(C)|$ and $4k+\alpha(S)-1< \frac{\frac{7}{2}k+|E(C)|}{f(t,b)}$.} In this case one has
$$
apx(t,b)\leq \frac{\alpha(S)-1+\frac{1}{2}|E(C)|}{\frac{\frac{7}{2}k+|E(C)|}{f(t,b)}}\leq  \frac{\min\{\frac{1}{2}|E(C)|,\frac{\frac{7}{2}k+|E(C)|}{f(t,b)}-4k\}+\frac{1}{2}|E(C)|}{\frac{\frac{7}{2}k+|E(C)|}{f(t,b)}}.
$$
For $\frac{1}{2}|E(C)|\leq \frac{\frac{7}{2}k+|E(C)|}{f(t,b)}-4k$, i.e. for $|E(C)|\geq k\frac{8f(t,b)-7}{2-f(t,b)}$, the righthand side of the above inequality is $\frac{f(t,b)|E(C)|}{\frac{7}{2}k+|E(C)|}$, which is a growing function of $|E(C)|$. Hence the worst-case value is $f(t,b)$. Like in case (2), this case is dominated by case (1). Otherwise, i.e. for $k>\frac{2-f(t,b)}{8f(t,b)-7}|E(C)|$, the righthand side of the above inequality is
$
1+f(t,b)\frac{\frac{1}{2}|E(C)|-4k}{|E(C)|+\frac{7}{2}k}.
$
This is a decreasing function of $k$, hence it is upper bounded by setting $k=\frac{2-f(t,b)}{8f(t,b)-7}|E(C)|$. This leads to 
$$
apx(t,b)\leq  1+f(t,b)\frac{\frac{1}{2}|E(C)|-4\frac{2-f(t,b)}{8f(t,b)-7}|E(C)|}{|E(C)|+\frac{7}{2}\frac{2-f(t,b)}{8f(t,b)-7}|E(C)|}=\frac{16f(t,b)-14}{9}.
$$
This matches the upper bound achieved in case (3), hence in particular \amEdit{it} is dominated by case (1). Altogether
$$
apx(t,b) \leq \frac{2+8f(t,b)}{9}
= \frac{1}{9}(2+8(\frac{3}{2}-\frac{1}{3}t+\frac{1}{2}b))
= \frac{14}{9}-\frac{8}{27}t+\frac{4}{9}b.
$$
\end{proof}

\section{Some Previous Attempts to Breach the $4/3$ Approximation Barrier}
\label{sec:failedAttempts}

To the best of our knowledge, there are only two (conference) papers which claim a better than $\frac{4}{3}$ approximation in the literature\footnote{More recently, Civril made several failed attempts to improve on $4/3$. However the last version of his work on ArXiV \cite{C22} only claims an alternative $4/3$ approximation for the general case.}: Krysta and Kumar \cite{KK01} claim a $\frac{4}{3}-\frac{1}{1344}$ approximation and Jothi, Raghavachari, and Varadarajan \cite{JRV03} a $\frac{5}{4}$ approximation. Both papers (after about two decades) have no journal version nor any public full version to the best of our knowledge. We remark that the two papers presenting a $4/3$ approximation algorithm that we mentioned earlier, namely  \cite{HVV19,SV14}, do not include \cite{JRV03,KK01} in their references, despite the fact that \cite{HVV19,SV14} were published several years later and claim a worse approximation factor w.r.t. \cite{JRV03,KK01}. In particular, \cite{SV14} states: ``Better approximation ratios have been claimed, but to the best of our knowledge, no correct proof has been published''.
We next give more details about the gaps in \cite{JRV03,KK01}.       

\cite{KK01} gives a rather vague description of the algorithm and of its analysis (only about 4 pages altogether), with no formal proof. We are not even close to being able to fill in the gaps in \cite{KK01}. The most critical part in our opinion is Section 3.3 ``2-Connecting Paths''. Here even the description of the algorithm is not fully detailed, as mentioned by the authors: ``The above cases have been stated in a highly oversimplified manner''. The subsequent case analysis is only sketched at a very high level. In our experience, most of the difficulty in case-analysis-based algorithms (like \cite{HVV19,JRV03,KK01} and this paper) is due to handling very specific extreme cases. Hence  it is very hard for the reader to verify the correctness without seeing the full details. 
Let us mention that \cite{KK01} heavily builds upon the algorithm and analysis in \cite{VV00}, which is the conference version of \cite{HVV19}. \cite{VV00} contains some gaps (as declared by a superset of the authors in \cite{HVV19}). This makes even more critical that the authors of \cite{KK01} publish a full and self-contained version of their paper. Finally, we remark that the approximation factor in \cite{KK01} is worse than the one claimed in this paper (hence here the only unclear point is which is the first paper to breach the $4/3$ approximation barrier).


\cite{JRV03} is much more detailed than \cite{KK01}, still major parts of their analysis are missing. The most important such gap (but not the only one), in our opinion, is in Section 4.1 ``Finding a 5-cycle''. The authors give a full proof for the case that the leaf vertex $v_1$ corresponds to a single node of $G$ (Lemma 4.1). The authors claim (Lemma 4.2) that they can extend that lemma to the case where $v_1$ is a ``2-connected component that has been shrunk into a single node'', but no proof is provided. The authors mention: ``The proof is more complicated \ldots There are many different cases to consider, and we omit the proof here due to lack of space''. Proving Lemma 4.2, if possible, is in our opinion highly non-trivial and we are not able to reconstruct their proof on our own. As mentioned before, missing even one case might lead to a too optimistic (hence wrong) approximation factor. So it is particularly important that the community has the chance to verify this type of case \amEdit{analysis}.


\end{document}